\documentclass[12pt]{article}
\usepackage[T1]{fontenc}
\usepackage[utf8]{inputenc}
\usepackage{lmodern}
\usepackage{260916_jmatNS_v2_assets/article_style} % Seongoh's style
\usepackage{pdflscape} % To adjust the table width
\usepackage{hyperref} % For table 
\usepackage{graphicx}
\usepackage{booktabs}
\usepackage{subcaption}
\usepackage{longtable}
\usepackage{rotating}
\usepackage{array}
\usepackage{microtype}

\title{Matrix Graphical Model Via Joint Estimation of Partial Correlations}

\author[1]{Hyewon Kim}

\author[2,3]{Seongoh Park\footnote{To whom all correspondence should be addressed. Email: \texttt{spark6@sungshin.ac.kr}}}
    
\affil[1]{Department of Statistics, Sungshin Women's University, Seoul, Korea} 
\affil[2]{School of Mathematics, Statistics and Data Science, Sungshin Women's University, Seoul, Korea} 
\affil[3]{Center for Data Science, Sungshin Women's University, Seoul, Korea}

\date{September 16, 2026}

\newcommand{\spacingset}[1]{\renewcommand{\baselinestretch}{#1}\small\normalsize}
\hypersetup{
  pdftitle={Matrix Graphical Model Via Joint Estimation of Partial Correlations},
  pdfauthor={Hyewon Kim; Seongoh Park},
  pdfkeywords={Matrix-variate data; Gaussian graphical model; Kronecker product; undirected graph},
  colorlinks=true, citecolor=blue, urlcolor=blue, linkcolor=blue
}

\begin{document}
\maketitle
\baselineskip 24pt

\begin{abstract}
\noindent
Matrix graphical models aim to characterize conditional dependence structures in matrix-variate data under a separable covariance assumption. In this framework, the precision matrix is decomposed as a Kronecker product, enabling separate modeling of undirected graphs across row and column domains.
Existing methods have been developed for this problem, including likelihood-based approaches and regression-based procedures for graph estimation. Likelihood-based methods estimate precision matrices directly and recover graph structures indirectly, whereas regression-based approaches directly target estimating edges among variables, thus outperforming the former. However, existing regression-based methods are based on multiple penalized regression problems, which naturally yields asymmetry in estimated graphs and computational difficulty in selecting tuning parameters.
To address the limitations, we propose a joint estimation of partial correlations in matrix graphical models. The proposed method estimates all partial correlations simultaneously within a unified optimization framework, thereby preserving symmetry and easing the pain of selecting the best models.
Numerical studies demonstrate that the proposed method improves graph recovery performance compared to existing approaches. We also analyze protein expression data collected from patients with pulmonary tuberculosis, measured repeatedly at multiple time points, where the proposed method compares protein networks between two groups of patients and recovers the temporal dependence structure.

\vskip0.5cm
\noindent {\bf Keywords:} Matrix-variate data; Gaussian graphical model; Kronecker product; undirected graph.
\end{abstract}
\baselineskip 18pt

\section{Introduction}\label{sec:intro}

In the modern era, high-dimensionality is not the only challenge to address in data analysis, but additional complex structures have complicated 
conventional methods. 
For example, matrix or more generally tensor structure in data has been observed in a variety of fields.
In the bioinformatics area,  \citet{Greenewald:2015} analyzed repeatedly measured gene expression data of the yeast (\textit{S. cerevisiae}) where gene probes (rows) were regularly sampled for every 24 minutes (columns). In \citet{Yin:2012}, graphical models were fitted to multi-tissue (rows) gene expression (columns) data.
Another example can be found in neuroscience. \citet{Oh:2026}  aimed to classify brain MRI images related to Alzheimer's disease symptom progression using nonparametric linear discriminant analysis. \citet{Zhou:2014}, \citet{Shin:2025}, and \citet{Xia:2017} proposed methods to estimate brain connectivity using electroencephalography (EEG) data that can be seen as multivariate time-series data.
Moreover, \citet{Yu:2022} examined macroeconomic indices (rows) for multiple nations (columns), which were collected from OECD database.
On the other hand, \citet{Greenewald:2019} viewed temporal-spatial data obtained from the US National Center for Environmental Prediction as 3-dimensional tensor (latitude versus longitude versus time) and developed a generalized graphical lasso.

Though this paper mainly focuses on graphical models, diverse statistical methodologies have been developed to analyze matrix-variate data; for example, multi-task learning \citep{Zhang:2010}, joint mean covariance estimation \citep{Hornstein:2019}, covariance estimation \citep{Zhang:2023}, factor analysis \citep{Wang:2019,Yu:2022}, differential network model \citep{Ji:2021}, and so on.
We refer interested readers to those cited papers for further understanding.

\subsection{Related work}
For a clear comparison with the proposed method, we review existing methods developed under the assumption of covariance separability. That is, for matrix-variate observations of dimension $p\times q$, the covariance matrix of the vectorized data is assumed to satisfy $\Sigma_{pq \times pq} = V_{q \times q} \otimes U_{p \times p}$ where $\otimes$ is the Kronecker product. For other structural assumptions such as Kronecker sum, readers may refer to \citet{Kalaitzis:2013,Greenewald:2019}.
The existing literature can be divided into two branches:
(1) estimation of precision matrices, and (2) estimation of graph support (or edges).
Methods in the first branch are typically based on likelihoods where the graph structure is estimated as a by-product of the procedure.
In contrast, methods in the second directly aim to estimate edges of the graph, mostly based on the idea of node-wise regression. 

To describe the likelihood-based methods, we define a function that returns the solution of the graphical lasso method \citep{Friedman:2008}\footnote{Note that this function returns a matrix of the same size as an input matrix. For notational convenience, we allow this slight abuse of notation.}:
    for $S \succeq 0$, 
    $$
    L_{GL}(S;\lambda) = \arg\min_{\Omega \succ 0} - \log \det(\Omega) + \text{tr}( \Omega S) + \lambda || \Omega ||_{1,\text{off}},
    $$
    where $\lambda > 0$ is a regularization parameter and $|| \cdot ||_{1,\text{off}}$ is an absolute sum of off-diagonal elements.
The initial work is the likelihood flip-flop method (\citet{Leng:2012}, \citet{Yin:2012}, \citet{Tsiligkaridis:2012,Tsiligkaridis:2013}). This method alternates updating one of $U^{-1}$ and $V^{-1}$, while fixing another as follows:
$$
\Omega^U_{(t+1)} = L_{GL}\big(\sum_{h=1}^n  Y^{(h)} \Omega^V_{(t)} (Y^{(h)})^\top / (nq);\lambda \big).
$$
where $\Omega^V_{(t)}$ is the estimate of $V^{-1}$ at iteration $t$ and $\Omega^U_{(t+1)}$ is that of $U^{-1}$ at iteration $t+1$.
The final estimator is 
$\Omega^U_{FF}= \lim_{t \to \infty} \Omega^U_{(t+1)}$ and $\Omega^V_{FF}= \lim_{t \to \infty} \Omega^V_{(t+1)}$.
To circumvent the alternating steps, \citet{Zhou:2014} proposed a non-iterative flip-flop method named ``GEMINI'', which defines the precision matrix estimators by             $\Omega^U_{NFF} = L_{GL}\big(\sum_{h=1}^n Y^{(h)} (Y^{(h)})^\top / (nq);\lambda \big)$ and 
$\Omega^V_{NFF} = L_{GL}\big(\sum_{h=1}^n (Y^{(h)})^\top Y^{(h)} / (np);\lambda \big)$. 
Moreover, \citet{Ning:2013} proposed more robust graph estimators based on matrix non-paranormal distribution. In this work, they used rank-based correlation matrices such as Spearman's rho or Kendall's tau instead of the usual sample covariance matrix and solved the followings:
$\Omega^U_{Rank} = L_{GL}\big(R^U;\lambda \big)$, where
$R^U$ is a rank-based correlation matrix modified to satisfy positive-definiteness. $\Omega^V_{Rank}$ is similarly defined.
Beyond the graphical lasso algorithm, \citet{Wu:2023} applied the sparse column-wise inverse operator (SCIO) that aims to induce column-wise sparsity in the estimates of $U^{-1}$ and $V^{-1}$. This method estimates the precision matrix in a column-wise manner, which yields a final asymmetric estimator by nature.

In literature of the other branch, the graph edges are directly estimated using the partial correlations. \cite{Chen:2019} and \cite{Xia:2017} described the graph estimation problem as a multiple testing problem. In estimating the partial covariances and variances, they obtained regression coefficient estimators by node-wise regression. Let $Y_{j}^{(h)} \in \mathbb{R}^q$, $h=1,\ldots, n$, be the $j$-th row-vector of the $h$-th matrix-variate observation. Then, the $j$-th node-wise regression is as follows:
\begin{equation}\label{eq:nodewise}
\hat{\theta}^j = \arg\min_{\theta \in \mathbb{R}^{p},\, \theta_j = 0} \left\{ \frac{1}{2nq}\sum_{h=1}^{n} \Vert Y_{j}^{(h)} - (Y^{(h)})^\top \theta \Vert _2^2+ \lambda \sum_{k=1}^p  (\hat{\psi}_{kk})^{\alpha} |\theta_k| \right\}, \quad j \in [p]
\end{equation}
where $\hat{\psi}_{kk}$ is a scaling factor defined by the $k$-th diagonal component of 
$\sum_{h=1}^n Y^{(h)} (Y^{(h)})
^\top / (n q)$ and $\alpha >0$ is a pre-determined constant.     
As a result, they need to solve $p+q$ separate regularized regression problems. A more recent work from \citet{Shin:2025} also used a similar regression-based approach to estimate matrix graphs, where they derived  theoretical results about (global) support recovery.

\subsection{Contributions}
The main contributions of this work are threefold. First, the proposed joint estimation preserves the symmetric structure of partial correlations while remaining within the regression framework to estimate graph edges. In contrast, the node-wise regression approaches yield asymmetric estimates (i.e. $(\hat{\theta}^j)_i \neq (\hat{\theta}^i)_j$).  
To avoid this issue, \citet{Shin:2025} adopted ``AND'' or ``OR'' rule (\citet{m&b2006}) to symmetrize the estimated graph. In \citet{Xia:2017}, 
the proposed variance estimator of the test statistic is asymmetric, but it is not clearly stated in the paper how to address the asymmetry. Though the variance estimator proposed in \citet{Chen:2019} did not suffer this problem, it requires a consistent covariance estimator and thus calls for an additional computational burden. Moreover, its theoretical validity relies on a weak sparsity condition on the population covariance matrix, which may not hold in general.

Second, due to the joint estimation, the proposed method only has a single regularization parameter to control the magnitude of partial correlations. As the correlations reside in $[-1, 1]$, the penalty term is relatively simple (see  \eqref{eq:loss}), and thus finding the optimal value of the regularization parameter is straightforward.
In contrast, each lasso regression may require regularization parameters on different scales, which complicates the selection of $\lambda$ and the scaling parameter $\alpha$ in \eqref{eq:nodewise}. For example, $\alpha=1/2$ is used in \citet{Xia:2017} and \citet{Chen:2019} with a slight modification, whereas $\alpha=0$ in \citet{Shin:2025} under the convention $a^0 \equiv 1$ for $a>0$.

Third, the proposed method empirically works better than the likelihood-based methods in estimating graph structure. 
It has also been verified in many previous studies \citep{Peng:2009,Shin:2025,Ravikumar:2011,Meinshausen:2008,Khare:2014} that the graphical lasso method shows worse finite-sample performance or may require a stronger condition to guarantee support recovery compared to node-wise regression methods. 

The rest of the paper is organized as follows. In Section~\ref{sec:meth}, we describe the proposed method and details of its implementation. In Section~\ref{sec:theory}, we present the consistency result of the proposed method. 
We evaluate its performance through simulation studies in Section~\ref{sec:simul}, while in Section~\ref{sec:real}, we demonstrate the practical utility of the proposed method using longitudinal proteomic profiles from pulmonary tuberculosis treatment data. Finally, we conclude the paper with a discussion in Section~\ref{sec:discus}.

\section{Methods}\label{sec:meth}

\subsection{Matrix graphical model and partial correlation}
Suppose $Y = (Y_{ij})_{i\in [p], j \in [q]}$ is a matrix-valued random variable where its vectorized version has mean zero and covariance $V_{q\times q} \otimes U_{p \times p}$. Here, the vectorization means binding columns of a matrix into a vector, and $\otimes$ is the Kronecker product of two matrices. 
First, we consider a node-wise regression for $j\in[q]$, that is, the least squares problem where $Y_{ij}$ is regressed on the rest $Y_{i,-j}$:
$$
\hat{b} = \arg\min_{b \in \mathbb{R}^{q-1}} \mathbb{E} || Y_{ij} - b^\top Y_{i,-j}||_2^2.
$$
The solution of the problem is 
$$
\hat{b} = \big[\mathbb{E}Y_{i,-j} Y_{i,-j}^\top \big]^{-1} \mathbb{E} Y_{i,-j} Y_{ij} = [u_{ii} V_{-j,-j}]^{-1} u_{ii} V_{-j,j} = (V_{-j,-j})^{-1}V_{-j,j}.
$$ 
Note that the coefficient does not depend on the row index $i$ due to the separable covariance structure.
Additionally, the partial correlation $\nu^V_{jk}$, defined by correlation between the best linear prediction error of $Y_{ij}$ and $Y_{ik}$, is known to be represented by 
$$
\nu^V_{jk}= - \omega^V_{jk} / \sqrt{\omega^V_{jj}\omega^V_{kk}},
$$
where we denote the inverse covariance matrix by $V^{-1}=(\omega^V_{jk})_{1\le j,k\le q}$.
As before, the partial correlation does not depend on the row index. 
Based on the block matrix inversion formula, we can derive $\hat{b}_{k} = \nu^V_{jk} \sqrt{\omega^V_{kk} / \omega^V_{jj}}$. Therefore, $\hat{b}_{k}=0$ is equivalent to $\nu^V_{jk}=0$, which moreover implies $\omega^V_{jk}=0$. Assuming the distribution of $Y$ satisfies the Markov property, $\omega^V_{jk}=0$ corresponds to the conditional independence between two columns, $Y_{\cdot j}$ and $Y_{\cdot k}$, given $Y_{\cdot, -\{j,k\}}$ (the rest except the $j,k$-th columns). Consequently, 
we can then recover the column-wise conditional independence graph $\mathcal{G}=\{(j_1, j_2)\in [q]\times [q]: Y_{\cdot j_1} \perp Y_{\cdot j_2} | Y_{\cdot, -\{j_1, j_2\}}\}$ by looking at zeros of the least square estimators.

\subsection{Estimation}
Suppose matrix data $Y^{(h)} \in \mathbb{R}^{p \times q}$ are independently observed, $h=1,\ldots, n$. 
Based on the aforementioned facts, we propose the least squares problem to estimate $\theta = (\nu^V_{jk})_{1 \le j < k \le q}$ and $w = (\omega^V_{jj})_{j \in [q]}$, with the $\ell_1$-regularization on $\theta$. The penalty term encourages partial correlations to be sparse and thus leads to sparse graph structure.
Then, the loss function to be minimized is given by:
\begin{equation}\label{eq:loss}
    \ell(\theta, w) = \dfrac{1}{2} \sum_{h=1}^n \sum_{i=1}^p \sum_{j=1}^q f_j^V \left(Y^{(h)}_{ij} -  \sum_{k\neq j} \sqrt{\frac{\omega^V_{kk}}{\omega^V_{jj}}} \nu^V_{jk} Y^{(h)}_{ik}\right)^2 + \lambda \sum_{j < k} |\nu^V_{jk}|.
\end{equation}
Here, $\{f_j^V\}_{j \in [q]}$ are non-negative weights. The choice of the weights will be discussed later.
This is a joint estimation of a precision matrix based on $q$ node-wise regressions.

To find the optimal solution of \eqref{eq:loss},  we use the alternating scheme to update $\theta$ and $w$, similar to \citet{Peng:2009}. Within each update, the coordinates are updated one at a time, also known as the shooting algorithm \citep{Fu:1998}.
For notational brevity, we omit the ``hat'' for estimators in the following updating formulas.
To update $\nu^V_{jk}$, suppose all the other parameters are fixed at their values from the previous iteration. Then, the next iterator of $\nu^V_{jk}$ has a closed form given by:
\begin{equation}\label{eq:update_nu}
    \nu^V_{jk} \leftarrow \dfrac{s_\lambda\bigg(\sum_{h,i} f_k^V\tilde{Y}_{ijk}^{(h)} \sqrt{\frac{\omega^V_{kk}}{\omega^V_{jj}}} Y_{ik}^{(h)} + 
    \sum_{h,i} f_j^V\tilde{Y}_{ikj}^{(h)} \sqrt{\frac{\omega^V_{jj}}{\omega^V_{kk}}} Y_{ij}^{(h)}\bigg)}{
    \sum_{h,i}\frac{\omega^V_{jj}}{\omega^V_{kk}} (Y_{ij}^{(h)})^2  + 
    \sum_{h,i}\frac{\omega^V_{kk}}{\omega^V_{jj}} (Y_{ik}^{(h)})^2}
\end{equation}

where $\tilde{Y}_{ijk}^{(h)} = Y^{(h)}_{ij} - \sum_{\ell\neq j,k} \sqrt{\omega^V_{\ell\ell}/\omega^V_{jj}} \cdot \nu^V_{j\ell} Y^{(h)}_{i\ell}$ and $\tilde{Y}_{ikj}^{(h)} = Y^{(h)}_{ik} - \sum_{\ell\neq j,k} \sqrt{\omega^V_{\ell\ell} / \omega^V_{kk}} \cdot \nu^V_{k\ell} Y^{(h)}_{i\ell}$. Here, $s_\lambda(x) = \text{sgn}(x) (|x|- \lambda)_+$ is the soft-thresholding function.
Next, the updating formula for the partial variance $\omega^V_{jj}$ is given:
\begin{equation}\label{eq:update_var}
\omega^V_{jj} \leftarrow \left\{\dfrac{1}{np}\sum_{h=1}^n \sum_{i=1}^p
\left(Y^{(h)}_{ij} -  \sum_{k\neq j} \sqrt{\frac{\omega^V_{kk}}{\omega^V_{jj}}} \nu^V_{jk} Y^{(h)}_{ik}\right)^2 \right\}^{-1}.    
\end{equation}
Note that the initial value of $\omega^V_{jj}$ is set to the sample variance; that is, $\hat{\omega}_{jj}^{V}=(np)\big\{\sum^n_{h=1}\sum^p_{i=1}(Y^{(h)}_{ij})^2\big\}^{-1}$.

Upon convergence, let $\hat{\theta}=(\hat{\nu}^V_{jk})_{j\neq k}$ and $\hat{w}=(\hat{\omega}^V_{kk})_{k}$ denote the minimizers of the loss function in (\ref{eq:loss}). Recalling the relationship $\hat{b}_{k} = \nu^V_{jk} \sqrt{\omega^V_{kk} / \omega^V_{jj}}$, we can readily estimate the precision matrix as $\widehat{V^{-1}} = (-\hat{\nu}^V_{jk} \sqrt{\hat{\omega}^V_{jj} \hat{\omega}^V_{kk}})_{j,k \in [q]}$, where $\hat{\nu}^V_{jj}\equiv -1$, $\forall j$.
The other precision matrix $U^{-1}$ can be estimated analogously by applying the same procedure to the transposed matrices $\{(Y^{(h)})^\top\}_{h=1}^n$. Note that it requires another $\ell_1$-regularized least squares problem. We add a few remarks on the proposed estimator.

\begin{itemize}
\item The proposed method simultaneously solves $q$ node-wise regressions, while \cite{Xia:2017}, \citet{Chen:2019}, and \citet{Shin:2025} solved $q$ individual penalized regressions separately. As pointed out in Section~\ref{sec:intro}, solving these regressions separately suffers from asymmetry of sparsity patterns and moreover computational burden for searching $q$ tuning parameters.
    
\item %% positive definite
A desirable property of a precision matrix estimator is the positive-definiteness, which is not guaranteed by the proposed method in general. However, we have empirically found that the resulting estimates are positive definite in almost all cases considered in our simulation study. Moreover, as noted by the previous work (\cite{Peng:2009}) for vector graphical models, the final estimate satisfies the positive-definite condition whenever the true network structure is sparse enough.

\item %% identifiability
Under the Kronecker product structure, each covariance matrix is only identifiable up to a multiplicative constant, as $(c \cdot V) \otimes (U/c) = V \otimes U$ for any $c> 0$. Therefore,
one may need a constraint to define a unique estimator such as $\text{tr}(U)=p$, $\text{tr}(V)=q$, or $U_{11}=1$. In our implentation, we adjust the estimate by $\omega^V_{jk} \leftarrow \omega^V_{jk} / \omega^V_{11}$ $\forall j,k \in [q]$ and 
$\omega^U_{jk} \leftarrow \omega^U_{jk} \cdot \omega^V_{11}$ $\forall j,k \in [p]$. This forces $(\widehat{V^{-1}})_{11}=1$, while leaving the Kronecker product matrix unchanged.

\item %% weights
Due to heterogeneity of the residual variances across the $q$ regression equations, choosing weights $\{f_j^V\}_{j\in[q]}$ adaptively may lead to more balanced network estimation.
Following \citet{Peng:2009}, we consider three choices. First, we consider uniform weights,  $f_j^V \equiv 1$. This treats all nodes as equally important.
Second, we consider inverse-variance weights, $f^V_j=q\cdot \omega_{jj}^V/{\sum_k\omega_{kk}^V}$, where $\omega_{jj}^V$ is the reciprocal of the residual variance defined in \eqref{eq:update_var}, evaluated at the current iterate. Similarly to the weighted least squares, this scheme assigns larger weights to regressions with smaller residual variance. 

Third, we consider degree-based weights, where $f_j^V= q\cdot d_j^V / \sum_k d_k^V$ where $d_j^V=\sum_{k\neq j}\text{I}(\nu_{jk}^V \neq 0)$ is evaluated at the current iterate. 
This choice gives weights proportional to node degree in order to account for local structural complexity such as hub structure.
Alternatively, we can fix $f_j^V$ at values reflecting reliable prior knowledge on the relative importance of the nodes, for example, based on external biological annotations.

\end{itemize}

We summarize the estimation procedure in Algorithm 
\ref{alg:matSPACE}. We name the proposed method ``matSPACE'' after the previous method (SPACE), as it can be viewed as a ``mat''rix-variate extension of SPACE.
The source code of the R implementation is available at \url{https://github.com/kimhyew1/matSPACE}.

\begin{algorithm}[H]
\caption{matSPACE algorithm}
\label{alg:matSPACE}
\begin{algorithmic}[1]
\raggedright
\Require Data $\{Y^{(h)}\}_{h=1}^n\subset\mathbb R^{p\times q}$; penalties $\lambda^V,\lambda^U$; weight schemes $f_V,f_U\in\{\mathrm{equal},\mathrm{variance},\mathrm{degree}\}$
\State Set $\hat\nu^V_{k\ell}\leftarrow0\ (k\neq \ell)$, $\hat\nu^V_{kk}\leftarrow-1$, $\hat f^V_k\leftarrow1$, and $\hat\omega^V_{kk}$ via \eqref{eq:update_var}, for $k,\ell\in[q]$ \Comment{initialization}
\While{$(\hat\nu^V_{kl})_{k,\ell\in[q]}$ not converged} \Comment{shooting algorithm}
            \State Update $\hat\nu^V_{k\ell}$ via \eqref{eq:update_nu}, $\forall k, \ell$, evaluated at $(\lambda_V,\hat f^V,\hat\omega^V)$
    \EndWhile
    \State $\hat\beta^V_{k\ell}\leftarrow\hat\nu^V_{k\ell}\sqrt{\hat\omega^V_{\ell\ell}/\hat\omega^V_{kk}}\ (k\neq \ell)$; \ update $\hat\omega^V_{kk}$ via \eqref{eq:update_var} with $\hat\beta^V$ \Comment{partial variance update}
    \State Update $\hat f^V_k$ according to the chosen scheme $f_V$, using $\hat\omega^V,\hat\nu^V$ \Comment{weight update}
\State Repeat step 1-6 with the transposed data, $\{(Y^{(h)})^\top\}_{h=1}^n$, and $(\lambda_U,f_U)$ in place of $(\lambda_V,f_V)$, yielding $\hat\nu^U,\hat\omega^U,\hat f^U$
\State $\hat\omega^U_{rs} \leftarrow \hat\omega^U_{rs}\cdot\hat\omega^V_{11}$, $\forall r,s\in[p]$, and $\hat\omega^V_{k\ell} \leftarrow \hat\omega^V_{k\ell}/\hat\omega^V_{11}$, $\forall k,\ell\in[q]$ \Comment{identifiability constraint}
\Ensure Precision matrices $\widehat{V^{-1}}=(\hat\omega^V_{k\ell})_{k,\ell\in[q]}$ and $\widehat{U^{-1}}=(\hat\omega^U_{rs})_{r,s\in[p]}$
\end{algorithmic}
\end{algorithm}

\subsection{Selection of regularization parameters}
The proposed model has an $\ell_1$-penalty parameter $\lambda^V$ to be tuned for estimation of $V^{-1}$. To select the optimal level of regularization, we compute an information criterion with the estimates obtained at each value of $\lambda^V$. The estimates in this subsection thus depend on $\lambda^V$, but we omit it for notational simplicity.

Given the estimates at $\lambda$, the residual sum of squares for the $j$-th regression is defined by
$$
RSS_j(\lambda)=\sum^n_{h=1}\sum^p_{i=1}\left(Y^{(h)}_{ij}-\sum_{k\neq j}\sqrt{\frac{\hat{\omega}^V_{kk}}{\hat{\omega}^V_{jj}}}\hat{\nu}^V_{jk}Y^{(h)}_{ik}\right)^2.
$$
The Bayes information criterion for the $j$-th regression takes the following form:
\begin{equation}\label{eq:bic_new}
BIC_j(\lambda)=np \cdot \log\left(RSS_j(\lambda)\right)+\log(np)\cdot \#\{k:k\neq j,\,\hat{\nu}^{V}_{jk}\neq0 \}.
\end{equation}

The optimal value is chosen to minimize their sum, i.e. $\lambda^{opt}:=\arg\min_{\lambda \in \Lambda}\sum^q_{j=1}BIC_j(\lambda)$, where $\Lambda$ is a set of candidates. We choose the range of candidates as follows.
Recall that the update formula \eqref{eq:update_nu} for $\nu_{jk}^V$ follows the soft-thresholding rule: $\nu_{jk}^V \leftarrow S(G_{jk}, \lambda) / D_{jk}$,
where 
$$
G_{jk}=\sum_{h,i}f_k^V\tilde{Y}_{ijk}^{(h)}\sqrt{\frac{\omega^V_{kk}}{\omega^V_{jj}}}Y^{(h)}_{ik}+\sum_{h,i}f_j^V\tilde{Y}_{ikj}^{(h)}\sqrt{\frac{\omega^V_{jj}}{\omega^V_{kk}}}Y^{(h)}_{ij}, \quad 
D_{jk}=\sum_{h,i}\frac{\omega^V_{jj}}{\omega^V_{kk}}(Y^{(h)}_{ij})^2+\sum_{h,i}\frac{\omega^V_{kk}}{\omega^V_{jj}}(Y^{(h)}_{ik})^2.
$$
Since all partial correlation estimates are identically zero for all $\lambda > \lambda_{\max}:=\max_{j\neq k}|G_{jk}|$, it suffices to search over the interval $[0, \lambda_{\max}]$. 
Following a standard practice, we construct a sequence of $K$ values of $\lambda$ decreasing from $\lambda_{\max}$ to  $\lambda_{\min}:=\epsilon\lambda_{\max}$ on the log scale. Unless otherwise stated, we use $\epsilon=10^{-6}$ and $K=40$.

\section{Theoretical result}\label{sec:theory}

Throughout this section, we omit the superscript $V$ for simplicity and assume $\text{tr}(U)=p$ as an identifiability condition.
Denote by $\omega^*$ the true value of $\omega$ and by $\theta^*=(\nu^*_{jk};j>k)$ the (enumerated) true partial correlations.
We first introduce assumptions required to derive theoretical guarantees.
\begin{assumption}[Boundedness of precisions]\label{assum:bd_w}
There exists $B_\omega>1$ such that
$$
1/B_\omega \le \min_{j} \omega^*_{jj} \le \max_{j} \omega^*_{jj} \le B_\omega.
$$
\end{assumption}

\begin{assumption}[Boundedness of eigenvalues]\label{assum:eigenvalue}
    The eigenvalues of $V$ are bounded, i.e. there exists $\eta \in (0,1)$ such that $\eta < \lambda_{\min}(V) \le \lambda_{\max}(V) < 1/\eta$.
\end{assumption}
\noindent
In fact, Assumption \ref{assum:eigenvalue} implies Assumption \ref{assum:bd_w}; however, we keep both constants for transparency.

We prove the consistency of the proposed estimator $\hat{\theta}$ when the estimator $\hat{\omega}$ of $\omega^*$ is given fixed. To be specific, consider the squared loss after normalization
$$
L_n(\theta;\omega) = \dfrac{1}{2pn}\sum_{h=1}^n \sum_{i=1}^p \sum_{j=1}^q f_j \left(Y^{(h)}_{ij} -  \sum_{k:k\neq j}  \gamma_{jk}(\omega)\nu_{jk} Y^{(h)}_{ik}\right)^2,
$$
where we define $\gamma_{jk}(\omega) = \sqrt{\omega_{kk}/\omega_{jj}}$ as a function of $\omega = (\omega_{jj})_{j}$. Then, we show the consistency of the estimator defined by
$$
\hat{\theta} =\hat{\theta}(\hat{\omega})= \arg\min_\theta L_n(\theta;\hat{\omega}) + \lambda || \theta ||_1,
$$
when $f_j>0$ are fixed.
Let $\widehat{\Sigma}= \sum_{h}(Y^{(h)})^\top Y^{(h)} /(np)$ be the covariance estimator for column dimension.
We denote the convergence rate $r_\omega$ and $r_\Sigma$ by $||\hat{\omega} - \omega^*||_\infty=O_p(r_\omega)$ and 
$||\hat{\Sigma} - V||_{\max} = O_p(r_\Sigma)$, respectively. From Lemma~\ref{lem:rate_Sigma_elemmax} in Appendix~\ref{sec:pf_lemma}, we know that 
$$
r_\Sigma = O
\left((\max_{1\le j\le q} V_{jj})
\frac{\|U\|_{F}}{p}
\sqrt{\frac{\log(p\vee q)}{n}}
\right)
$$
Let $S=\{(j,k):j>k,\ \nu_{jk}^*\neq 0\}$ denote the support of the true partial correlations.

\begin{theorem}\label{thm:main} 
Suppose Assumptions \ref{assum:bd_w} and \ref{assum:eigenvalue} hold. 
Assume the vectorization of $Y^{(h)}$, $h=1,\ldots, n$, are independently distributed as $N_{pq}(0, V \otimes U)$.
Assume $np > c_1 (q\vee \log n) / \eta^6$ for some numerical constant $c_1>0$.
Assume that the estimator $\hat{\omega}$ satisfies 
$||\hat{\omega} - \omega^*||_\infty\le \min\{1/(2B_\omega), r_\omega\}$ and with probability at least $1-\delta_\omega$,
and $r_\Sigma = O(r_\omega)$.
\noindent
Choose the tuning parameter $\lambda$ such that 
\[ 
\lambda = \Omega\left( (\max_j f_j)(1+|S|) \left\{ B_\omega^2 r_\Sigma + B_\omega^4(\max_j V_{jj})r_\omega \right\} \right). 
\] 

Then, with probability at least \[ 1-\delta_\omega - 3/(p\vee q)^2 - d\exp\left\{ -c_2 \eta^6np +q  \right\}, \] for some positive numerical constants $c_2$ and $d$, we have
\[ 
\|\hat{\theta}-\theta^*\|_2 \le  \frac{6\lambda\sqrt{|S|}} {\eta\min_j f_j}.
\]
\end{theorem}
\noindent
The proof of the theorem is given in Appendix~\ref{sec:pf_main}. 
We present a suitable estimator $\tilde{\omega}$ in Appendix~\ref{sec:suitable_w}
and prove that its rate is $||\tilde{\omega} - \omega^*||_\infty = O_p(r_{\tilde{\omega}})$ in Lemma~\ref{lem:rate_w} 
(with $\alpha=3$), where
$$
r_{\tilde{\omega}} = B_\omega \max\left\{\dfrac{||U||_F}{p} \sqrt{\dfrac{\log(p\vee q)^3}{n}}, 
\dfrac{||U||_2}{p} \cdot \dfrac{q \vee \log(p\vee q)^3}{n} \right\}.
$$
Then, this yields the convergence rate of order
$$
\|\hat{\theta}-\theta^*\|_2 = O_p \left(\frac{ (\max_j f_j)|S|^{3/2} B_\omega^4 (\max_j V_{jj}) r_{\tilde{\omega}}}{\eta \cdot \min_j f_j}
\right).
$$
As a consequence, if $q = o(np)$, the proposed estimator is consistent with the rate of order $O(\sqrt{q\log(p\vee q)/(np)})$ under the assumed conditions.

\section{Simulation study}\label{sec:simul}

We generate $n$ samples, $Y^{(1)}, ...,Y^{(n)}$, from $N_{pq}(0, V \otimes U)$ independently. 
Following \citet{Liu:2013} and \citet{Shin:2025}, we adopt four network structures in this study: band, cluster, hub, and random. For each structure, we first generate an adjacency matrix and then assign nonzero values to the support positions. For diagonal entries, we adjust them by adding a small positive value to guarantee the positive definiteness of the resulting matrix. Then, each diagonal entry is multiplied by a random number from $\text{Unif}(1,5)$, allowing for heterogeneity. 
The details for construction can be found in Appendix~\ref{sec:supp_generating}. 
The resulting precision matrices are illustrated in Figure~\ref{fig:GraphStructure}.

\begin{figure}[H]
\centering

\includegraphics[page=1, width=0.24\textwidth, trim=30pt 40pt 30pt 20pt, clip]{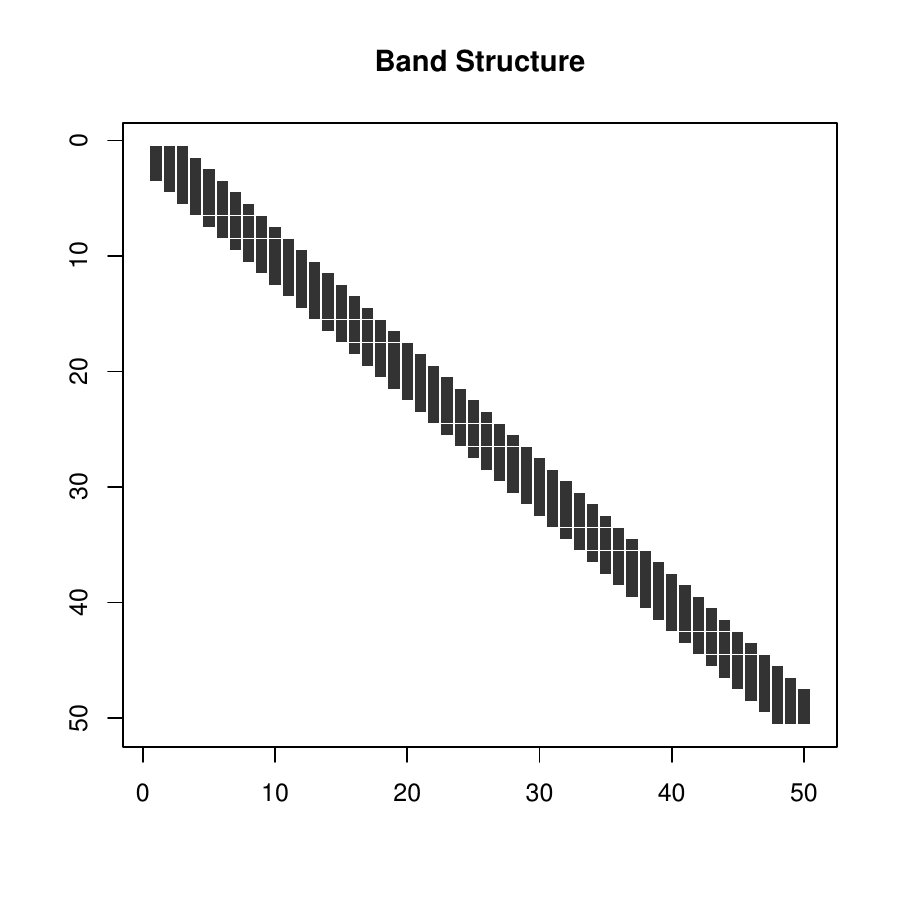}
\includegraphics[page=2, width=0.24\textwidth, trim=30pt 40pt 30pt 20pt, clip]{260916_jmatNS_v2_assets/fig/network_structure.pdf}
\includegraphics[page=3, width=0.24\textwidth, trim=30pt 40pt 30pt 20pt, clip]{260916_jmatNS_v2_assets/fig/network_structure.pdf}
\includegraphics[page=4, width=0.24\textwidth, trim=30pt 40pt 30pt 20pt, clip]{260916_jmatNS_v2_assets/fig/network_structure.pdf}

\caption{Network structures used in the simulation study: band, cluster, hub, and random networks with dimension $q=50$.}
\label{fig:GraphStructure}
\end{figure}

In our simulation study, we fix $U^{-1}$ to follow a banded network, while the structure of $V^{-1}$ varies among the four networks. $Y^{(1)}, ...,Y^{(n)} $ are iid generated from $N_{pq}(0, V\otimes U)$. We consider sample sizes $n \in \{10, 40, 70\}$ and matrix dimensions $(p,q)$ with $p,q \in \{20,50\}$. For each simulation setup, we conduct 300 replications.

\subsection{Comparative methods and computational details}
The proposed method (matSPACE) has two weighting schemes: inverse-variance based (``sw'') and degree-based (``dew'') - matSPACE-sw, and matSPACE-dew - reported in Table~\ref{tab:band_n}. 
We compare the proposed method with three existing approaches: ``GEMINI''  in \citet{Zhou:2014}, the matrix neighborhood selection method, or ``matNS'', in \citet{Shin:2025}, and the multiple testing--based graph estimation method with FDR control, or ``G-FDR'' in \citet{Chen:2019}. 
We exclude the penalized flip-flop method \citep{Leng:2012} in comparison, due to computational costs; under $n=10,p=20,q=20$, a single estimation including tuning parameter selection takes about 0.4 hours, implying that 300 replications would require about 5 days to complete.
Due to the asymmetric nature, matNS either adopts AND or OR rules \citep{m&b2006} to estimate the final graph. Thus, we consider two variants; (i) ``matNS-AND'' using AND rule and (ii) ``matNS-OR'' using OR rule. For G-FDR, the significance level for FDR control is set to 0.10, which is referred to as ``G-FDR-0.10''. 

All methods involve regularization parameters for row and column dimensions, each of which is optimally selected via 5-fold cross-validation.
For instance, 
GEMINI minimizes the cross-validated negative Gaussian log-likelihood, separately for the row and column penalties. matNS minimizes the cross-validated mean squares error. G-FDR minimizes the cross-validated prediction error for the nodewise Lasso penalties and the squared Frobenius-norm error for the covariance thresholding parameters.

\subsection{Results}\label{sec:results}

Let $E$ and $\hat{E}$ denote the edge sets of the true graph and the estimated graph, respectively. We define the counts of true positive (TP), true negative (TN), false positive (FP), and false negative (FN) as 
$$
\text{TP}=|\hat{E} \cap E|, \quad \text{TN}=|\hat{E}^C \cap E^C|, \quad \text{FP}=|\hat{E} \cap E^C|, \quad \text{FN}=|\hat{E}^C \cap E|.
$$
For evaluation of graph recovery, we measure true positive rate (TPR), false positive rate (FPR), and Matthews correlation coefficient (MCC) scores, defined by:
$$
\begin{array}{c}
\text{TPR}=\frac{\mathrm{TP}}{\mathrm{TP}+\mathrm{FN}}, \quad
\text{FPR}=\frac{\mathrm{FP}}{\mathrm{FP}+\mathrm{TN}}, \quad
\text{MCC}=\frac{\mathrm{TP}\cdot \mathrm{TN} - \mathrm{FP}\cdot \mathrm{FN}}{\sqrt{(\mathrm{TP}+\mathrm{FP})(\mathrm{TP}+\mathrm{FN})(\mathrm{TN}+\mathrm{FP})(\mathrm{TN}+\mathrm{FN})}}.
\end{array}
$$
To further assess performance across the full range of sparsity levels, we construct the ROC curve by plotting TPR against FPR across the entire regularization path. The area under the ROC curve (AUC) is then used as an aggregate measure of graph recovery performance, 
where a larger value indicates better recovery.
Due to space limitations, we report here the representative results for the setting where $U^{-1}$ and $V^{-1}$ are both fixed at the band structure. Results for the other structures are provided in Appendix~\ref{sec:supp_sim}. 

\begin{table}[H]
\centering
\spacingset{1}
\caption{Performance of graph estimation of $V^{-1}$, where true $U^{-1}$ and $V^{-1}$ are both band structure. Each entry reports an average
of each metric with the standard deviation in parentheses, based on 300 data replications.
Boldface indicates the best-performing method.}
\label{tab:band_n}
\renewcommand{\arraystretch}{0.9}
\footnotesize
\setlength{\tabcolsep}{4pt}
\resizebox{0.90\textwidth}{!}{%
\begin{tabular}{clcccc}
\toprule
$n$ & Method & TPR & FPR & MCC & AUC \\
\midrule
\multirow{14}{*}{10} & \multicolumn{5}{l}{\textit{$p=20$, $q=20$}} \\
 & matSPACE-sw & \textbf{0.392 (0.028)} & 0.133 (0.009) & \textbf{0.277 (0.008)} & \textbf{0.706 (0.002)} \\
 & matSPACE-dew & 0.262 (0.021) & 0.077 (0.006) & 0.245 (0.007) & 0.683 (0.002) \\
 & matNS-AND & 0.021 (0.001) & 0.001 (0.000) & 0.084 (0.009) & 0.695 (0.002) \\
 & matNS-OR & 0.079 (0.005) & 0.012 (0.000) & 0.160 (0.010) & 0.686 (0.002) \\
 & GEMINI & 0.310 (0.033) & 0.100 (0.008) & 0.245 (0.012) & 0.691 (0.002) \\
 & G-FDR-0.10 & 0.001 (0.000) & \textbf{0.000 (0.000)} & 0.006 (0.001) & 0.537 (0.001) \\
\cmidrule(lr){2-6}
 & \multicolumn{5}{l}{\textit{$p=50$, $q=50$}} \\
 & matSPACE-sw & \textbf{0.489 (0.009)} & 0.067 (0.001) & \textbf{0.390 (0.002)} & \textbf{0.800 (0.001)} \\
 & matSPACE-dew & 0.275 (0.010) & 0.029 (0.000) & 0.319 (0.002) & 0.763 (0.001) \\
 & matNS-AND & 0.052 (0.001) & 0.001 (0.000) & 0.199 (0.003) & 0.781 (0.001) \\
 & matNS-OR & 0.146 (0.003) & 0.009 (0.000) & 0.268 (0.003) & 0.768 (0.001) \\
 & GEMINI & 0.483 (0.015) & 0.084 (0.003) & 0.361 (0.003) & 0.779 (0.001) \\
 & G-FDR-0.10 & 0.006 (0.000) & \textbf{0.000 (0.000)} & 0.045 (0.003) & 0.616 (0.001) \\
\midrule
\midrule
\multirow{14}{*}{40} & \multicolumn{5}{l}{\textit{$p=20$, $q=20$}} \\
 & matSPACE-sw & 0.731 (0.013) & 0.156 (0.006) & \textbf{0.521 (0.006)} & \textbf{0.864 (0.001)} \\
 & matSPACE-dew & 0.549 (0.040) & 0.113 (0.011) & 0.461 (0.007) & 0.828 (0.001) \\
 & matNS-AND & 0.105 (0.004) & 0.001 (0.000) & 0.269 (0.010) & 0.843 (0.001) \\
 & matNS-OR & 0.249 (0.011) & 0.013 (0.000) & 0.393 (0.009) & 0.833 (0.001) \\
 & GEMINI & \textbf{0.789 (0.011)} & 0.268 (0.012) & 0.434 (0.006) & 0.843 (0.001) \\
 & G-FDR-0.10 & 0.111 (0.005) & \textbf{0.000 (0.000)} & 0.284 (0.013) & 0.741 (0.003) \\
\cmidrule(lr){2-6}
 & \multicolumn{5}{l}{\textit{$p=50$, $q=50$}} \\
 & matSPACE-sw & 0.817 (0.004) & 0.072 (0.001) & \textbf{0.605 (0.004)} & \textbf{0.937 (0.000)} \\
 & matSPACE-dew & 0.550 (0.020) & 0.037 (0.001) & 0.541 (0.002) & 0.898 (0.000) \\
 & matNS-AND & 0.231 (0.004) & 0.001 (0.000) & 0.450 (0.004) & 0.916 (0.000) \\
 & matNS-OR & 0.413 (0.005) & 0.009 (0.000) & 0.550 (0.003) & 0.906 (0.000) \\
 & GEMINI & \textbf{0.865 (0.004)} & 0.178 (0.007) & 0.458 (0.007) & 0.918 (0.000) \\
 & G-FDR-0.10 & 0.237 (0.003) & \textbf{0.000 (0.000)} & 0.469 (0.003) & 0.844 (0.001) \\
\midrule
\midrule
\multirow{14}{*}{70} & \multicolumn{5}{l}{\textit{$p=20$, $q=20$}} \\
 & matSPACE-sw & 0.836 (0.006) & 0.150 (0.004) & \textbf{0.605 (0.006)} & \textbf{0.917 (0.001)} \\
 & matSPACE-dew & 0.727 (0.032) & 0.170 (0.021) & 0.527 (0.009) & 0.882 (0.001) \\
 & matNS-AND & 0.186 (0.007) & 0.002 (0.000) & 0.376 (0.010) & 0.897 (0.001) \\
 & matNS-OR & 0.378 (0.014) & 0.014 (0.000) & 0.516 (0.009) & 0.888 (0.001) \\
 & GEMINI & \textbf{0.893 (0.005)} & 0.322 (0.012) & 0.460 (0.006) & 0.897 (0.001) \\
 & G-FDR-0.10 & 0.287 (0.011) & \textbf{0.000 (0.000)} & 0.489 (0.010) & 0.833 (0.002) \\
\cmidrule(lr){2-6}
 & \multicolumn{5}{l}{\textit{$p=50$, $q=50$}} \\
 & matSPACE-sw & 0.900 (0.002) & 0.065 (0.001) & \textbf{0.678 (0.005)} & \textbf{0.969 (0.000)} \\
 & matSPACE-dew & 0.662 (0.018) & 0.038 (0.003) & 0.630 (0.005) & 0.938 (0.000) \\
 & matNS-AND & 0.355 (0.005) & 0.001 (0.000) & 0.571 (0.003) & 0.952 (0.000) \\
 & matNS-OR & 0.557 (0.004) & 0.008 (0.000) & 0.669 (0.002) & 0.945 (0.000) \\
 & GEMINI & \textbf{0.938 (0.001)} & 0.195 (0.008) & 0.480 (0.010) & 0.954 (0.000) \\
 & G-FDR-0.10 & 0.442 (0.004) & \textbf{0.000 (0.000)} & 0.648 (0.002) & 0.912 (0.000) \\
\bottomrule
\end{tabular}
}
\end{table}

Table \ref{tab:band_n} shows that matSPACE-sw tends to identify more edges than matSPACE-dew, thus accompanied by higher TPR and FPR. The overall performance measured by MCC is consistently higher in matSPACE-sw, except a single case in hub structure with $n=10, p=q=20$. For AUC, matSPACE-sw wins every comparative method across all settings.

Compared with the competing methods, matSPACE achieves the highest AUC and MCC in most settings. 
As $n$ grows, matSPACE maintains a stable FPR, whereas that of GEMINI increases. On the other hand, matNS has a limited detection power with small sample size, resulting in low TPR and MCC. G-FDR performs poorly as it recovers few true edges, especially when the sample size is small ($n=10$). In this case, AUC is close to 0.5, indicating performance near random guessing. 

Increasing row dimension improves the estimation of the precision matrix $V^{-1}$ for the other dimension.
This is because the loss function in \eqref{eq:loss} consists of the sum of squared errors from $np$ nodewise regressions, 
yielding an effective sample size of $np$.
We observe similar patterns in all the other methods, which have also been reported in previous studies.

\section{Real data analysis}\label{sec:real}

To demonstrate a practical usability of the proposed method, we analyze the Longitudinal Plasma Proteomic Profiles of Pulmonary Tuberculosis Treatment dataset \citep{Call:2026}, originally collected to identify plasma-based biomarkers predictive of disease relapse during antituberculosis therapy. The study cohort consisted of 60 patients with confirmed pulmonary tuberculosis: 30 who achieved sustained cure (``Cured'' group) and 30 who subsequently relapsed (``Relapsed'' group). Patients were treated under either the standard National Tuberculosis Program (NTP) or the Rapid Evaluation of Moxifloxacin in Tuberculosis (REMoxTB) clinical trial protocol, with treatment administered over 26 weeks and follow-up extended to week 52 to ascertain relapse status. Plasma samples were collected at seven pre-specified time points (0, 2, 4, 8, 17, 26, and 52 weeks), yielding quantification of 2,418 proteins implicated in immune response, tissue repair and recovery, and lipid metabolism. We refer readers to the original paper for further details.

Our goal is to identify the key proteins and biological pathways that distinguish two groups (Cured and Relapsed). 
Since the prior work \citep{Call:2026}, has shown that treatment regimen does not meaningfully influence proteomic profiles, we pool patients treated under the NTP and REMoxTB regimens within each group, rather than analyzing each regimen separately.
Patients and proteins with missing observations at any time point were removed. For each protein, we center expression values within each group.
The resulting data consist of $p=641$ proteins measured at $q=7$ time points for $n_1=30$ patients in the Cured group and $n_2=23$ patients in the Relapsed group. In this context, the row-wise network of dimension $p=641$ indicates the protein-protein interaction network, while the column-wise network means the temporal network among $q=7$ time points.
The proposed method is used to estimate networks for both dimensions, where the optimal model selection is carried out over a grid of 40 candidates $\lambda$. 
For the row-wise network, we use degree-based weighting, motivated by the hub structure often observed in protein networks \citet{Barabasi:2004}. 
For the column-wise network, we adopt inverse-variance weighting, which performed better across most simulation settings including hub-structured networks. 

\begin{figure}[H]
    \centering
    \begin{subfigure}{0.48\linewidth}
        \centering
        \includegraphics[page=1, width=\linewidth, trim=1cm 0.5cm 1cm 2cm, clip]{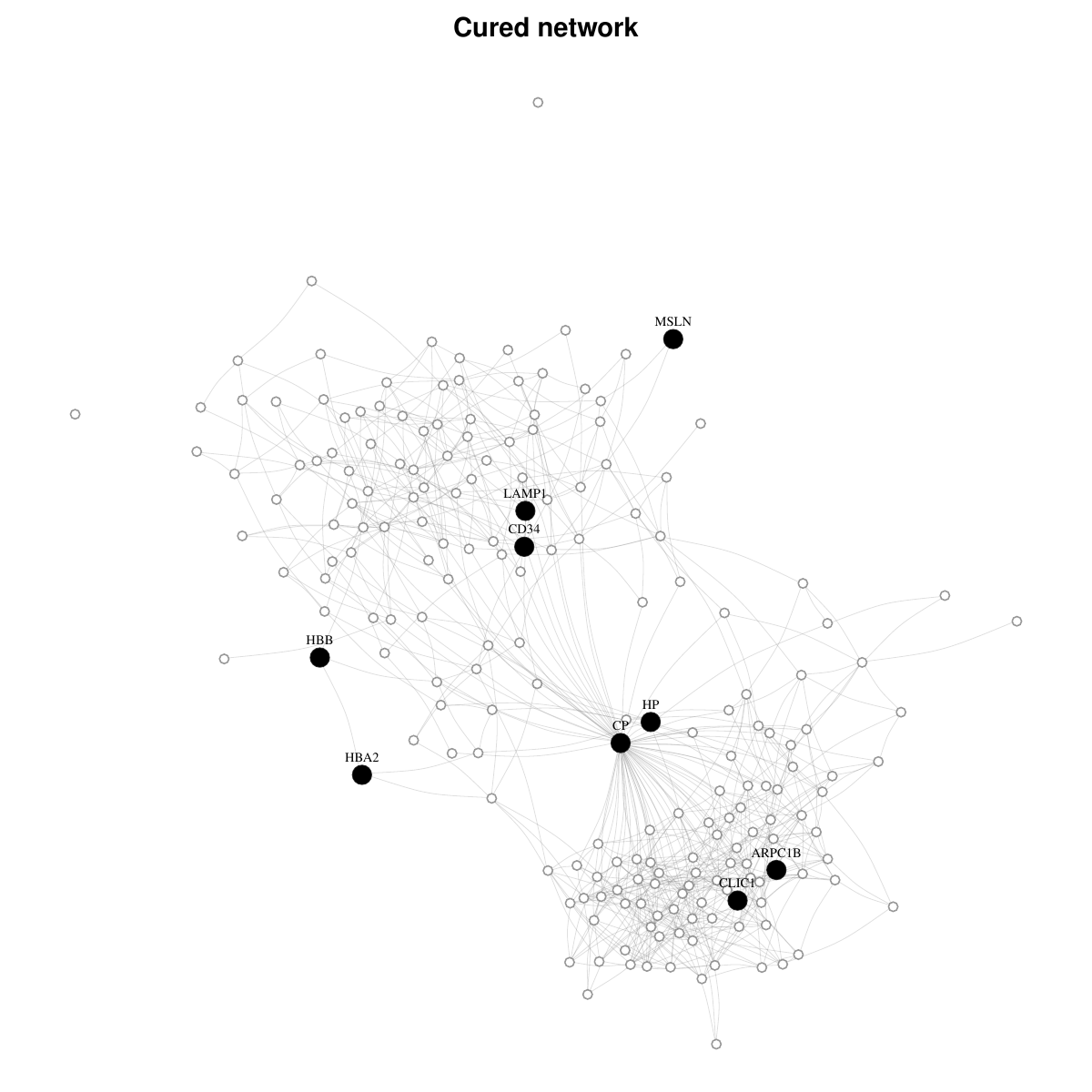}
        \caption{Cured group}
        \label{fig:network_cured}
    \end{subfigure}
    \hfill
    \begin{subfigure}{0.48\linewidth}
        \centering
        \includegraphics[page=2, width=\linewidth, trim=1cm 0.5cm 1cm 1cm, clip]{260916_jmatNS_v2_assets/fig/network_top200.pdf}
        \caption{Relapsed group}
        \label{fig:network_relapsed}
    \end{subfigure}
    \caption{Estimated protein networks for plasma proteomic measurements in the cured (left) and relapsed (right) groups. For visualization, only 200 edges with the largest  estimated (absolute) partial correlation are shown in each network; nodes with no remaining edges after the thresholding are omitted. Filled circles indicate hub proteins, defined as top 5\% of nodes by the degree within each network.}
    \label{fig:network_protein}
\end{figure}

The Relapsed network shows denser connectivity than the Cured network, with an overall network density of 0.151 compared to 0.130. The two networks also differ in structural organization. The Cured network consists of two sub-networks, each of them is densely connected internally. The two clusters appear to be linked primarily through a protein named ceruloplasmin (CP), yielding a modular architecture centered on a single connecting hub (defined by the one with a large degree). The Relapsed network, by contrast, form a single, more tightly integrated cluster. Its hub proteins are closely located and densely interconnected, which leads to a more cohesive overall architecture.

For biological interpretation, functional enrichment analysis is performed for each group based on hub proteins, defined as the top 5\% of nodes by degree within each network.
Based on the Database for Annotation, Visualization and Integrated Discovery (DAVID) \citep{Huang:2009, Sherman:2022}, we perform functional annotation clustering, which groups related biological annotations identified by the enrichment analysis into clusters and thus facilitates higher-level interpretation.
The functional enrichment analysis is based on the Gene Ontology (GO) Biological Process (BP), Cellular Component (CC), and Molecular Function (MF) direct categories, together with Kyoto Encyclopedia of Genes and Genomes (KEGG) pathway annotations. When implementing DAVID, we choose the following parameters: a kappa similarity term overlap of 3, a similarity threshold of 0.50, initial and final group membership thresholds of 3, a multiple linkage threshold of 0.50, and an EASE score threshold of 1.0. 

\begin{table}[H]
\centering
\spacingset{1}

\footnotesize
\renewcommand{\arraystretch}{0.85}
\setlength{\tabcolsep}{4pt}
\caption{DAVID functional annotation clustering results for hub proteins in the Cured and Relapsed groups. Each cluster is named after the term with the smallest $p$-value within the cluster. Hub proteins are listed in alphabetical order.}
\label{tab:david_cluster}
\begin{tabular}{@{}ccc>{\raggedright\arraybackslash}p{3.9cm}>{\raggedright\arraybackslash}p{5.3cm}@{}}
\toprule
Cluster & Enrichment Score & No.\ of Terms & Representative Term & Hub Proteins \\
\midrule
\multicolumn{5}{l}{\textbf{Cured}} \\
C1 & 3.80 & 5  & Blood microparticle & APOA4, CALR, CLIC1, CP, HBA2, HBB, HP, YWHAZ \\
C2 & 2.71 & 4  & Actin binding & ARPC1B, CAP1, COTL1, PARVB, TAGLN2, TPM1, TPM4 \\
C3 & 0.97 & 4  & Oocyte meiosis & BLVRB, CALM1, CAVIN2, MAPRE2, PRKACA, SLC3A2, TPT1, YWHAZ \\
\midrule
\multicolumn{5}{l}{\textbf{Relapsed}} \\
R1 & 3.17 & 9  & Antibacterial humoral response & B2M, FCRL5, IGHA1, IGHD, IGHG3, IGKV3-20, IGLL5, PPBP, TF \\
R2 & 2.88 & 18 & HDL particle receptor binding & APOA1, APOA2, APOC3, PLA2G7, TF \\
R3 & 2.00 & 4  & Focal adhesion & ARHGDIB, ARPC1B, B2M, CAP1, PARVB, TLN1, TPM4 \\
R4 & 1.70 & 3  & Endoplasmic reticulum lumen & APOA1, APOA2, APOC3, B2M, CP, TF \\
\bottomrule
\end{tabular}
\end{table}

In the Cured group, three functional clusters are identified: (i) Cluster-C1: blood microparticle, (ii) Cluster-C2: actin binding, and (iii) Cluster-C3: oocyte meiosis. In the Relapsed group, four functional clusters are identified: (i) Cluster-R1: antibacterial humoral response, (ii) Cluster-R2: HDL particle receptor binding, (iii) Cluster-R3: focal adhesion, and (iv) Cluster-R4: endoplasmic reticulum lumen. Table \ref{tab:david_cluster} summarizes the enrichment score, number of terms, and associated hub proteins for each cluster.

Several of these clusters align with the key pathways and proteins found in \citet{Call:2026}. In the Cured group, Cluster-C1 (blood microparticle), driven by HBB, HBA2, and HP, has a close connection with the increased oxygen transport identified as a key recovery signature in cured participants. In the Relapsed group, Cluster-R1 (antibacterial humoral response), including IGHG3, IGHD, and IGKV3-20, parallels the enhanced humoral immune response identified as a major relapse signature, while Cluster-R2 (HDL particle receptor binding), driven by APOA1, APOA2, and APOC3, represents the HDL particle enrichment identified as another main relapse signature.

\begin{figure}[H]
    \centering
    \begin{subfigure}{0.45\linewidth}
        \centering
        \includegraphics[width=\linewidth, trim=0 0 5pt 0, clip]{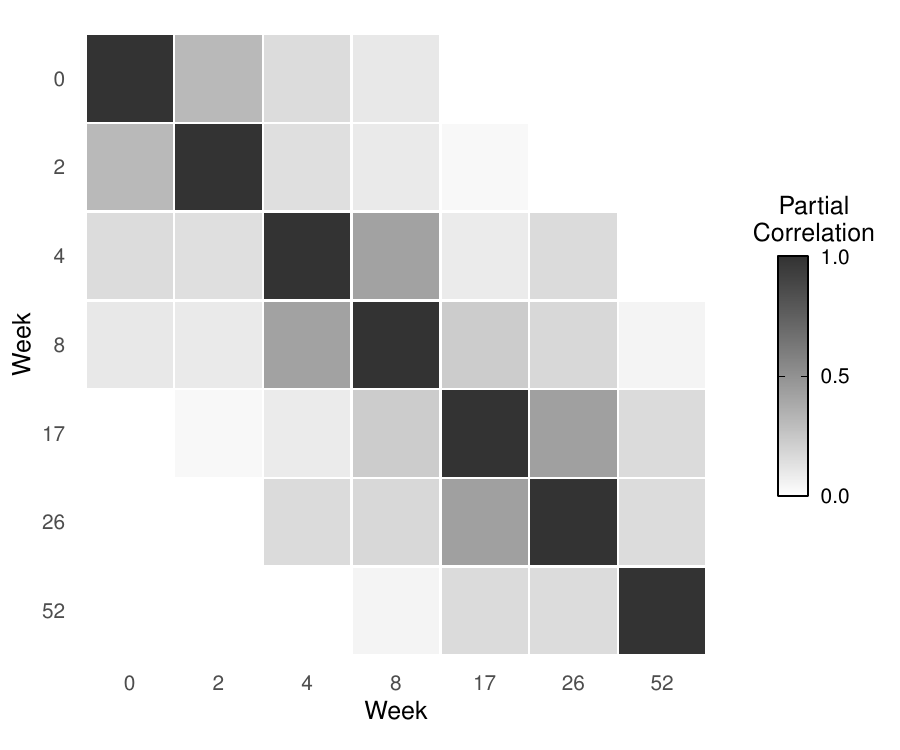}
        \caption{Cured group}
        \label{fig:heatmap_cured}
    \end{subfigure}
    \hfill
    \begin{subfigure}{0.45\linewidth}
        \centering
        \includegraphics[width=\linewidth, trim=0 0 5pt 0, clip]{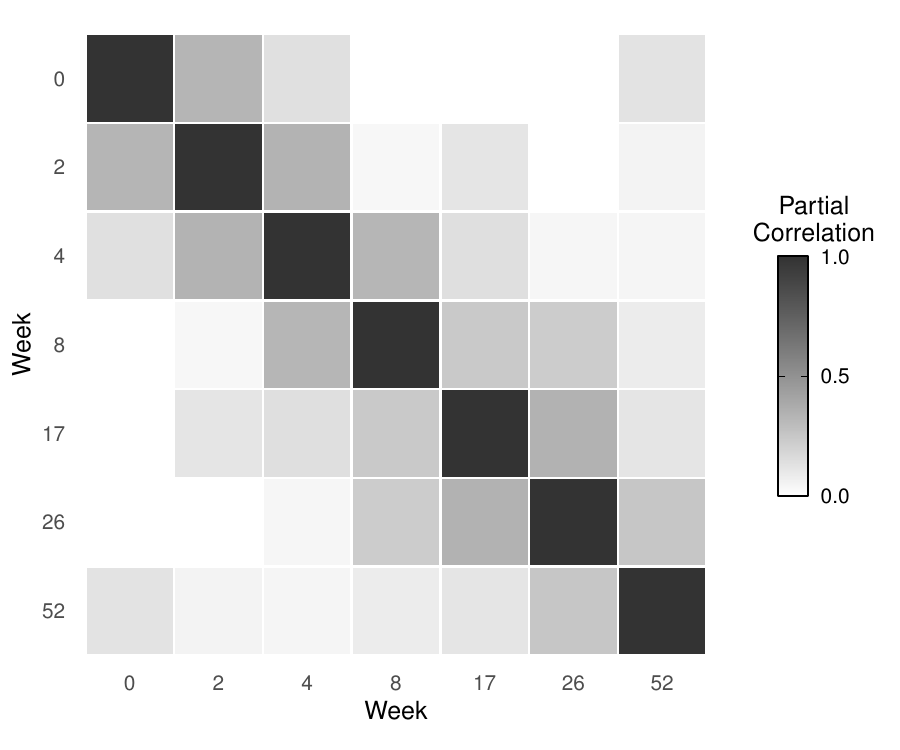}
        \caption{Relapsed group}
        \label{fig:heatmap_relapsed}
    \end{subfigure}
    \caption{The estimated column-wise network for the plasma proteomic measurements across the seven time points (0, 2, 4, 8, 17, 26, and 52 weeks). The left panel is for the Cured group, and the right panel for the Relapsed group.}
    \label{fig:heatmap_time}
\end{figure}

Figure \ref{fig:heatmap_time} visualizes the estimated column graph among temporal variables. In general, temporal dependency weakens as the interval between time points increases. This trend is clearly observed in the Cured group, but is partially observed in the other group. To be specific, the last row and column (at Week 52) in Figure \ref{fig:heatmap_time} deviates from the decaying pattern, which we believe may be due to the following reasons.
In the study design, patients received treatment for the first 26 weeks, followed by a 26-week post-treatment follow-up period. Since patients in the Relapsed group experienced recurrence during this follow-up period, the last phase may show biological characteristics different from the first 26 weeks. This may explain the different dependence patterns observed in the estimated precision matrix.

\section{Discussion}\label{sec:discus}

This paper develops a graphical model for matrix-variate data by jointly parameterizing the individual regression problems.
The key insight is that the node-wise regression coefficients can be expressed as a common set of partial correlations and conditional variances. 
This formulation naturally enforces symmetry in edge estimation and dramatically reduces the computational burden associated with tuning parameter selection. Our numerical results show that the simultaneous estimation leads to improved graph support recovery compared to existing ones.

Several directions for future research are worth investigating. An immediate extension is to graphical modeling for tensor-valued data. Under the covariance separability, such an extension is conceptually straightforward: for each mode of the tensor, standard graphical model methods (e.g. graphical lasso or neighborhood selection) are applied to observations obtained by stacking the tensor data across the rest of modes. 
A second direction is to develop a convex formulation for the joint estimation problem, along the lines of the work by \citet{Khare:2014}. The authors showed that the SPACE method for vector graphical models may fail to converge, partly due to the non-convexity of the optimization problem. Thus, a convex joint-estimation framework for matrix graphical models could achieve computationally more stable estimation.
Finally, we propose a new BIC-type criterion that considers dependence among matrix-variate observations. A theoretical justification for this criterion is another important direction for future work.

\section*{Acknowledgement}
Seongoh Park is supported by the government of the Republic of Korea (MSIT) and the National Research Foundation of Korea (RS-2024-00338876).

\section*{Data availability statement}
The dataset analyzed in this study is publicly available at \url{https://doi.org/10.1021/acs.jproteome.5c01206} as Supporting Information (see Table S1) of \citet{Call:2026}.

\clearpage
\bibliographystyle{apalike}
\bibliography{260916_jmatNS_v2_assets/references}

\clearpage
\appendix
\section*{Supplementary Materials}\label{supplementary-material}

This supplementary material provides the proof of the theoretical results, construction details of the precision matrices, and additional numerical experiments supporting the main paper.

\section{Proof of theoretical results}\label{sec:pf_theory}

We use the notation, assumptions, and main theorem stated in Section~\ref{sec:theory}.

\subsection{Proof of Theorem \ref{thm:main}}\label{sec:pf_main}
\begin{proof}
For a given estimate $\hat{\omega}$ of $\omega$, define the solution of partial correlations: for given $\lambda>0$,
$$
\hat{\theta} = \arg\min_\theta L_n(\theta;\hat{\omega}) + \lambda || \theta ||_1.
$$
Note that we will denote the loss function evaluated at $\omega = \hat{\omega}$ without specifying it unless it occurs confusion.
Then, the optimality implies
$$
L_n(\hat{\theta}) + \lambda || \hat{\theta} ||_1 \le L_n(\theta^*) + \lambda || \theta^* ||_1.
$$
Rearranging it with the difference $\Delta := \hat{\theta} - \theta^*$, we get
$$
L_n(\theta^* +\Delta) - L_n(\theta^*) \le  + \lambda (|| \theta^* ||_1 - || \theta^* +\Delta ||_1 ).
$$
As the loss function is quadratic in $\theta$, we get
$$
L_n(\theta^* +\Delta) - L_n(\theta^*) =
\nabla L_n(\theta^*)^\top \Delta +
\frac{1}{2}\Delta^\top \nabla^2 L_n(\theta^*) \Delta,
$$
where $\nabla L_n(\theta^*), \nabla^2 L_n(\theta^*)$ are the gradient vector and  Hessian matrix, respectively, of $L_n$ with respect to $\theta$ evaluated at $(\theta^*, \hat{\omega})$.

We state the two lemmas that are essential for the proof.
Assume $\text{tr}(U)=p$ for the identifiability of the separable covariance matrix.
In Section \ref{sec:suitable_w}, we
introduce a suitable $\hat{\omega}$ satisfying conditions in the following lemma, while the rate $r_\Sigma$ is specified in Lemma \ref{lem:rate_Sigma_elemmax}.
The proofs of the lemmas are pended until Appendix \ref{sec:pf_lemma}.

\begin{lemma}\label{lem:bound_grad}
Suppose Assumption \ref{assum:bd_w} holds.
Let $\hat{\omega}$ be any estimate of $\omega^*$  satisfying $||\hat{\omega} - \omega^*||_\infty\le \min\{1/(2B_\omega), r_\omega\}$ and with probability at least $1-\delta_\omega$,  and $r_\Sigma = O(r_\omega)$.
Then, choose $\lambda$ such that
$$
\lambda = \Omega\left( (\max_j f_j)(1+|S|) \left\{ B_\omega^2 r_\Sigma + B_\omega^4(\max_j V_{jj})r_\omega \right\} \right).
$$
where $\Omega$ denotes the asymptotic lower bound. Then,
$||\nabla L_n(\theta^*;\hat{\omega} )||_\infty \le \lambda/2$ holds with probability at least $1- \delta_\omega - 3/(p\vee q)^2$.
\end{lemma}

\begin{lemma}\label{lem:pd}
Suppose Assumption \ref{assum:eigenvalue} holds. Let $\kappa:= 0.5 \cdot \eta  \cdot \min_j f_j > 0$.
Assume $(c\eta^6np)/16-(q\vee \log n)\log 21 > 0$ for some numerical constant $c>0$.
Then, we have with probability at least $1-2\exp\{
-(c\eta^6np)/16
+q\log 21\}$ that
$$
t^\top \nabla^2 L_n(\theta^*) t \ge \kappa || t||_2^2, \quad \forall t.
$$
\end{lemma}
\noindent
Together with H\"{o}lder's inequality, Lemma \ref{lem:bound_grad} implies that  with probability at least $1- \delta_\omega - 3/(p\vee q)^2$
$$
|\nabla L_n(\theta^*)^\top \Delta |\le ||\nabla L_n(\theta^*)||_\infty ||\Delta||_1 \le \lambda ||\Delta||_1  / 2.
$$
Combining this with Lemma \ref{lem:pd}, we get
$$
- \dfrac{\lambda}{2} ||\Delta||_1 + \dfrac{\kappa}{2} ||\Delta||_2^2 \le \lambda (|| \theta^* ||_1 - || \theta^* +\Delta ||_1 ),
$$
with probability at least $1-3/(p\vee q)^2 - 2\exp\{-(c\eta^6np)/16 +q\log 21\}$.

Let $S = \{j: \theta^*_j \neq 0\}$ be the index set of non-zero partial correlations, which can be thought as an enumeration of $\{(j,k): \nu^*_{jk}\neq 0\}$. Then,
$$
|| \theta^* ||_1 - || \theta^* +\Delta ||_1
= || \theta^*_S ||_1 - (||  \theta^*_S + \Delta_S ||_1 +
||\Delta_{S^c} ||_1)
\le  ||\Delta_S ||_1 - || \Delta_{S^c} ||_1,
$$
where the last inequality uses $||  \theta^*_S + \Delta_S ||_1 \ge ||  \theta^*_S||_1 - ||\Delta_S ||_1$. Thus,
$$
\dfrac{\kappa}{2} ||\Delta||_2^2 \le \dfrac{\lambda}{2}(3||\Delta_S ||_1 - ||\Delta_{S^c}||_1).
$$
Dropping $||\Delta_{S^c}||_1$, we can get
$$
\dfrac{\kappa}{2} ||\Delta||_2^2 \le \dfrac{3\lambda}{2} ||\Delta_S||_1.
$$
Using $ ||\Delta_S||_1 \le \sqrt{|S|}  ||\Delta_S||_2$, we conclude that
$$
||\Delta||_2 \le \dfrac{3\lambda\sqrt{|S|}}{\kappa}.
$$
\end{proof}

\section{Proof of the lemmas}\label{sec:pf_lemma}
\subsection{Proof of Lemma \ref{lem:bound_grad}}
As a sketch of the proof, we need  the following two:
\begin{enumerate}
    \item The event $\{||\nabla L_n(\theta^*;\omega^*)||_\infty \le \lambda/2\}$ holds with high probability.

    \item $||\nabla L_n(\theta^*;\hat{\omega}) - \nabla L_n(\theta^*;\omega^*)||_\infty$ is small enough.
\end{enumerate}

First, observe that the gradient of $L_n$ with respect to $\nu_{jk}$ is given by
$$
\nabla_{jk} L_n(\theta;\omega)
=-\dfrac{1}{np}\Big(
f_j \gamma_{jk}(\omega) \sum_{h,i} r_{ij}^{(h)}(\theta,\omega) Y_{ik}^{(h)} +
f_k \gamma_{kj} (\omega) \sum_{h,i} r_{ik}^{(h)}(\theta,\omega) Y_{ij}^{(h)}
\Big), \quad j > k,
$$
where $r_{ij}^{(h)}=r_{ij}^{(h)}(\theta,\omega) := Y^{(h)}_{ij} -  \sum_{\ell:\ell\neq j}  \gamma_{j\ell}(\omega)\nu_{j\ell} Y^{(h)}_{i\ell}$.
Note that $\widehat{\Sigma}_{jk} = \sum_{h,i}Y_{ij}^{(h)} Y_{ik}^{(h)}/(np)$.
Then, we can derive
\begin{equation}\label{eq:sum_rY}
\begin{array}{rcl}
\dfrac{1}{np}\Big|\sum_{h,i} r_{ij}^{(h)}(\theta^*,\omega^*) Y_{ik}^{(h)} \Big|
&=&
\Big| \widehat{\Sigma}_{jk} - \sum\limits_{\ell:\ell\neq j}  \gamma_{j\ell}(\omega^*) \nu^*_{j\ell} \widehat{\Sigma}_{\ell k}
\Big|
\\
&=&
\Big| (\widehat{\Sigma}_{jk} - V_{jk})  - \sum\limits_{\ell:\ell\neq j}  \gamma_{j\ell}(\omega^*) \nu^*_{j\ell} (\widehat{\Sigma}_{\ell k} - V_{\ell k})
\Big|
\\
&\le &
\big| \widehat{\Sigma}_{jk} - V_{jk} \big| +
\max\limits_{\ell:\ell\neq j} \gamma_{j\ell}(\omega^*)
\max\limits_{\ell:\ell\neq j} \big|\widehat{\Sigma}_{\ell k} - V_{\ell k}
\big|
\sum\limits_{\ell:\ell\neq j} |\nu^*_{j\ell}| \\
&\le &
\big(1 +
\max\limits_{\ell:\ell\neq j} \gamma_{j\ell}(\omega^*)
\sum\limits_{\ell:\ell\neq j} |\nu^*_{j\ell}| \big) \max\limits_{\ell} \big|\widehat{\Sigma}_{\ell k} - V_{\ell k}
\big|,
\end{array}
\end{equation}
where we use $\Sigma_{jk}  = \sum\limits_{\ell:\ell\neq j}  \gamma_{j\ell}(\omega^*) \nu^*_{j\ell} \Sigma_{\ell k}$ in the second equality.

Using Assumption \ref{assum:bd_w}, we see that $\max_{j,k} \gamma_{jk}(\omega^*) \le B_\omega$, i.e. bounded from above.
Moreover, remark that $\sum\limits_{\ell\neq j} |\nu^*_{j\ell}| \le |S|$ as $|\nu^*_{jk}| \le 1$.
Combining all together, we can get the upper bound of the gradient vector:
$$
||\nabla L_n(\theta^*)||_\infty \le
2 (\max_{j} f_j) (B_\omega)^2 (1 + |S| ) \max\limits_{\ell,k} \big|\widehat{\Sigma}_{\ell k} - V_{\ell k}
\big|.
$$

Second, observe the difference of two terms below evaluated at $\omega^*$ and $\hat{\omega}$, respectively:
$$
\begin{array}{l}
\dfrac{1}{np}f_j\Big(
 \gamma_{jk}(\hat{\omega}) \sum_{h,i} r_{ij}^{(h)}(\theta^*,\hat{\omega}) Y_{ik}^{(h)}
- \gamma_{jk}(\omega^*) \sum_{h,i} r_{ij}^{(h)}(\theta^*,\omega^*) Y_{ik}^{(h)} \Big)\\
=
\dfrac{1}{np}f_j\Big\{
 \gamma_{jk}(\hat{\omega}) \sum_{h,i} \big(r_{ij}^{(h)}(\theta^*,\hat{\omega}) - r_{ij}^{(h)}(\theta^*,\omega^*)\big) Y_{ik}^{(h)}
\\
\qquad + \big(\gamma_{jk}(\hat{\omega}) - \gamma_{jk}(\omega^*)\big) \sum_{h,i} r_{ij}^{(h)}(\theta^*,\omega^*) Y_{ik}^{(h)} \Big\}\\
=: A_1 + A_2
\end{array}
$$
We can get an upper bound of $A_1$ as follows.
$$
\begin{array}{rcl}
|A_1| &=&
\dfrac{1}{np}f_j
 \gamma_{jk}(\hat{\omega})
 \Big|\sum_{h,i}\sum_{\ell:\ell \neq j} \nu_{j\ell}^* (\gamma_{j\ell}(\hat{\omega}) - \gamma_{j\ell}(\omega^*)) Y_{i\ell}^{(h)} Y_{ik}^{(h)} \Big|\\
&=& \dfrac{1}{np}f_j
 (\gamma_{jk}(\hat{\omega}) - \gamma_{jk}(\omega^*) + \gamma_{jk}(\omega^*))
 \Big|\sum_{h,i}\sum_{\ell:\ell \neq j} \nu_{j\ell}^* (\gamma_{j\ell}(\hat{\omega}) - \gamma_{j\ell}(\omega^*)) Y_{i\ell}^{(h)} Y_{ik}^{(h)} \Big|\\
&\le& f_j
\big( |\gamma_{jk}(\hat{\omega}) - \gamma_{jk}(\omega^*)| + \gamma_{jk}(\omega^*) \big)
 \sum\limits_{\ell:\ell \neq j} |\nu_{j\ell}^*|
 \cdot \max\limits_{\ell:\ell \neq j} |\gamma_{j\ell}(\hat{\omega}) - \gamma_{j\ell}(\omega^*)|
 \cdot \max\limits_{\ell:\ell \neq j}\big|\widehat{\Sigma}_{\ell k}
\big|
\end{array}
$$
Using \eqref{eq:sum_rY}, we can get the upper bound of $|A_2|$:
$$
|A_2| \le f_j
 \big|\gamma_{jk}(\hat{\omega}) - \gamma_{jk}(\omega^*)\big|
 \big(1 +
\max\limits_{\ell:\ell\neq j} \gamma_{j\ell}(\omega^*)
\sum\limits_{\ell:\ell\neq j} |\nu^*_{j\ell}| \big) \max\limits_{\ell} \big|\widehat{\Sigma}_{\ell k}
\big|.
$$
Then, we get
$$
|A_1|  + |A_2|  \le
f_j B_\omega
\big( |\gamma_{jk}(\hat{\omega}) - \gamma_{jk}(\omega^*)| + B_\omega \big)
 \max_{\ell:\ell \neq j} |\gamma_{j\ell}(\hat{\omega}) - \gamma_{j\ell}(\omega^*)|
 \big(1 +
\sum\limits_{\ell:\ell\neq j} |\nu^*_{j\ell}| \big) \max\limits_{\ell} \big|\widehat{\Sigma}_{\ell k}
\big|.
$$
Now, we derive that
$$
\begin{array}{l}
||\nabla L_n(\theta^*;\hat{\omega}) - \nabla L_n(\theta^*;\omega^*)||_\infty\\
\qquad \le
2 (\max_j f_j) B_\omega
\big( \max\limits_{j\neq k} |\gamma_{jk}(\hat{\omega}) - \gamma_{jk}(\omega^*)| + B_\omega \big)
\\
\qquad\quad {}\times \max\limits_{\ell \neq j} |\gamma_{j\ell}(\hat{\omega}) - \gamma_{j\ell}(\omega^*)|
 \big(1 +
\sum\limits_{\ell\neq j} |\nu^*_{j\ell}| \big) \max\limits_{\ell,k} \big|\widehat{\Sigma}_{\ell k}
\big|
\end{array}
$$
\begin{lemma}
Suppose Assumption \ref{assum:bd_w} holds, and let $\epsilon_0= 1/(2B_\omega)$.
For $\hat{\omega} \in N_{\omega^*}(\epsilon_0)$, we have
$$
|\gamma_{j\ell}(\hat{\omega}) - \gamma_{j\ell}(\omega^*)| \le 2(B_\omega)^2||\hat{\omega} - \omega^*||_\infty, \quad \forall j,k
$$
\end{lemma}
\begin{proof}
As $\omega \mapsto \gamma_{j\ell}(\omega)$ is differentiable, the mean-value theorem yields that there exists $t \in (0,1)$ such that
$$
\gamma_{j\ell}(\hat{\omega}) - \gamma_{j\ell}(\omega^*) =
\nabla \gamma_{j\ell}(\tilde{\omega})^\top (\hat{\omega} - \omega^*),
$$
where $\tilde{\omega} = t\hat{\omega} + (1-t) \omega^*$. The norm of the gradient vector is
$$
||\nabla \gamma_{j\ell}(\omega)^\top||_1 \le 0.5(\omega_{jj}^{-3/2}\omega_{kk}^{1/2} + \omega_{jj}^{-1/2}\omega_{kk}^{-1/2} ).
$$
Since $B_\omega^{-1}\le \omega_{jj} \le B_\omega$ and $\hat{\omega} \in N_{\omega^*}(\epsilon_0)$, we have $1/(2B_\omega) \le \hat{\omega}_{jj} \le 2B_\omega$, so as $1/(2B_\omega) \le \tilde{\omega}_{jj} \le 2B_\omega$, $\forall j$.
Consequently, $||\nabla \gamma_{j\ell}(\tilde{\omega})^\top||_1 \le 3 (B_\omega)^2$.
Then, we get
$$
|\gamma_{j\ell}(\hat{\omega}) - \gamma_{j\ell}(\omega^*)| \le
||\nabla \gamma_{j\ell}(\tilde{\omega})||_1 \cdot ||\hat{\omega} - \omega^*||_\infty
\le
3(B_\omega)^2||\hat{\omega} - \omega^*||_\infty
$$
\end{proof}

Using this lemma to the last inequality, we get for some numerical constant $C>0$
$$
\begin{array}{l}
||\nabla L_n(\theta^*;\hat{\omega}) - \nabla L_n(\theta^*;\omega^*)||_\infty\\
\qquad \le
C (\max_j f_j) (B_\omega)^4\big( 3B_\omega||\hat{\omega} - \omega^*||_\infty   + 1 \big) (1 + |S|)
||\hat{\omega} - \omega^*||_\infty
  \max\limits_{\ell,k} \big|\widehat{\Sigma}_{\ell k}
\big|.
\end{array}
$$

Now, we combine the two results:
$$
\begin{aligned}
&||\nabla L_n(\theta^*;\hat{\omega} )||_\infty\\
&\quad\le ||\nabla L_n(\theta^*;\hat{\omega}) - \nabla L_n(\theta^*; \omega^* )||_\infty
 + ||\nabla L_n(\theta^*; \omega^* )||_\infty\\
&\quad\le C' (\max_{j} f_j) (1 + |S| ) (B_\omega)^2
 \max\limits_{\ell,k} \big|\widehat{\Sigma}_{jk} - V_{jk}\big|\\
&\qquad + C' (\max_j f_j) (1 + |S|) (B_\omega)^4
 \big(3B_\omega||\hat{\omega} - \omega^*||_\infty + 1\big)\\
&\qquad\quad {}\times \max\limits_{\ell,k}\big|\widehat{\Sigma}_{\ell k}\big|
 \cdot ||\hat{\omega} - \omega^*||_\infty,
\end{aligned}
$$
where $C'>0$ is a numerical constant.
Using the triangular inequality
$$
\|\widehat{\Sigma}\|_{\max} \le \|\widehat{\Sigma}-V\|_{\max} + \|V\|_{\max} =
\|\widehat{\Sigma}-V\|_{\max} + \max_{1\le j\le q} V_{jj},
$$
we conclude the proof.

\subsection{Proof of Lemma \ref{lem:pd}}

A tedious calculation of the second derivative $\nabla_{j_2 k_2} \nabla_{j_1k_1} L_n(\theta;\omega)$, $j_2 >k_2, j_1> k_1$, of $L_n$ leads to the following:
\begin{enumerate}
    \item[\underline{Case 1}:] $j_1=j_2(=:d)$, $k_1 \neq k_2$, and $\max(k_1, k_2) < d$.
    $$
    \nabla_{j_2 k_2} \nabla_{j_1k_1} L_n(\theta;\omega) = f_d \gamma_{dk_1}\gamma_{dk_2} \widehat{\Sigma}_{k_1 k_2}.
    $$

    \item[\underline{Case 2}:] $j_1=j_2(=:j)$, $k_1 = k_2(=:k)$
    $$
    \nabla_{j_2 k_2} \nabla_{j_1k_1} L_n(\theta;\omega) = f_j (\gamma_{jk})^2\widehat{\Sigma}_{kk} + f_k(\gamma_{kj})^2\widehat{\Sigma}_{jj}.
    $$

    \item[\underline{Case 3}:] $k_1=k_2(=:d)$, $j_1 \neq j_2$
    $$
    \nabla_{j_2 k_2} \nabla_{j_1k_1} L_n(\theta;\omega) = f_d \gamma_{dj_1}\gamma_{dj_2} \widehat{\Sigma}_{j_1 j_2}.
    $$

    \item[\underline{Case 4}:] $j_1=k_2(=:d)$ or $j_2=k_1(=:d)$
    $$
    \nabla_{j_2 k_2} \nabla_{j_1k_1} L_n(\theta;\omega) = f_d \gamma_{dj_2}\gamma_{dk_1} \widehat{\Sigma}_{k_1 j_2}.
    $$

    \item[\underline{Case 5}:] $j_1, j_2, k_1, k_2$ are all distinct
    $$
    \nabla_{j_2 k_2} \nabla_{j_1k_1} L_n(\theta;\omega) = 0.
    $$
\end{enumerate}
Thus, the quadratic form of the Hessian matrix is
$$
\Delta^\top \nabla^2 L_n(\theta;\omega) \Delta = \sum_{d \in [q]}\sum_{c_1, c_2: c_1 \neq d, c_2 \neq d} f_d \gamma_{dc_1}\Delta_{d c_1}\gamma_{dc_2} \Delta_{d c_2}\widehat{\Sigma}_{c_1 c_2}
=
\sum_{d \in [q]} f_d \tilde{\Delta}_d^\top \widehat{\Sigma}_{-d,-d}\tilde{\Delta}_d,
$$
where $\tilde{\Delta}_d = (\Delta_{d c}\gamma_{dc}; c \neq d)$. Since $\widehat{\Sigma}_{-d,-d}$ is the principal submatrix of $\widehat{\Sigma}$ and the smallest eigenvalue of $\widehat{\Sigma}$ is strictly larger than a certain constant $\eta >0$ by Lemma \ref{lem:norm_spectral}, we can derive that with probability at least $1-2\exp\{
-(c\eta^6np)/16
+q\log 21\}$
$$
\sum_{d \in [q]} f_d \tilde{\Delta}_d^\top \widehat{\Sigma}_{-d,-d}\tilde{\Delta}_d \ge 0.5 \eta \sum_{d \in [q]} f_d  ||\tilde{\Delta}_d ||_2^2 \ge
0.5 \cdot \eta \cdot \min_j f_j   ||\Delta ||_2^2 >
0.
$$
Let $\kappa = 0.5 \cdot \eta \cdot \min_j f_j$, which concludes the proof.

\subsection{Concentration inequalities}

\subsubsection{Convergence rate of the covariance estimator}

The following lemma specifies the rate of convergence of the covariance estimator in terms of the element-wise maximum norm.
\begin{lemma}\label{lem:rate_Sigma_elemmax}
Suppose $\text{tr}(U)=p$.
It holds with probability at least $1 -3/{(p\vee q)^2}$
$$
\|\widehat{\Sigma}-V\|_{\max}
\le
C
\left(\max_{1\le j\le q} V_{jj}\right)
\frac{\|U\|_{F}}{p}
\sqrt{\frac{\log(p\vee q)}{n}}.
$$
\end{lemma}
\noindent
\begin{proof}
The following concentration is a straightforward result from Lemma 4.3 of \citet{Zhou:2014}:
$$
\Pr\left(
\|\widehat{\Sigma}-V\|_{\max}
\le
C
\left(\max_{1\le j\le q} V_{jj}\right)
\frac{\|U\|_{F}}{p}
\sqrt{\frac{\log(p\vee q)}{n}}
\right)
\ge
1-\frac{3}{(p\vee q)^2}.
$$
Here, $C>0$ is a numerical constant.

\end{proof}

We also presents the convergence rate of the covariance estimator in terms of the spectral norm.
\begin{lemma}\label{lem:norm_spectral}
Suppose Assumption \ref{assum:eigenvalue}. For some numerical constant $c>0$,
if $(c\eta^6np)/16-(q\vee \log n)\log 21 > 0$, it holds with probability at least $1-2\exp\{
-(c\eta^6np)/16
+q\log 21\}$ that
$$
||\widehat{\Sigma} - V||_2 \le \eta/2.
$$
With the same probability, this implies that
$$
\lambda_{\min}(\widehat{\Sigma}) \ge \eta/2.
$$
\end{lemma}
\begin{proof}
To prove the first concentration inequality, we use the result from \citet{sun2026highdimensionalconvergencerates}, which we copy here for completeness.
\begin{lemma}[Equation (75) in the proof of Theorem 4.1 of \citet{sun2026highdimensionalconvergencerates}]
Suppose Assumption \ref{assum:eigenvalue}.
For any $\widetilde C_0>
\frac{2}{\eta^2}
\sqrt{\frac{\log 21}{c}}$ ($c$: a numerical constant)
we have
\[
P\left(
\left\|\widehat{\Sigma}-V\right\|_2
>
\widetilde C_0
\sqrt{\frac{q\vee\log n}{np}}
\right)
\le
2\exp\left\{
-\frac{c\eta^4\widetilde C_0^2}{4}(q\vee\log n)
+q\log 21
\right\}.
\]
\end{lemma}
\noindent
We choose $\widetilde C_0$ such that
$$
\widetilde C_0
\sqrt{\frac{q \vee \log n}{np}}
=
\frac{\eta}{2}.
$$
Then, the choice yields $
||\widehat{\Sigma} - V||_2 > \eta/2$,
with probability at most $2\exp\{
-(c \eta^6np)/16
+q\log 21\}$.

The second conclusion of the lemma is a straightforward application of the derived result.
Using Weyl's inequality, we know that
$|\lambda_{\min}(\widehat{\Sigma}) - \lambda_{\min}(V)| \le \eta/2$.
Then,
$$
\lambda_{\min}(\widehat{\Sigma}) \ge \lambda_{\min}(\Sigma) - \eta/2 \ge \eta/2.
$$
\end{proof}

\subsubsection{Convergence rate of the conditional variance estimators}\label{sec:suitable_w}

Define
$$
\tilde{\beta}= \arg\min_{\beta \in \mathbb{R}^{q-1}}(np)^{-1}\sum^n_{h=1}\sum^p_{i=1}\big(Y^{(h)}_{ij}-\sum_{k\neq j} Y^{(h)}_{ik}\beta_k \big)^2
$$
The proposed estimator for $\omega_{jj}$ is $\tilde{\omega}_{jj}=1/\tilde{\phi}$, where
$$
\tilde{\phi}:= (np)^{-1}\sum^n_{h=1}\sum^p_{i=1}\big(Y^{(h)}_{ij}-\sum_{k\neq j} Y^{(h)}_{ik}\tilde{\beta}_k \big)^2.
$$
The following lemma presents the convergence rate of $\tilde{\omega}_{jj}$.
\begin{lemma}\label{lem:rate_w}
Suppose $\text{tr}(U)=p$ and Assumption \ref{assum:bd_w}. Also, assume $np > q$ and $R_\phi = o(1)$, where
$$
R_\phi= c \max\left\{\dfrac{||U||_F}{p} \sqrt{\dfrac{\log(p\vee q)^\alpha}{n}},
\dfrac{||U||_2}{p} \cdot \dfrac{q \vee \log(p\vee q)^\alpha}{n} \right\},
$$
where $c>0$ is a numerical constant.
For $\alpha > 1$, we have with probability at least $1 - 2/(p \vee q)^{\alpha-1}$
$$
\max_j |\tilde{\omega}_{jj} - \omega_{jj}|  \le C B_\omega R_\phi,
$$
where $C>0$ is a numerical constant.

\end{lemma}

\begin{proof}
    Note that
$\tilde{\beta} = (Z^\top Z)^{-1} Z^\top W$
where $Z \in \mathbb{R}^{np \times (q-1)}$ denotes the matrix binding $\{Y^{(h)}\}_h$ in rows after deleting $j$-th columns and $
W\in \mathbb{R}^{np}$ denotes a column vector binding all $j$-th columns of $\{Y^{(h)}\}_h$. If $np > q$, the ordinary least square estimator is well-defined.
More concisely, it can be written as $\tilde{\phi}= (np)^{-1}W^\top (I - P_Z) W$, where $P_Z = Z(Z^\top Z)^{-1} Z^\top$.
Using the Gaussianity, we can express $W$ by a regression form:
$$
W = Z \beta + \epsilon,
$$
where $\beta \in \mathbb{R}^{q-1}$ is a deterministic parameter (as a function of $V$ and $U$) and $\epsilon \sim N(0,\phi \cdot I_n \otimes U)$ with $\phi = 1/\omega_{jj}^V$. Remark that $Z$ and $\epsilon$ are independent. Therefore, we get
$$
\tilde{\phi}= (np)^{-1}\epsilon^\top (I - P_Z) \epsilon
$$
because of $(I-P_Z)Z=0$. As $\epsilon \overset{\text{d}}{=} \sqrt{\phi} \cdot I_n \otimes U^{1/2}\tilde{\epsilon}$ with every entry of $\tilde{\epsilon}$ being independent standard Gaussian variable, $\tilde{\phi} \overset{\text{d}}{=} \tilde{\epsilon}^\top A \tilde{\epsilon}$ where $A = \phi (I_n \otimes U^{1/2}) (I-P_Z) \cdot (I_n \otimes U^{1/2}) / (np)$ and $U^{1/2}$ is the square-root matrix.

Let $\bar{\phi}(Z) := \mathbb{E}[\tilde{\phi}|Z]$. Conditioning on $Z$, the Hanson-Wright inequality describes the concentration of $\tilde{\phi}$:
$$
P[
|\tilde{\phi} - \bar{\phi}(Z) | > t
] \le 2 \exp\left\{
-c \min \Big(
\dfrac{t^2}{||A||_F^2}, \dfrac{t}{||A||_2}
\Big)
\right\}, \quad t>0,
$$
where $c>0$ is a numerical constant. Using the property of the projection matrix, we get
$$
||A||_F^2 \le \phi^2 ||I_n \otimes U||_F^2 / (np)^2 \le \phi^2 ||U||_F^2 / (np^2),
$$
and
$$
||A||_2 \le \phi ||I_n \otimes U||_2 / np = \phi ||U||_2 / (np).
$$
Thus, the concentration inequality is presented as
$$
P[
|\tilde{\phi} - \bar{\phi}(Z) | > t
] \le 2 \exp\left\{
-c \min \Big(
\dfrac{np^2 t^2}{\phi^2 ||U||_F^2}, \dfrac{npt}{\phi ||U||_2}
\Big)
\right\}.
$$
Since the upper bound does not depend on $Z$, the conditioning can be removed, and thus the same probability bound holds unconditionally.
For $\alpha>0$, putting
$$
t = c^{-1} \cdot \phi \max\{||U||_F \sqrt{\log(p\vee q)^\alpha}/(\sqrt{n}p), ||U||_2 \log(p\vee q)^\alpha/ (np)\},
$$
we get
$$
|\tilde{\phi} - \bar{\phi}(Z) | > c_2 \cdot \phi \max\left\{\dfrac{||U||_F}{p} \sqrt{\dfrac{\log(p\vee q)^\alpha}{n}},
\dfrac{||U||_2}{p} \dfrac{\log(p\vee q)^\alpha}{n} \right\},
$$
with probability at most $2/(p \vee q)^\alpha$ for some numerical constant $c_2>0$.

It is easy to check that
$$
\bar{\phi}(Z) = \phi \big\{1 - \text{tr}(P_Z (I_n \otimes U) / (np)) \big\},
$$
which means that $\tilde{\phi}$ is a biased estimator of $\phi$. Note that the bias term tends to zero if $q=o(np)$ because
$$
\phi \cdot \text{tr}(P_Z (I_n \otimes U) / (np))
\le \phi \cdot ||P_Z ||_* || I_n \otimes U||_2 / (np)
\le
\phi \cdot q || U||_2 / (np) \to 0,
$$
where the first inequality uses the Schatten-norm H\"{o}lder inequality and $||\cdot||_*$ is the nuclear norm of a matrix.
Then, we get with probability at least $1-2/(p \vee q)^\alpha$ that
$$
|\tilde{\phi} - \phi| / \phi \le
|\tilde{\phi} - \bar{\phi}(Z)| / \phi + |\bar{\phi}(Z) - \phi| / \phi \le
 R_\phi.
$$
The rate $R_\phi$ is defined by
$$
R_\phi= c_3 \max\left\{\dfrac{||U||_F}{p} \sqrt{\dfrac{\log(p\vee q)^\alpha}{n}},
\dfrac{||U||_2}{p} \cdot \dfrac{q \vee \log(p\vee q)^\alpha}{n} \right\},
$$
where $c_3>0$ is a numerical constant.
Taking union bound argument across different $j \in [q]$, we have
$\max_{j \in [q]} |\tilde{\phi}_{jj} - \phi_{jj} |/\phi_{jj} \le R_\phi$ with probability at least $1 - 2/(p \vee q)^{\alpha-1}$.

Finally, using
$$
|1/x-1/y| \le \dfrac{|x-y|/y^2}{1 - |x-y|/y}, \quad \text{if } x,y>0, |x-y|/y <1,
$$
we have
$$
|\tilde{\omega}_{jj} - \omega_{jj}| = |1/\tilde{\phi}_{jj} - 1/\phi_{jj}|
\le \dfrac{1}{\phi_{jj}} \cdot \dfrac{|\tilde{\phi}_{jj} - \phi_{jj}| /\phi_{jj}}{1 - |\tilde{\phi}_{jj} - \phi_{jj}|/\phi_{jj}}.
$$
If $np$ is sufficiently large so that $R_\phi < 1/2$, then $|\tilde{\omega}_{jj} - \omega_{jj}| \le 2 |\tilde{\phi}_{jj} - \phi_{jj}|/\phi_{jj}^2$.
Thus, we get with probability at least $1-2/(p \vee q)^\alpha$
$$\max_j |\tilde{\omega}_{jj} - \omega_{jj}| \le \max_j \dfrac{2}{\phi_{jj}} \cdot \max_j \dfrac{|\tilde{\phi}_{jj} - \phi_{jj}|}{\phi_{jj}} \le 2B_\omega R_\phi,
$$
where the last inequality uses Assumption \ref{assum:bd_w}.

\end{proof}

\section{Construction of a precision matrix}\label{sec:supp_generating}

We first generate an adjacency matrix $\Omega=(\omega_{ij})_{i,j\in[q]}$ as follows, depending on each network structure.
\begin{itemize}
    \item \textbf{Band structure.}
    We generate a banded precision matrix structure with band width $d$.
    In our experiments, we set $d=2$, meaning each variable is connected to its two nearest neighbors on both sides. The precision matrix entries are defined as
    $$\omega_{ij} =
    \begin{cases}
    1, & i = j\\
    \rho_{\text{band}}, & 1\leq |i-j|< d\\
    0.5\cdot\rho_{\text{band}}, & |i-j|=d\\
    0, & \text{otherwise},
    \end{cases}$$
    where $\rho_{\text{band}}\sim \text{Unif}(0.2,0.6)$

    \item \textbf{Cluster structure.}
    Connections are allowed only within the same cluster, and no edges are allowed between variables from different clusters. Within each cluster, the adjacency structure is generated as $a_{ij}\sim \text{Ber}(0.1)$.
    Given the adjacency, the precision matrix elements follow
    $$\omega_{ij} = \begin{cases}
    1, & i = j,\\
    \rho_{\text{cluster}}, & i \neq j,\ a_{ij} = 1,\\
    0, & \text{otherwise},
    \end{cases}$$
    where $\rho_{\text{cluster}}\sim \text{Unif}(0.3, 0.5)$.

    \item \textbf{Hub structure}.
    For the hub structure, the variables are divided into blocks of size 10. Within each block, the first variable acts as a hub and is connected to all remaining variables in the same block.
    Let $k=1,\dots,\lfloor p/10\rfloor$ index the blocks, and define the hub index in block $k$ as $h_k=10(k-1)+1$.
    The nonzero precision values are drawn from
    $$\omega_{ij} = \begin{cases}
    1, & i = j, \\
    \rho_{\text{hub}},& i = h_k \text{ or } j = h_k,\ \text{with } i,j \text{ in block } k, \\
    0, & \text{otherwise} ,
    \end{cases}$$
    where $\rho_{\text{hub}}\sim \text{Unif}(0.2, 0.6)$.

    \item \textbf{Random structure.}
    For the random network, we adopt an Erd\H{o}s--R\'enyi random graph model. The adjacency is generated as $a_{ij}\sim \text{Ber}(q)$, $q=\min(0.05, 5/p)$.
    For every existing edge, we assign a fixed nonzero precision value
    $$\omega_{ij} = \begin{cases}
    1, & i = j,\\
    \rho_{\text{random}}, & i \neq j,\ a_{ij} = 1,\\
    0, & \text{otherwise},
    \end{cases}$$
    where $\rho_{\text{random}}=0.4$
\end{itemize}
Next, we define $\tilde{\Omega}=\Omega+(0.05+|\mu_{\min}(\Omega)|)\cdot I_q$, where $\mu_{\min}(\Omega)$ denotes the smallest eigenvalue of $\Omega$, which makes $\tilde{\Omega}$ be positive definite. Finally, all diagonal entries of $\tilde{\Omega}$ are then multiplied by $\text{Unif}(1,5)$.

\section{Additional results for simulation study}\label{sec:supp_sim}

This section reports the simulation results omitted from Section~\ref{sec:results} due to space limitations.
We present results for all four network structures of $V^{-1}$ (band, cluster, hub, and random), with $U^{-1}$ fixed at the band structure throughout, evaluated separately on the estimated $\widehat{U^{-1}}$ and $\widehat{V^{-1}}$.
The overall pattern is consistent with that reported in Section~\ref{sec:results}: matSPACE-sw attains the highest AUC and MCC in nearly all settings, and its advantage over the competing methods persists across network structures.

\subsection{Band structure}\label{sec:supp_band}
Tables~\ref{tab:band_U} and \ref{tab:band_V} report the full results for the band structure.

\subsection{Cluster structure}\label{sec:supp_cluster}
Tables~\ref{tab:cluster_U} and \ref{tab:cluster_V} report results for the cluster structure.
The overall pattern is similar to that under the band structure, with matSPACE-sw continuing to achieve the highest AUC and MCC in most settings.

\subsection{Hub structure}\label{sec:supp_hub}
Tables~\ref{tab:hub_U} and \ref{tab:hub_V} present results for the hub structure.
The overall pattern is similar to the other structures, with matSPACE-sw achieving the highest AUC throughout.
For MCC, however, matSPACE-dew attains a higher value than matSPACE-sw when $n=10$ and $p=20$ (for both $q=20$ and $q=50$), evaluated on $\widehat{V^{-1}}$, whereas matSPACE-sw remains higher in all other settings.

\subsection{Random structure}\label{sec:supp_random}
Tables~\ref{tab:random_U} and \ref{tab:random_V} report results for the random structure. matSPACE-sw again achieves the highest AUC and MCC in most settings, and the gap over GEMINI and matNS narrows only slightly compared with the band structure. As in Section~\ref{sec:results}, increasing $n$ improves all methods' performance while matSPACE maintains a comparatively stable FPR.

\clearpage
\begin{sidewaystable}
\centering
\spacingset{1}
\caption{Simulation results based on 300 data replications, where $U^{-1}$ is fixed at the band structure and $V^{-1}$ is fixed at the band structure, evaluated on the estimated $\widehat{U^{-1}}$.
Each entry reports an average of each metric with the standard error in parentheses. Boldface indicates the best-performing method.
}
\label{tab:band_U}
\renewcommand{\arraystretch}{0.85}
\tiny
\setlength{\tabcolsep}{4pt}
\begin{minipage}{0.49\textwidth}
\centering
\resizebox{\textwidth}{!}{%
\begin{tabular}{clcccc}
\toprule
$n$ & Method & TPR & FPR & MCC & AUC \\
\midrule
\multirow{28}{*}{10} & \multicolumn{5}{l}{\textit{$p=20$, $q=20$}} \\
 & matSPACE-sw & \textbf{0.405 (0.010)} & 0.136 (0.006) & \textbf{0.287 (0.005)} & \textbf{0.712 (0.003)} \\
 & matSPACE-dew & 0.259 (0.008) & 0.067 (0.004) & 0.257 (0.005) & 0.689 (0.003) \\
 & matNS-AND & 0.023 (0.002) & 0.001 (0.000) & 0.087 (0.006) & 0.701 (0.003) \\
 & matNS-OR & 0.082 (0.004) & 0.011 (0.001) & 0.169 (0.006) & 0.691 (0.003) \\
 & GEMINI & 0.318 (0.010) & 0.096 (0.005) & 0.267 (0.006) & 0.698 (0.003) \\
 & G-FDR-0.10 & 0.002 (0.000) & \textbf{0.000 (0.000)} & 0.010 (0.002) & 0.542 (0.002) \\
\cmidrule(lr){2-6}
 & \multicolumn{5}{l}{\textit{$p=20$, $q=50$}} \\
 & matSPACE-sw & 0.618 (0.008) & 0.148 (0.005) & \textbf{0.448 (0.005)} & \textbf{0.815 (0.002)} \\
 & matSPACE-dew & 0.430 (0.010) & 0.089 (0.005) & 0.391 (0.004) & 0.780 (0.002) \\
 & matNS-AND & 0.066 (0.003) & 0.001 (0.000) & 0.198 (0.006) & 0.796 (0.002) \\
 & matNS-OR & 0.182 (0.005) & 0.014 (0.001) & 0.312 (0.006) & 0.785 (0.002) \\
 & GEMINI & \textbf{0.643 (0.008)} & 0.199 (0.006) & 0.405 (0.004) & 0.794 (0.002) \\
 & G-FDR-0.10 & 0.038 (0.003) & \textbf{0.000 (0.000)} & 0.131 (0.007) & 0.654 (0.003) \\
\cmidrule(lr){2-6}
 & \multicolumn{5}{l}{\textit{$p=50$, $q=20$}} \\
 & matSPACE-sw & \textbf{0.291 (0.006)} & 0.065 (0.003) & \textbf{0.237 (0.003)} & \textbf{0.701 (0.002)} \\
 & matSPACE-dew & 0.166 (0.005) & 0.029 (0.001) & 0.199 (0.003) & 0.676 (0.002) \\
 & matNS-AND & 0.018 (0.001) & \textbf{0.000 (0.000)} & 0.096 (0.004) & 0.690 (0.002) \\
 & matNS-OR & 0.071 (0.002) & 0.009 (0.000) & 0.150 (0.003) & 0.678 (0.002) \\
 & GEMINI & 0.200 (0.007) & 0.035 (0.002) & 0.230 (0.004) & 0.688 (0.002) \\
 & G-FDR-0.10 & 0.000 (0.000) & \textbf{0.000 (0.000)} & 0.000 (0.000) & 0.529 (0.001) \\
\cmidrule(lr){2-6}
 & \multicolumn{5}{l}{\textit{$p=50$, $q=50$}} \\
 & matSPACE-sw & \textbf{0.502 (0.005)} & 0.071 (0.002) & \textbf{0.390 (0.003)} & \textbf{0.803 (0.002)} \\
 & matSPACE-dew & 0.281 (0.005) & 0.029 (0.001) & 0.325 (0.003) & 0.765 (0.001) \\
 & matNS-AND & 0.054 (0.002) & 0.001 (0.000) & 0.201 (0.003) & 0.783 (0.001) \\
 & matNS-OR & 0.145 (0.003) & 0.008 (0.000) & 0.274 (0.003) & 0.770 (0.001) \\
 & GEMINI & 0.490 (0.007) & 0.084 (0.003) & 0.367 (0.003) & 0.783 (0.001) \\
 & G-FDR-0.10 & 0.007 (0.001) & \textbf{0.000 (0.000)} & 0.051 (0.004) & 0.615 (0.002) \\
\midrule
\midrule
\multirow{14}{*}{40} & \multicolumn{5}{l}{\textit{$p=20$, $q=20$}} \\
 & matSPACE-sw & 0.728 (0.006) & 0.156 (0.004) & \textbf{0.518 (0.004)} & \textbf{0.864 (0.002)} \\
 & matSPACE-dew & 0.542 (0.011) & 0.106 (0.006) & 0.467 (0.004) & 0.827 (0.002) \\
 & matNS-AND & 0.110 (0.004) & 0.001 (0.000) & 0.274 (0.006) & 0.844 (0.002) \\
 & matNS-OR & 0.254 (0.006) & 0.014 (0.001) & 0.397 (0.005) & 0.833 (0.002) \\
 & GEMINI & \textbf{0.778 (0.006)} & 0.259 (0.006) & 0.437 (0.004) & 0.842 (0.002) \\
 & G-FDR-0.10 & 0.121 (0.005) & \textbf{0.000 (0.000)} & 0.296 (0.007) & 0.745 (0.003) \\
\cmidrule(lr){2-6}
 & \multicolumn{5}{l}{\textit{$p=20$, $q=50$}} \\
 & matSPACE-sw & 0.875 (0.004) & 0.147 (0.004) & \textbf{0.641 (0.005)} & \textbf{0.941 (0.001)} \\
 & matSPACE-dew & 0.813 (0.009) & 0.221 (0.011) & 0.539 (0.007) & 0.910 (0.002) \\
 & matNS-AND & 0.257 (0.006) & 0.002 (0.000) & 0.454 (0.006) & 0.923 (0.001) \\
 & matNS-OR & 0.453 (0.006) & 0.012 (0.001) & 0.588 (0.005) & 0.914 (0.002) \\
 & GEMINI & \textbf{0.931 (0.003)} & 0.338 (0.006) & 0.474 (0.005) & 0.923 (0.001) \\
 & G-FDR-0.10 & 0.414 (0.006) & \textbf{0.000 (0.000)} & 0.600 (0.005) & 0.878 (0.002) \\
\bottomrule
\end{tabular}
}
\end{minipage}%
\hfill
\begin{minipage}{0.49\textwidth}
\centering
\resizebox{\textwidth}{!}{%
\begin{tabular}{clcccc}
\toprule
$n$ & Method & TPR & FPR & MCC & AUC \\
\midrule
\multirow{14}{*}{40} & \multicolumn{5}{l}{\textit{$p=50$, $q=20$}} \\
 & matSPACE-sw & 0.621 (0.005) & 0.071 (0.002) & \textbf{0.478 (0.003)} & \textbf{0.859 (0.001)} \\
 & matSPACE-dew & 0.376 (0.007) & 0.034 (0.002) & 0.401 (0.003) & 0.817 (0.001) \\
 & matNS-AND & 0.095 (0.002) & 0.001 (0.000) & 0.276 (0.004) & 0.837 (0.001) \\
 & matNS-OR & 0.221 (0.003) & 0.009 (0.000) & 0.363 (0.003) & 0.824 (0.001) \\
 & GEMINI & \textbf{0.658 (0.006)} & 0.127 (0.004) & 0.407 (0.004) & 0.837 (0.001) \\
 & G-FDR-0.10 & 0.038 (0.002) & \textbf{0.000 (0.000)} & 0.171 (0.004) & 0.707 (0.002) \\
\cmidrule(lr){2-6}
 & \multicolumn{5}{l}{\textit{$p=50$, $q=50$}} \\
 & matSPACE-sw & 0.829 (0.003) & 0.076 (0.002) & \textbf{0.604 (0.004)} & \textbf{0.940 (0.001)} \\
 & matSPACE-dew & 0.573 (0.008) & 0.042 (0.003) & 0.545 (0.003) & 0.901 (0.001) \\
 & matNS-AND & 0.234 (0.004) & 0.001 (0.000) & 0.453 (0.004) & 0.919 (0.001) \\
 & matNS-OR & 0.421 (0.004) & 0.010 (0.000) & 0.551 (0.003) & 0.909 (0.001) \\
 & GEMINI & \textbf{0.874 (0.003)} & 0.181 (0.005) & 0.458 (0.005) & 0.921 (0.001) \\
 & G-FDR-0.10 & 0.240 (0.003) & \textbf{0.000 (0.000)} & 0.472 (0.003) & 0.850 (0.002) \\
\midrule
\midrule
\multirow{28}{*}{70} & \multicolumn{5}{l}{\textit{$p=20$, $q=20$}} \\
 & matSPACE-sw & 0.835 (0.004) & 0.153 (0.004) & \textbf{0.602 (0.004)} & \textbf{0.917 (0.002)} \\
 & matSPACE-dew & 0.719 (0.011) & 0.173 (0.008) & 0.520 (0.006) & 0.881 (0.002) \\
 & matNS-AND & 0.188 (0.006) & 0.001 (0.000) & 0.377 (0.006) & 0.896 (0.002) \\
 & matNS-OR & 0.376 (0.007) & 0.013 (0.001) & 0.518 (0.005) & 0.886 (0.002) \\
 & GEMINI & \textbf{0.894 (0.004)} & 0.320 (0.006) & 0.461 (0.004) & 0.896 (0.002) \\
 & G-FDR-0.10 & 0.287 (0.006) & \textbf{0.000 (0.000)} & 0.489 (0.006) & 0.835 (0.002) \\
\cmidrule(lr){2-6}
 & \multicolumn{5}{l}{\textit{$p=20$, $q=50$}} \\
 & matSPACE-sw & 0.938 (0.003) & 0.128 (0.004) & 0.714 (0.005) & \textbf{0.973 (0.001)} \\
 & matSPACE-dew & 0.916 (0.006) & 0.336 (0.015) & 0.515 (0.011) & 0.947 (0.001) \\
 & matNS-AND & 0.371 (0.006) & 0.001 (0.000) & 0.561 (0.005) & 0.958 (0.001) \\
 & matNS-OR & 0.586 (0.006) & 0.011 (0.001) & 0.693 (0.004) & 0.951 (0.001) \\
 & GEMINI & \textbf{0.976 (0.002)} & 0.368 (0.006) & 0.480 (0.005) & 0.959 (0.001) \\
 & G-FDR-0.10 & 0.624 (0.005) & \textbf{0.000 (0.000)} & \textbf{0.756 (0.004)} & 0.935 (0.002) \\
\cmidrule(lr){2-6}
 & \multicolumn{5}{l}{\textit{$p=50$, $q=20$}} \\
 & matSPACE-sw & 0.753 (0.004) & 0.071 (0.002) & \textbf{0.567 (0.004)} & \textbf{0.913 (0.001)} \\
 & matSPACE-dew & 0.490 (0.008) & 0.036 (0.002) & 0.495 (0.003) & 0.871 (0.001) \\
 & matNS-AND & 0.172 (0.003) & 0.001 (0.000) & 0.384 (0.004) & 0.891 (0.001) \\
 & matNS-OR & 0.337 (0.004) & 0.009 (0.000) & 0.481 (0.003) & 0.879 (0.001) \\
 & GEMINI & \textbf{0.804 (0.005)} & 0.159 (0.005) & 0.448 (0.005) & 0.892 (0.001) \\
 & G-FDR-0.10 & 0.143 (0.003) & \textbf{0.000 (0.000)} & 0.359 (0.004) & 0.800 (0.002) \\
\cmidrule(lr){2-6}
 & \multicolumn{5}{l}{\textit{$p=50$, $q=50$}} \\
 & matSPACE-sw & 0.902 (0.002) & 0.060 (0.002) & \textbf{0.691 (0.004)} & \textbf{0.971 (0.001)} \\
 & matSPACE-dew & 0.667 (0.008) & 0.037 (0.003) & 0.637 (0.004) & 0.939 (0.001) \\
 & matNS-AND & 0.358 (0.004) & 0.001 (0.000) & 0.573 (0.003) & 0.954 (0.001) \\
 & matNS-OR & 0.561 (0.004) & 0.008 (0.000) & 0.672 (0.002) & 0.946 (0.001) \\
 & GEMINI & \textbf{0.945 (0.002)} & 0.207 (0.005) & 0.470 (0.006) & 0.957 (0.001) \\
 & G-FDR-0.10 & 0.449 (0.004) & \textbf{0.000 (0.000)} & 0.653 (0.003) & 0.916 (0.001) \\
\bottomrule
\end{tabular}
}
\end{minipage}
\end{sidewaystable}

\begin{sidewaystable}
\centering
\spacingset{1}
\caption{Simulation results based on 300 data replications, where $U^{-1}$ is fixed at the band structure and $V^{-1}$ is fixed at the band structure, evaluated on the estimated $\widehat{V^{-1}}$.
Each entry reports an average of each metric with the standard error in parentheses. Boldface indicates the best-performing method.
}
\label{tab:band_V}
\renewcommand{\arraystretch}{0.85}
\tiny
\setlength{\tabcolsep}{4pt}
\begin{minipage}{0.49\textwidth}
\centering
\resizebox{\textwidth}{!}{%
\begin{tabular}{clcccc}
\toprule
$n$ & Method & TPR & FPR & MCC & AUC \\
\midrule
\multirow{28}{*}{10} & \multicolumn{5}{l}{\textit{$p=20$, $q=20$}} \\
 & matSPACE-sw & \textbf{0.392 (0.010)} & 0.133 (0.005) & \textbf{0.277 (0.005)} & \textbf{0.706 (0.003)} \\
 & matSPACE-dew & 0.262 (0.008) & 0.077 (0.004) & 0.245 (0.005) & 0.683 (0.003) \\
 & matNS-AND & 0.021 (0.002) & 0.001 (0.000) & 0.084 (0.005) & 0.695 (0.003) \\
 & matNS-OR & 0.079 (0.004) & 0.012 (0.001) & 0.160 (0.006) & 0.686 (0.003) \\
 & GEMINI & 0.310 (0.011) & 0.100 (0.005) & 0.245 (0.006) & 0.691 (0.003) \\
 & G-FDR-0.10 & 0.001 (0.000) & \textbf{0.000 (0.000)} & 0.006 (0.002) & 0.537 (0.002) \\
\cmidrule(lr){2-6}
 & \multicolumn{5}{l}{\textit{$p=20$, $q=50$}} \\
 & matSPACE-sw & \textbf{0.285 (0.006)} & 0.064 (0.002) & \textbf{0.231 (0.003)} & \textbf{0.694 (0.002)} \\
 & matSPACE-dew & 0.157 (0.004) & 0.028 (0.001) & 0.192 (0.002) & 0.669 (0.002) \\
 & matNS-AND & 0.017 (0.001) & \textbf{0.000 (0.000)} & 0.095 (0.004) & 0.682 (0.002) \\
 & matNS-OR & 0.066 (0.002) & 0.009 (0.000) & 0.140 (0.003) & 0.671 (0.002) \\
 & GEMINI & 0.184 (0.007) & 0.032 (0.002) & 0.219 (0.004) & 0.680 (0.002) \\
 & G-FDR-0.10 & 0.000 (0.000) & \textbf{0.000 (0.000)} & 0.000 (0.000) & 0.525 (0.001) \\
\cmidrule(lr){2-6}
 & \multicolumn{5}{l}{\textit{$p=50$, $q=20$}} \\
 & matSPACE-sw & 0.618 (0.008) & 0.154 (0.005) & \textbf{0.437 (0.004)} & \textbf{0.810 (0.002)} \\
 & matSPACE-dew & 0.423 (0.011) & 0.089 (0.005) & 0.387 (0.004) & 0.777 (0.002) \\
 & matNS-AND & 0.065 (0.003) & 0.001 (0.000) & 0.192 (0.006) & 0.792 (0.002) \\
 & matNS-OR & 0.174 (0.005) & 0.012 (0.001) & 0.307 (0.006) & 0.781 (0.002) \\
 & GEMINI & \textbf{0.647 (0.008)} & 0.206 (0.006) & 0.397 (0.005) & 0.790 (0.002) \\
 & G-FDR-0.10 & 0.037 (0.002) & \textbf{0.000 (0.000)} & 0.133 (0.007) & 0.652 (0.003) \\
\cmidrule(lr){2-6}
 & \multicolumn{5}{l}{\textit{$p=50$, $q=50$}} \\
 & matSPACE-sw & \textbf{0.489 (0.005)} & 0.067 (0.002) & \textbf{0.390 (0.003)} & \textbf{0.800 (0.002)} \\
 & matSPACE-dew & 0.275 (0.006) & 0.029 (0.001) & 0.319 (0.003) & 0.763 (0.002) \\
 & matNS-AND & 0.052 (0.002) & 0.001 (0.000) & 0.199 (0.003) & 0.781 (0.002) \\
 & matNS-OR & 0.146 (0.003) & 0.009 (0.000) & 0.268 (0.003) & 0.768 (0.002) \\
 & GEMINI & 0.483 (0.007) & 0.084 (0.003) & 0.361 (0.003) & 0.779 (0.002) \\
 & G-FDR-0.10 & 0.006 (0.000) & \textbf{0.000 (0.000)} & 0.045 (0.003) & 0.616 (0.002) \\
\midrule
\midrule
\multirow{14}{*}{40} & \multicolumn{5}{l}{\textit{$p=20$, $q=20$}} \\
 & matSPACE-sw & 0.731 (0.007) & 0.156 (0.004) & \textbf{0.521 (0.005)} & \textbf{0.864 (0.002)} \\
 & matSPACE-dew & 0.549 (0.011) & 0.113 (0.006) & 0.461 (0.005) & 0.828 (0.002) \\
 & matNS-AND & 0.105 (0.004) & 0.001 (0.000) & 0.269 (0.006) & 0.843 (0.002) \\
 & matNS-OR & 0.249 (0.006) & 0.013 (0.001) & 0.393 (0.006) & 0.833 (0.002) \\
 & GEMINI & \textbf{0.789 (0.006)} & 0.268 (0.006) & 0.434 (0.004) & 0.843 (0.002) \\
 & G-FDR-0.10 & 0.111 (0.004) & \textbf{0.000 (0.000)} & 0.284 (0.007) & 0.741 (0.003) \\
\cmidrule(lr){2-6}
 & \multicolumn{5}{l}{\textit{$p=20$, $q=50$}} \\
 & matSPACE-sw & 0.619 (0.005) & 0.071 (0.002) & \textbf{0.476 (0.003)} & \textbf{0.855 (0.001)} \\
 & matSPACE-dew & 0.353 (0.007) & 0.029 (0.001) & 0.396 (0.003) & 0.814 (0.001) \\
 & matNS-AND & 0.089 (0.002) & 0.001 (0.000) & 0.268 (0.003) & 0.833 (0.001) \\
 & matNS-OR & 0.215 (0.003) & 0.009 (0.000) & 0.358 (0.003) & 0.821 (0.001) \\
 & GEMINI & \textbf{0.643 (0.006)} & 0.118 (0.004) & 0.413 (0.004) & 0.833 (0.001) \\
 & G-FDR-0.10 & 0.036 (0.001) & \textbf{0.000 (0.000)} & 0.166 (0.004) & 0.703 (0.002) \\
\bottomrule
\end{tabular}
}
\end{minipage}%
\hfill
\begin{minipage}{0.49\textwidth}
\centering
\resizebox{\textwidth}{!}{%
\begin{tabular}{clcccc}
\toprule
$n$ & Method & TPR & FPR & MCC & AUC \\
\midrule
\multirow{14}{*}{40} & \multicolumn{5}{l}{\textit{$p=50$, $q=20$}} \\
 & matSPACE-sw & 0.877 (0.004) & 0.143 (0.004) & \textbf{0.649 (0.005)} & \textbf{0.940 (0.001)} \\
 & matSPACE-dew & 0.817 (0.009) & 0.240 (0.011) & 0.524 (0.007) & 0.907 (0.002) \\
 & matNS-AND & 0.248 (0.005) & 0.001 (0.000) & 0.447 (0.006) & 0.920 (0.002) \\
 & matNS-OR & 0.452 (0.006) & 0.013 (0.001) & 0.584 (0.005) & 0.912 (0.002) \\
 & GEMINI & \textbf{0.930 (0.003)} & 0.346 (0.006) & 0.466 (0.005) & 0.921 (0.001) \\
 & G-FDR-0.10 & 0.415 (0.006) & \textbf{0.000 (0.000)} & 0.601 (0.005) & 0.876 (0.002) \\
\cmidrule(lr){2-6}
 & \multicolumn{5}{l}{\textit{$p=50$, $q=50$}} \\
 & matSPACE-sw & 0.817 (0.004) & 0.072 (0.002) & \textbf{0.605 (0.004)} & \textbf{0.937 (0.001)} \\
 & matSPACE-dew & 0.550 (0.008) & 0.037 (0.002) & 0.541 (0.003) & 0.898 (0.001) \\
 & matNS-AND & 0.231 (0.003) & 0.001 (0.000) & 0.450 (0.004) & 0.916 (0.001) \\
 & matNS-OR & 0.413 (0.004) & 0.009 (0.000) & 0.550 (0.003) & 0.906 (0.001) \\
 & GEMINI & \textbf{0.865 (0.004)} & 0.178 (0.005) & 0.458 (0.005) & 0.918 (0.001) \\
 & G-FDR-0.10 & 0.237 (0.003) & \textbf{0.000 (0.000)} & 0.469 (0.003) & 0.844 (0.002) \\
\midrule
\midrule
\multirow{28}{*}{70} & \multicolumn{5}{l}{\textit{$p=20$, $q=20$}} \\
 & matSPACE-sw & 0.836 (0.005) & 0.150 (0.004) & \textbf{0.605 (0.005)} & \textbf{0.917 (0.002)} \\
 & matSPACE-dew & 0.727 (0.010) & 0.170 (0.008) & 0.527 (0.006) & 0.882 (0.002) \\
 & matNS-AND & 0.186 (0.005) & 0.002 (0.000) & 0.376 (0.006) & 0.897 (0.002) \\
 & matNS-OR & 0.378 (0.007) & 0.014 (0.001) & 0.516 (0.005) & 0.888 (0.002) \\
 & GEMINI & \textbf{0.893 (0.004)} & 0.322 (0.006) & 0.460 (0.005) & 0.897 (0.002) \\
 & G-FDR-0.10 & 0.287 (0.006) & \textbf{0.000 (0.000)} & 0.489 (0.006) & 0.833 (0.003) \\
\cmidrule(lr){2-6}
 & \multicolumn{5}{l}{\textit{$p=20$, $q=50$}} \\
 & matSPACE-sw & 0.748 (0.004) & 0.072 (0.002) & \textbf{0.562 (0.004)} & \textbf{0.910 (0.001)} \\
 & matSPACE-dew & 0.477 (0.008) & 0.035 (0.002) & 0.489 (0.003) & 0.868 (0.001) \\
 & matNS-AND & 0.165 (0.003) & 0.001 (0.000) & 0.377 (0.004) & 0.888 (0.001) \\
 & matNS-OR & 0.333 (0.004) & 0.009 (0.000) & 0.479 (0.003) & 0.877 (0.001) \\
 & GEMINI & \textbf{0.794 (0.005)} & 0.156 (0.005) & 0.449 (0.005) & 0.889 (0.001) \\
 & G-FDR-0.10 & 0.139 (0.003) & \textbf{0.000 (0.000)} & 0.354 (0.004) & 0.796 (0.002) \\
\cmidrule(lr){2-6}
 & \multicolumn{5}{l}{\textit{$p=50$, $q=20$}} \\
 & matSPACE-sw & 0.936 (0.003) & 0.127 (0.003) & 0.712 (0.004) & \textbf{0.971 (0.001)} \\
 & matSPACE-dew & 0.915 (0.006) & 0.344 (0.015) & 0.510 (0.011) & 0.945 (0.001) \\
 & matNS-AND & 0.372 (0.006) & 0.001 (0.000) & 0.560 (0.005) & 0.956 (0.001) \\
 & matNS-OR & 0.588 (0.006) & 0.011 (0.001) & 0.694 (0.004) & 0.949 (0.001) \\
 & GEMINI & \textbf{0.974 (0.002)} & 0.368 (0.006) & 0.479 (0.005) & 0.957 (0.001) \\
 & G-FDR-0.10 & 0.623 (0.005) & \textbf{0.000 (0.000)} & \textbf{0.756 (0.004)} & 0.932 (0.002) \\
\cmidrule(lr){2-6}
 & \multicolumn{5}{l}{\textit{$p=50$, $q=50$}} \\
 & matSPACE-sw & 0.900 (0.002) & 0.065 (0.002) & \textbf{0.678 (0.004)} & \textbf{0.969 (0.001)} \\
 & matSPACE-dew & 0.662 (0.008) & 0.038 (0.003) & 0.630 (0.004) & 0.938 (0.001) \\
 & matNS-AND & 0.355 (0.004) & 0.001 (0.000) & 0.571 (0.003) & 0.952 (0.001) \\
 & matNS-OR & 0.557 (0.004) & 0.008 (0.000) & 0.669 (0.002) & 0.945 (0.001) \\
 & GEMINI & \textbf{0.938 (0.002)} & 0.195 (0.005) & 0.480 (0.006) & 0.954 (0.001) \\
 & G-FDR-0.10 & 0.442 (0.004) & \textbf{0.000 (0.000)} & 0.648 (0.003) & 0.912 (0.001) \\
\bottomrule
\end{tabular}
}
\end{minipage}
\end{sidewaystable}

\begin{sidewaystable}
\centering
\spacingset{1}
\caption{Simulation results based on 300 data replications, where $U^{-1}$ is fixed at the band structure and $V^{-1}$ is fixed at the cluster structure, evaluated on the estimated $\widehat{U^{-1}}$.
Each entry reports an average of each metric with the standard error in parentheses. Boldface indicates the best-performing method.
}
\label{tab:cluster_U}
\renewcommand{\arraystretch}{0.85}
\tiny
\setlength{\tabcolsep}{4pt}
\begin{minipage}{0.49\textwidth}
\centering
\resizebox{\textwidth}{!}{%
\begin{tabular}{clcccc}
\toprule
$n$ & Method & TPR & FPR & MCC & AUC \\
\midrule
\multirow{28}{*}{10} & \multicolumn{5}{l}{\textit{$p=20$, $q=20$}} \\
 & matSPACE-sw & \textbf{0.398 (0.009)} & 0.129 (0.005) & \textbf{0.286 (0.005)} & \textbf{0.713 (0.003)} \\
 & matSPACE-dew & 0.262 (0.008) & 0.072 (0.004) & 0.255 (0.005) & 0.689 (0.003) \\
 & matNS-AND & 0.023 (0.002) & 0.001 (0.000) & 0.089 (0.005) & 0.701 (0.003) \\
 & matNS-OR & 0.086 (0.004) & 0.014 (0.001) & 0.167 (0.006) & 0.691 (0.003) \\
 & GEMINI & 0.315 (0.011) & 0.093 (0.005) & 0.267 (0.006) & 0.699 (0.003) \\
 & G-FDR-0.10 & 0.002 (0.000) & \textbf{0.000 (0.000)} & 0.011 (0.002) & 0.543 (0.002) \\
\cmidrule(lr){2-6}
 & \multicolumn{5}{l}{\textit{$p=20$, $q=50$}} \\
 & matSPACE-sw & 0.615 (0.008) & 0.138 (0.004) & \textbf{0.457 (0.004)} & \textbf{0.816 (0.002)} \\
 & matSPACE-dew & 0.432 (0.011) & 0.090 (0.005) & 0.392 (0.004) & 0.782 (0.002) \\
 & matNS-AND & 0.064 (0.003) & 0.001 (0.000) & 0.190 (0.007) & 0.797 (0.002) \\
 & matNS-OR & 0.175 (0.005) & 0.013 (0.001) & 0.309 (0.006) & 0.786 (0.002) \\
 & GEMINI & \textbf{0.649 (0.009)} & 0.207 (0.007) & 0.402 (0.004) & 0.796 (0.002) \\
 & G-FDR-0.10 & 0.037 (0.002) & \textbf{0.000 (0.000)} & 0.130 (0.007) & 0.659 (0.003) \\
\cmidrule(lr){2-6}
 & \multicolumn{5}{l}{\textit{$p=50$, $q=20$}} \\
 & matSPACE-sw & \textbf{0.289 (0.006)} & 0.063 (0.003) & \textbf{0.237 (0.003)} & \textbf{0.700 (0.002)} \\
 & matSPACE-dew & 0.165 (0.005) & 0.029 (0.001) & 0.199 (0.003) & 0.674 (0.002) \\
 & matNS-AND & 0.018 (0.001) & \textbf{0.000 (0.000)} & 0.100 (0.004) & 0.688 (0.002) \\
 & matNS-OR & 0.068 (0.002) & 0.009 (0.000) & 0.145 (0.003) & 0.676 (0.002) \\
 & GEMINI & 0.199 (0.007) & 0.035 (0.002) & 0.227 (0.003) & 0.686 (0.002) \\
 & G-FDR-0.10 & 0.000 (0.000) & \textbf{0.000 (0.000)} & 0.001 (0.000) & 0.528 (0.001) \\
\cmidrule(lr){2-6}
 & \multicolumn{5}{l}{\textit{$p=50$, $q=50$}} \\
 & matSPACE-sw & \textbf{0.495 (0.006)} & 0.066 (0.002) & \textbf{0.396 (0.003)} & \textbf{0.804 (0.002)} \\
 & matSPACE-dew & 0.275 (0.006) & 0.027 (0.001) & 0.326 (0.003) & 0.766 (0.001) \\
 & matNS-AND & 0.053 (0.002) & 0.001 (0.000) & 0.200 (0.004) & 0.784 (0.001) \\
 & matNS-OR & 0.146 (0.003) & 0.008 (0.000) & 0.279 (0.003) & 0.771 (0.001) \\
 & GEMINI & 0.490 (0.006) & 0.080 (0.003) & 0.368 (0.003) & 0.783 (0.002) \\
 & G-FDR-0.10 & 0.007 (0.001) & \textbf{0.000 (0.000)} & 0.050 (0.004) & 0.619 (0.002) \\
\midrule
\midrule
\multirow{14}{*}{40} & \multicolumn{5}{l}{\textit{$p=20$, $q=20$}} \\
 & matSPACE-sw & 0.727 (0.006) & 0.154 (0.004) & \textbf{0.521 (0.005)} & \textbf{0.864 (0.002)} \\
 & matSPACE-dew & 0.543 (0.011) & 0.106 (0.006) & 0.466 (0.004) & 0.827 (0.002) \\
 & matNS-AND & 0.110 (0.004) & 0.001 (0.000) & 0.275 (0.006) & 0.843 (0.002) \\
 & matNS-OR & 0.250 (0.006) & 0.013 (0.001) & 0.397 (0.006) & 0.832 (0.002) \\
 & GEMINI & \textbf{0.791 (0.006)} & 0.276 (0.006) & 0.430 (0.004) & 0.842 (0.002) \\
 & G-FDR-0.10 & 0.119 (0.004) & \textbf{0.000 (0.000)} & 0.294 (0.007) & 0.743 (0.003) \\
\cmidrule(lr){2-6}
 & \multicolumn{5}{l}{\textit{$p=20$, $q=50$}} \\
 & matSPACE-sw & 0.880 (0.004) & 0.149 (0.004) & \textbf{0.641 (0.004)} & \textbf{0.942 (0.001)} \\
 & matSPACE-dew & 0.805 (0.010) & 0.215 (0.011) & 0.543 (0.007) & 0.911 (0.002) \\
 & matNS-AND & 0.251 (0.006) & 0.001 (0.000) & 0.449 (0.005) & 0.924 (0.001) \\
 & matNS-OR & 0.449 (0.006) & 0.013 (0.001) & 0.582 (0.005) & 0.915 (0.002) \\
 & GEMINI & \textbf{0.934 (0.003)} & 0.343 (0.006) & 0.471 (0.005) & 0.923 (0.001) \\
 & G-FDR-0.10 & 0.419 (0.006) & \textbf{0.000 (0.000)} & 0.604 (0.005) & 0.880 (0.002) \\
\bottomrule
\end{tabular}
}
\end{minipage}%
\hfill
\begin{minipage}{0.49\textwidth}
\centering
\resizebox{\textwidth}{!}{%
\begin{tabular}{clcccc}
\toprule
$n$ & Method & TPR & FPR & MCC & AUC \\
\midrule
\multirow{14}{*}{40} & \multicolumn{5}{l}{\textit{$p=50$, $q=20$}} \\
 & matSPACE-sw & 0.625 (0.005) & 0.073 (0.002) & \textbf{0.477 (0.003)} & \textbf{0.858 (0.001)} \\
 & matSPACE-dew & 0.368 (0.007) & 0.032 (0.002) & 0.400 (0.003) & 0.817 (0.001) \\
 & matNS-AND & 0.093 (0.002) & 0.001 (0.000) & 0.274 (0.004) & 0.836 (0.001) \\
 & matNS-OR & 0.218 (0.003) & 0.008 (0.000) & 0.364 (0.003) & 0.824 (0.001) \\
 & GEMINI & \textbf{0.651 (0.006)} & 0.120 (0.004) & 0.413 (0.004) & 0.836 (0.001) \\
 & G-FDR-0.10 & 0.037 (0.002) & \textbf{0.000 (0.000)} & 0.169 (0.004) & 0.706 (0.002) \\
\cmidrule(lr){2-6}
 & \multicolumn{5}{l}{\textit{$p=50$, $q=50$}} \\
 & matSPACE-sw & 0.827 (0.003) & 0.072 (0.002) & \textbf{0.612 (0.004)} & \textbf{0.941 (0.001)} \\
 & matSPACE-dew & 0.564 (0.008) & 0.037 (0.002) & 0.551 (0.003) & 0.902 (0.001) \\
 & matNS-AND & 0.231 (0.003) & 0.001 (0.000) & 0.452 (0.003) & 0.920 (0.001) \\
 & matNS-OR & 0.414 (0.004) & 0.009 (0.000) & 0.554 (0.003) & 0.910 (0.001) \\
 & GEMINI & \textbf{0.876 (0.003)} & 0.178 (0.004) & 0.459 (0.005) & 0.922 (0.001) \\
 & G-FDR-0.10 & 0.241 (0.003) & \textbf{0.000 (0.000)} & 0.473 (0.003) & 0.852 (0.002) \\
\midrule
\midrule
\multirow{28}{*}{70} & \multicolumn{5}{l}{\textit{$p=20$, $q=20$}} \\
 & matSPACE-sw & 0.832 (0.005) & 0.161 (0.005) & \textbf{0.594 (0.005)} & \textbf{0.916 (0.002)} \\
 & matSPACE-dew & 0.720 (0.010) & 0.170 (0.008) & 0.523 (0.006) & 0.880 (0.002) \\
 & matNS-AND & 0.189 (0.005) & 0.001 (0.000) & 0.381 (0.006) & 0.895 (0.002) \\
 & matNS-OR & 0.371 (0.006) & 0.013 (0.001) & 0.513 (0.005) & 0.885 (0.002) \\
 & GEMINI & \textbf{0.891 (0.004)} & 0.323 (0.006) & 0.456 (0.005) & 0.896 (0.002) \\
 & G-FDR-0.10 & 0.287 (0.006) & \textbf{0.000 (0.000)} & 0.489 (0.006) & 0.835 (0.002) \\
\cmidrule(lr){2-6}
 & \multicolumn{5}{l}{\textit{$p=20$, $q=50$}} \\
 & matSPACE-sw & 0.938 (0.003) & 0.126 (0.004) & 0.716 (0.005) & \textbf{0.973 (0.001)} \\
 & matSPACE-dew & 0.913 (0.006) & 0.312 (0.014) & 0.532 (0.011) & 0.948 (0.001) \\
 & matNS-AND & 0.372 (0.007) & 0.001 (0.000) & 0.560 (0.005) & 0.958 (0.001) \\
 & matNS-OR & 0.582 (0.006) & 0.011 (0.001) & 0.691 (0.004) & 0.951 (0.001) \\
 & GEMINI & \textbf{0.978 (0.002)} & 0.389 (0.006) & 0.465 (0.005) & 0.959 (0.001) \\
 & G-FDR-0.10 & 0.629 (0.005) & \textbf{0.000 (0.000)} & \textbf{0.759 (0.004)} & 0.935 (0.002) \\
\cmidrule(lr){2-6}
 & \multicolumn{5}{l}{\textit{$p=50$, $q=20$}} \\
 & matSPACE-sw & 0.752 (0.004) & 0.069 (0.002) & \textbf{0.571 (0.003)} & \textbf{0.913 (0.001)} \\
 & matSPACE-dew & 0.486 (0.008) & 0.035 (0.002) & 0.495 (0.003) & 0.870 (0.001) \\
 & matNS-AND & 0.170 (0.003) & 0.001 (0.000) & 0.380 (0.004) & 0.890 (0.001) \\
 & matNS-OR & 0.335 (0.004) & 0.009 (0.000) & 0.480 (0.003) & 0.879 (0.001) \\
 & GEMINI & \textbf{0.805 (0.005)} & 0.166 (0.005) & 0.441 (0.005) & 0.891 (0.001) \\
 & G-FDR-0.10 & 0.143 (0.003) & \textbf{0.000 (0.000)} & 0.359 (0.004) & 0.799 (0.002) \\
\cmidrule(lr){2-6}
 & \multicolumn{5}{l}{\textit{$p=50$, $q=50$}} \\
 & matSPACE-sw & 0.906 (0.002) & 0.063 (0.002) & \textbf{0.686 (0.004)} & \textbf{0.971 (0.001)} \\
 & matSPACE-dew & 0.675 (0.007) & 0.038 (0.003) & 0.640 (0.004) & 0.940 (0.001) \\
 & matNS-AND & 0.356 (0.004) & 0.001 (0.000) & 0.572 (0.003) & 0.954 (0.001) \\
 & matNS-OR & 0.556 (0.004) & 0.008 (0.000) & 0.671 (0.003) & 0.947 (0.001) \\
 & GEMINI & \textbf{0.942 (0.002)} & 0.189 (0.005) & 0.486 (0.005) & 0.958 (0.001) \\
 & G-FDR-0.10 & 0.446 (0.004) & \textbf{0.000 (0.000)} & 0.652 (0.003) & 0.917 (0.001) \\
\bottomrule
\end{tabular}
}
\end{minipage}
\end{sidewaystable}

\begin{sidewaystable}
\centering
\spacingset{1}
\caption{Simulation results based on 300 data replications, where $U^{-1}$ is fixed at the band structure and $V^{-1}$ is fixed at the cluster structure, evaluated on the estimated $\widehat{V^{-1}}$.
Each entry reports an average of each metric with the standard error in parentheses. Boldface indicates the best-performing method.
}
\label{tab:cluster_V}
\renewcommand{\arraystretch}{0.85}
\tiny
\setlength{\tabcolsep}{4pt}
\begin{minipage}{0.49\textwidth}
\centering
\resizebox{\textwidth}{!}{%
\begin{tabular}{clcccc}
\toprule
$n$ & Method & TPR & FPR & MCC & AUC \\
\midrule
\multirow{28}{*}{10} & \multicolumn{5}{l}{\textit{$p=20$, $q=20$}} \\
 & matSPACE-sw & \textbf{0.363 (0.012)} & 0.149 (0.007) & \textbf{0.237 (0.007)} & \textbf{0.691 (0.003)} \\
 & matSPACE-dew & 0.221 (0.009) & 0.070 (0.004) & 0.214 (0.006) & 0.672 (0.003) \\
 & matNS-AND & 0.012 (0.001) & 0.001 (0.000) & 0.049 (0.004) & 0.681 (0.003) \\
 & matNS-OR & 0.073 (0.004) & 0.016 (0.001) & 0.127 (0.006) & 0.673 (0.003) \\
 & GEMINI & 0.234 (0.012) & 0.084 (0.005) & 0.185 (0.008) & 0.675 (0.003) \\
 & G-FDR-0.10 & 0.000 (0.000) & \textbf{0.000 (0.000)} & 0.001 (0.001) & 0.520 (0.002) \\
\cmidrule(lr){2-6}
 & \multicolumn{5}{l}{\textit{$p=20$, $q=50$}} \\
 & matSPACE-sw & \textbf{0.148 (0.006)} & 0.044 (0.003) & \textbf{0.147 (0.004)} & \textbf{0.645 (0.002)} \\
 & matSPACE-dew & 0.090 (0.004) & 0.022 (0.001) & 0.128 (0.003) & 0.617 (0.002) \\
 & matNS-AND & 0.003 (0.000) & \textbf{0.000 (0.000)} & 0.023 (0.002) & 0.626 (0.002) \\
 & matNS-OR & 0.030 (0.001) & 0.006 (0.000) & 0.077 (0.003) & 0.620 (0.002) \\
 & GEMINI & 0.053 (0.004) & 0.012 (0.001) & 0.086 (0.005) & 0.636 (0.002) \\
 & G-FDR-0.10 & 0.000 (0.000) & \textbf{0.000 (0.000)} & 0.000 (0.000) & 0.502 (0.000) \\
\cmidrule(lr){2-6}
 & \multicolumn{5}{l}{\textit{$p=50$, $q=20$}} \\
 & matSPACE-sw & 0.657 (0.010) & 0.207 (0.006) & \textbf{0.438 (0.005)} & \textbf{0.809 (0.003)} \\
 & matSPACE-dew & 0.444 (0.013) & 0.118 (0.006) & 0.376 (0.005) & 0.775 (0.003) \\
 & matNS-AND & 0.050 (0.003) & 0.002 (0.000) & 0.140 (0.007) & 0.788 (0.003) \\
 & matNS-OR & 0.184 (0.007) & 0.025 (0.001) & 0.272 (0.007) & 0.776 (0.003) \\
 & GEMINI & \textbf{0.678 (0.010)} & 0.265 (0.008) & 0.395 (0.005) & 0.781 (0.003) \\
 & G-FDR-0.10 & 0.013 (0.002) & \textbf{0.000 (0.000)} & 0.049 (0.005) & 0.636 (0.004) \\
\cmidrule(lr){2-6}
 & \multicolumn{5}{l}{\textit{$p=50$, $q=50$}} \\
 & matSPACE-sw & 0.351 (0.008) & 0.056 (0.002) & \textbf{0.332 (0.004)} & \textbf{0.761 (0.002)} \\
 & matSPACE-dew & 0.191 (0.005) & 0.027 (0.001) & 0.251 (0.003) & 0.715 (0.002) \\
 & matNS-AND & 0.014 (0.001) & \textbf{0.000 (0.000)} & 0.079 (0.004) & 0.731 (0.002) \\
 & matNS-OR & 0.078 (0.003) & 0.007 (0.000) & 0.176 (0.004) & 0.722 (0.002) \\
 & GEMINI & \textbf{0.370 (0.009)} & 0.074 (0.003) & 0.318 (0.004) & 0.743 (0.002) \\
 & G-FDR-0.10 & 0.000 (0.000) & \textbf{0.000 (0.000)} & 0.000 (0.000) & 0.531 (0.001) \\
\midrule
\midrule
\multirow{14}{*}{40} & \multicolumn{5}{l}{\textit{$p=20$, $q=20$}} \\
 & matSPACE-sw & 0.809 (0.008) & 0.237 (0.006) & \textbf{0.532 (0.005)} & \textbf{0.875 (0.002)} \\
 & matSPACE-dew & 0.615 (0.014) & 0.176 (0.008) & 0.459 (0.005) & 0.833 (0.002) \\
 & matNS-AND & 0.102 (0.005) & 0.004 (0.000) & 0.223 (0.008) & 0.849 (0.002) \\
 & matNS-OR & 0.300 (0.009) & 0.030 (0.001) & 0.384 (0.008) & 0.836 (0.002) \\
 & GEMINI & \textbf{0.858 (0.006)} & 0.365 (0.007) & 0.447 (0.004) & 0.845 (0.002) \\
 & G-FDR-0.10 & 0.086 (0.006) & \textbf{0.000 (0.000)} & 0.198 (0.009) & 0.757 (0.005) \\
\cmidrule(lr){2-6}
 & \multicolumn{5}{l}{\textit{$p=20$, $q=50$}} \\
 & matSPACE-sw & 0.541 (0.008) & 0.077 (0.002) & \textbf{0.450 (0.003)} & \textbf{0.832 (0.002)} \\
 & matSPACE-dew & 0.281 (0.006) & 0.032 (0.001) & 0.335 (0.004) & 0.777 (0.002) \\
 & matNS-AND & 0.032 (0.002) & 0.001 (0.000) & 0.137 (0.004) & 0.797 (0.002) \\
 & matNS-OR & 0.137 (0.004) & 0.010 (0.000) & 0.262 (0.004) & 0.786 (0.002) \\
 & GEMINI & \textbf{0.593 (0.008)} & 0.120 (0.004) & 0.415 (0.003) & 0.811 (0.002) \\
 & G-FDR-0.10 & 0.001 (0.000) & \textbf{0.000 (0.000)} & 0.006 (0.001) & 0.610 (0.003) \\
\bottomrule
\end{tabular}
}
\end{minipage}%
\hfill
\begin{minipage}{0.49\textwidth}
\centering
\resizebox{\textwidth}{!}{%
\begin{tabular}{clcccc}
\toprule
$n$ & Method & TPR & FPR & MCC & AUC \\
\midrule
\multirow{14}{*}{40} & \multicolumn{5}{l}{\textit{$p=50$, $q=20$}} \\
 & matSPACE-sw & 0.962 (0.002) & 0.246 (0.005) & 0.646 (0.004) & \textbf{0.963 (0.001)} \\
 & matSPACE-dew & 0.917 (0.007) & 0.325 (0.010) & 0.550 (0.006) & 0.923 (0.002) \\
 & matNS-AND & 0.305 (0.009) & 0.008 (0.001) & 0.454 (0.008) & 0.938 (0.001) \\
 & matNS-OR & 0.575 (0.009) & 0.040 (0.002) & 0.614 (0.006) & 0.927 (0.001) \\
 & GEMINI & \textbf{0.983 (0.002)} & 0.480 (0.006) & 0.462 (0.004) & 0.937 (0.001) \\
 & G-FDR-0.10 & 0.532 (0.011) & \textbf{0.001 (0.000)} & \textbf{0.662 (0.008)} & 0.927 (0.002) \\
\cmidrule(lr){2-6}
 & \multicolumn{5}{l}{\textit{$p=50$, $q=50$}} \\
 & matSPACE-sw & 0.853 (0.004) & 0.105 (0.002) & \textbf{0.612 (0.003)} & \textbf{0.942 (0.001)} \\
 & matSPACE-dew & 0.549 (0.009) & 0.059 (0.002) & 0.499 (0.003) & 0.884 (0.001) \\
 & matNS-AND & 0.143 (0.004) & 0.002 (0.000) & 0.331 (0.005) & 0.905 (0.001) \\
 & matNS-OR & 0.365 (0.006) & 0.016 (0.000) & 0.482 (0.004) & 0.896 (0.001) \\
 & GEMINI & \textbf{0.902 (0.003)} & 0.213 (0.005) & 0.496 (0.004) & 0.922 (0.001) \\
 & G-FDR-0.10 & 0.058 (0.003) & \textbf{0.000 (0.000)} & 0.202 (0.006) & 0.819 (0.003) \\
\midrule
\midrule
\multirow{28}{*}{70} & \multicolumn{5}{l}{\textit{$p=20$, $q=20$}} \\
 & matSPACE-sw & 0.926 (0.004) & 0.252 (0.005) & \textbf{0.612 (0.004)} & \textbf{0.938 (0.001)} \\
 & matSPACE-dew & 0.840 (0.010) & 0.273 (0.009) & 0.535 (0.005) & 0.895 (0.002) \\
 & matNS-AND & 0.214 (0.008) & 0.007 (0.000) & 0.360 (0.008) & 0.910 (0.002) \\
 & matNS-OR & 0.469 (0.010) & 0.038 (0.002) & 0.527 (0.007) & 0.898 (0.002) \\
 & GEMINI & \textbf{0.959 (0.003)} & 0.436 (0.007) & 0.475 (0.004) & 0.909 (0.001) \\
 & G-FDR-0.10 & 0.319 (0.011) & \textbf{0.000 (0.000)} & 0.477 (0.011) & 0.878 (0.003) \\
\cmidrule(lr){2-6}
 & \multicolumn{5}{l}{\textit{$p=20$, $q=50$}} \\
 & matSPACE-sw & 0.752 (0.005) & 0.095 (0.002) & \textbf{0.564 (0.003)} & \textbf{0.906 (0.001)} \\
 & matSPACE-dew & 0.442 (0.008) & 0.047 (0.002) & 0.442 (0.003) & 0.846 (0.002) \\
 & matNS-AND & 0.086 (0.003) & 0.001 (0.000) & 0.247 (0.005) & 0.868 (0.001) \\
 & matNS-OR & 0.264 (0.005) & 0.013 (0.000) & 0.396 (0.004) & 0.858 (0.001) \\
 & GEMINI & \textbf{0.806 (0.006)} & 0.180 (0.006) & 0.480 (0.004) & 0.885 (0.001) \\
 & G-FDR-0.10 & 0.013 (0.001) & \textbf{0.000 (0.000)} & 0.073 (0.005) & 0.740 (0.003) \\
\cmidrule(lr){2-6}
 & \multicolumn{5}{l}{\textit{$p=50$, $q=20$}} \\
 & matSPACE-sw & 0.992 (0.001) & 0.210 (0.004) & 0.708 (0.004) & \textbf{0.988 (0.000)} \\
 & matSPACE-dew & 0.975 (0.003) & 0.403 (0.013) & 0.534 (0.009) & 0.961 (0.001) \\
 & matNS-AND & 0.497 (0.010) & 0.007 (0.001) & 0.622 (0.007) & 0.972 (0.001) \\
 & matNS-OR & 0.759 (0.008) & 0.042 (0.002) & 0.754 (0.005) & 0.963 (0.001) \\
 & GEMINI & \textbf{0.998 (0.000)} & 0.493 (0.006) & 0.466 (0.004) & 0.971 (0.001) \\
 & G-FDR-0.10 & 0.817 (0.007) & \textbf{0.001 (0.000)} & \textbf{0.872 (0.005)} & 0.970 (0.001) \\
\cmidrule(lr){2-6}
 & \multicolumn{5}{l}{\textit{$p=50$, $q=50$}} \\
 & matSPACE-sw & 0.951 (0.002) & 0.107 (0.002) & \textbf{0.673 (0.004)} & \textbf{0.979 (0.001)} \\
 & matSPACE-dew & 0.718 (0.009) & 0.076 (0.003) & 0.590 (0.003) & 0.933 (0.001) \\
 & matNS-AND & 0.278 (0.006) & 0.002 (0.000) & 0.480 (0.005) & 0.951 (0.001) \\
 & matNS-OR & 0.552 (0.006) & 0.018 (0.000) & 0.628 (0.004) & 0.943 (0.001) \\
 & GEMINI & \textbf{0.974 (0.001)} & 0.242 (0.006) & 0.511 (0.005) & 0.965 (0.001) \\
 & G-FDR-0.10 & 0.303 (0.007) & \textbf{0.000 (0.000)} & 0.516 (0.007) & 0.926 (0.002) \\
\bottomrule
\end{tabular}
}
\end{minipage}
\end{sidewaystable}

\begin{sidewaystable}
\centering
\spacingset{1}
\caption{Simulation results based on 300 data replications, where $U^{-1}$ is fixed at the band structure and $V^{-1}$ is fixed at the hub structure, evaluated on the estimated $\widehat{U^{-1}}$.
Each entry reports an average of each metric with the standard error in parentheses. Boldface indicates the best-performing method.
}
\label{tab:hub_U}
\renewcommand{\arraystretch}{0.85}
\tiny
\setlength{\tabcolsep}{4pt}
\begin{minipage}{0.49\textwidth}
\centering
\resizebox{\textwidth}{!}{%
\begin{tabular}{clcccc}
\toprule
$n$ & Method & TPR & FPR & MCC & AUC \\
\midrule
\multirow{28}{*}{10} & \multicolumn{5}{l}{\textit{$p=20$, $q=20$}} \\
 & matSPACE-sw & \textbf{0.394 (0.009)} & 0.118 (0.005) & \textbf{0.300 (0.005)} & \textbf{0.716 (0.003)} \\
 & matSPACE-dew & 0.255 (0.008) & 0.063 (0.004) & 0.265 (0.005) & 0.692 (0.003) \\
 & matNS-AND & 0.022 (0.002) & 0.001 (0.000) & 0.086 (0.006) & 0.705 (0.003) \\
 & matNS-OR & 0.081 (0.004) & 0.011 (0.001) & 0.164 (0.006) & 0.695 (0.003) \\
 & GEMINI & 0.319 (0.010) & 0.092 (0.005) & 0.271 (0.006) & 0.701 (0.003) \\
 & G-FDR-0.10 & 0.002 (0.001) & \textbf{0.000 (0.000)} & 0.010 (0.002) & 0.543 (0.002) \\
\cmidrule(lr){2-6}
 & \multicolumn{5}{l}{\textit{$p=20$, $q=50$}} \\
 & matSPACE-sw & 0.622 (0.008) & 0.139 (0.004) & \textbf{0.463 (0.004)} & \textbf{0.820 (0.002)} \\
 & matSPACE-dew & 0.429 (0.010) & 0.084 (0.004) & 0.401 (0.004) & 0.785 (0.002) \\
 & matNS-AND & 0.064 (0.003) & 0.001 (0.000) & 0.194 (0.006) & 0.801 (0.002) \\
 & matNS-OR & 0.173 (0.005) & 0.012 (0.001) & 0.309 (0.006) & 0.789 (0.002) \\
 & GEMINI & \textbf{0.672 (0.008)} & 0.215 (0.006) & 0.406 (0.004) & 0.798 (0.002) \\
 & G-FDR-0.10 & 0.037 (0.003) & \textbf{0.000 (0.000)} & 0.127 (0.007) & 0.660 (0.003) \\
\cmidrule(lr){2-6}
 & \multicolumn{5}{l}{\textit{$p=50$, $q=20$}} \\
 & matSPACE-sw & \textbf{0.276 (0.006)} & 0.053 (0.002) & \textbf{0.249 (0.003)} & \textbf{0.704 (0.002)} \\
 & matSPACE-dew & 0.158 (0.004) & 0.024 (0.001) & 0.205 (0.003) & 0.678 (0.002) \\
 & matNS-AND & 0.017 (0.001) & \textbf{0.000 (0.000)} & 0.098 (0.004) & 0.692 (0.002) \\
 & matNS-OR & 0.066 (0.002) & 0.007 (0.000) & 0.151 (0.003) & 0.681 (0.002) \\
 & GEMINI & 0.192 (0.007) & 0.031 (0.002) & 0.234 (0.003) & 0.691 (0.002) \\
 & G-FDR-0.10 & 0.000 (0.000) & \textbf{0.000 (0.000)} & 0.001 (0.001) & 0.532 (0.001) \\
\cmidrule(lr){2-6}
 & \multicolumn{5}{l}{\textit{$p=50$, $q=50$}} \\
 & matSPACE-sw & 0.499 (0.006) & 0.064 (0.002) & \textbf{0.405 (0.003)} & \textbf{0.808 (0.002)} \\
 & matSPACE-dew & 0.277 (0.006) & 0.027 (0.001) & 0.333 (0.003) & 0.769 (0.001) \\
 & matNS-AND & 0.055 (0.002) & \textbf{0.000 (0.000)} & 0.205 (0.004) & 0.788 (0.001) \\
 & matNS-OR & 0.147 (0.003) & 0.007 (0.000) & 0.281 (0.003) & 0.775 (0.001) \\
 & GEMINI & \textbf{0.511 (0.008)} & 0.095 (0.004) & 0.364 (0.003) & 0.787 (0.002) \\
 & G-FDR-0.10 & 0.007 (0.001) & \textbf{0.000 (0.000)} & 0.050 (0.004) & 0.621 (0.002) \\
\midrule
\midrule
\multirow{14}{*}{40} & \multicolumn{5}{l}{\textit{$p=20$, $q=20$}} \\
 & matSPACE-sw & 0.718 (0.006) & 0.140 (0.004) & \textbf{0.537 (0.005)} & \textbf{0.868 (0.002)} \\
 & matSPACE-dew & 0.540 (0.011) & 0.095 (0.005) & 0.480 (0.004) & 0.831 (0.002) \\
 & matNS-AND & 0.109 (0.004) & 0.001 (0.000) & 0.275 (0.006) & 0.847 (0.002) \\
 & matNS-OR & 0.255 (0.006) & 0.012 (0.001) & 0.406 (0.005) & 0.836 (0.002) \\
 & GEMINI & \textbf{0.788 (0.006)} & 0.268 (0.007) & 0.436 (0.005) & 0.846 (0.002) \\
 & G-FDR-0.10 & 0.122 (0.005) & \textbf{0.000 (0.000)} & 0.296 (0.007) & 0.751 (0.003) \\
\cmidrule(lr){2-6}
 & \multicolumn{5}{l}{\textit{$p=20$, $q=50$}} \\
 & matSPACE-sw & 0.882 (0.004) & 0.140 (0.004) & \textbf{0.655 (0.004)} & \textbf{0.944 (0.001)} \\
 & matSPACE-dew & 0.807 (0.009) & 0.210 (0.011) & 0.551 (0.007) & 0.913 (0.001) \\
 & matNS-AND & 0.252 (0.006) & 0.001 (0.000) & 0.450 (0.005) & 0.926 (0.001) \\
 & matNS-OR & 0.451 (0.006) & 0.011 (0.001) & 0.588 (0.005) & 0.917 (0.001) \\
 & GEMINI & \textbf{0.937 (0.003)} & 0.343 (0.006) & 0.473 (0.005) & 0.926 (0.001) \\
 & G-FDR-0.10 & 0.423 (0.006) & \textbf{0.000 (0.000)} & 0.608 (0.005) & 0.884 (0.002) \\
\bottomrule
\end{tabular}
}
\end{minipage}%
\hfill
\begin{minipage}{0.49\textwidth}
\centering
\resizebox{\textwidth}{!}{%
\begin{tabular}{clcccc}
\toprule
$n$ & Method & TPR & FPR & MCC & AUC \\
\midrule
\multirow{14}{*}{40} & \multicolumn{5}{l}{\textit{$p=50$, $q=20$}} \\
 & matSPACE-sw & 0.618 (0.005) & 0.065 (0.002) & \textbf{0.492 (0.003)} & \textbf{0.863 (0.001)} \\
 & matSPACE-dew & 0.363 (0.007) & 0.029 (0.001) & 0.409 (0.003) & 0.821 (0.001) \\
 & matNS-AND & 0.094 (0.002) & 0.001 (0.000) & 0.277 (0.004) & 0.841 (0.001) \\
 & matNS-OR & 0.218 (0.004) & 0.008 (0.000) & 0.365 (0.004) & 0.828 (0.001) \\
 & GEMINI & \textbf{0.676 (0.006)} & 0.136 (0.004) & 0.405 (0.004) & 0.840 (0.001) \\
 & G-FDR-0.10 & 0.039 (0.002) & \textbf{0.000 (0.000)} & 0.174 (0.004) & 0.713 (0.002) \\
\cmidrule(lr){2-6}
 & \multicolumn{5}{l}{\textit{$p=50$, $q=50$}} \\
 & matSPACE-sw & 0.826 (0.003) & 0.068 (0.002) & \textbf{0.623 (0.004)} & \textbf{0.942 (0.001)} \\
 & matSPACE-dew & 0.564 (0.008) & 0.035 (0.002) & 0.558 (0.003) & 0.904 (0.001) \\
 & matNS-AND & 0.235 (0.004) & 0.001 (0.000) & 0.455 (0.004) & 0.922 (0.001) \\
 & matNS-OR & 0.422 (0.004) & 0.009 (0.000) & 0.561 (0.003) & 0.912 (0.001) \\
 & GEMINI & \textbf{0.878 (0.003)} & 0.185 (0.005) & 0.460 (0.005) & 0.924 (0.001) \\
 & G-FDR-0.10 & 0.247 (0.003) & \textbf{0.000 (0.000)} & 0.479 (0.003) & 0.856 (0.001) \\
\midrule
\midrule
\multirow{28}{*}{70} & \multicolumn{5}{l}{\textit{$p=20$, $q=20$}} \\
 & matSPACE-sw & 0.834 (0.004) & 0.144 (0.004) & \textbf{0.615 (0.005)} & \textbf{0.920 (0.002)} \\
 & matSPACE-dew & 0.707 (0.011) & 0.152 (0.008) & 0.537 (0.006) & 0.884 (0.002) \\
 & matNS-AND & 0.192 (0.005) & 0.002 (0.000) & 0.384 (0.006) & 0.899 (0.002) \\
 & matNS-OR & 0.376 (0.007) & 0.012 (0.001) & 0.521 (0.005) & 0.889 (0.002) \\
 & GEMINI & \textbf{0.895 (0.004)} & 0.324 (0.007) & 0.461 (0.005) & 0.900 (0.002) \\
 & G-FDR-0.10 & 0.293 (0.006) & \textbf{0.000 (0.000)} & 0.495 (0.006) & 0.841 (0.002) \\
\cmidrule(lr){2-6}
 & \multicolumn{5}{l}{\textit{$p=20$, $q=50$}} \\
 & matSPACE-sw & 0.941 (0.003) & 0.121 (0.003) & 0.725 (0.004) & \textbf{0.974 (0.001)} \\
 & matSPACE-dew & 0.913 (0.006) & 0.313 (0.015) & 0.533 (0.011) & 0.949 (0.001) \\
 & matNS-AND & 0.376 (0.006) & 0.001 (0.000) & 0.564 (0.005) & 0.960 (0.001) \\
 & matNS-OR & 0.585 (0.006) & 0.010 (0.001) & 0.696 (0.004) & 0.953 (0.001) \\
 & GEMINI & \textbf{0.979 (0.001)} & 0.380 (0.005) & 0.471 (0.004) & 0.960 (0.001) \\
 & G-FDR-0.10 & 0.636 (0.005) & \textbf{0.000 (0.000)} & \textbf{0.765 (0.004)} & 0.936 (0.001) \\
\cmidrule(lr){2-6}
 & \multicolumn{5}{l}{\textit{$p=50$, $q=20$}} \\
 & matSPACE-sw & 0.754 (0.004) & 0.064 (0.002) & \textbf{0.586 (0.004)} & \textbf{0.916 (0.001)} \\
 & matSPACE-dew & 0.486 (0.008) & 0.033 (0.002) & 0.504 (0.003) & 0.874 (0.001) \\
 & matNS-AND & 0.169 (0.003) & 0.001 (0.000) & 0.382 (0.004) & 0.894 (0.001) \\
 & matNS-OR & 0.335 (0.004) & 0.008 (0.000) & 0.489 (0.003) & 0.882 (0.001) \\
 & GEMINI & \textbf{0.810 (0.004)} & 0.159 (0.005) & 0.452 (0.005) & 0.895 (0.001) \\
 & G-FDR-0.10 & 0.145 (0.003) & \textbf{0.000 (0.000)} & 0.362 (0.004) & 0.805 (0.002) \\
\cmidrule(lr){2-6}
 & \multicolumn{5}{l}{\textit{$p=50$, $q=50$}} \\
 & matSPACE-sw & 0.905 (0.002) & 0.059 (0.002) & \textbf{0.698 (0.004)} & \textbf{0.972 (0.001)} \\
 & matSPACE-dew & 0.671 (0.007) & 0.034 (0.002) & 0.648 (0.004) & 0.941 (0.001) \\
 & matNS-AND & 0.363 (0.004) & 0.001 (0.000) & 0.578 (0.003) & 0.956 (0.001) \\
 & matNS-OR & 0.564 (0.004) & 0.008 (0.000) & 0.677 (0.003) & 0.948 (0.001) \\
 & GEMINI & \textbf{0.947 (0.002)} & 0.214 (0.005) & 0.463 (0.006) & 0.959 (0.001) \\
 & G-FDR-0.10 & 0.453 (0.003) & \textbf{0.000 (0.000)} & 0.657 (0.003) & 0.920 (0.001) \\
\bottomrule
\end{tabular}
}
\end{minipage}
\end{sidewaystable}

\begin{sidewaystable}
\centering
\spacingset{1}
\caption{Simulation results based on 300 data replications, where $U^{-1}$ is fixed at the band structure and $V^{-1}$ is fixed at the hub structure, evaluated on the estimated $\widehat{V^{-1}}$.
Each entry reports an average of each metric with the standard error in parentheses. Boldface indicates the best-performing method.
}
\label{tab:hub_V}
\renewcommand{\arraystretch}{0.85}
\tiny
\setlength{\tabcolsep}{4pt}
\begin{minipage}{0.49\textwidth}
\centering
\resizebox{\textwidth}{!}{%
\begin{tabular}{clcccc}
\toprule
$n$ & Method & TPR & FPR & MCC & AUC \\
\midrule
\multirow{28}{*}{10} & \multicolumn{5}{l}{\textit{$p=20$, $q=20$}} \\
 & matSPACE-sw & \textbf{0.288 (0.013)} & 0.076 (0.004) & 0.207 (0.008) & \textbf{0.712 (0.004)} \\
 & matSPACE-dew & 0.225 (0.010) & 0.045 (0.003) & \textbf{0.225 (0.010)} & 0.693 (0.005) \\
 & matNS-AND & 0.012 (0.002) & \textbf{0.000 (0.000)} & 0.044 (0.005) & 0.691 (0.005) \\
 & matNS-OR & 0.100 (0.008) & 0.005 (0.001) & 0.170 (0.012) & 0.695 (0.005) \\
 & GEMINI & 0.192 (0.012) & 0.044 (0.004) & 0.164 (0.009) & 0.700 (0.004) \\
 & G-FDR-0.10 & 0.000 (0.000) & \textbf{0.000 (0.000)} & 0.001 (0.001) & 0.511 (0.001) \\
\cmidrule(lr){2-6}
 & \multicolumn{5}{l}{\textit{$p=20$, $q=50$}} \\
 & matSPACE-sw & \textbf{0.220 (0.009)} & 0.036 (0.002) & 0.179 (0.005) & \textbf{0.711 (0.003)} \\
 & matSPACE-dew & 0.164 (0.006) & 0.020 (0.001) & \textbf{0.194 (0.006)} & 0.685 (0.004) \\
 & matNS-AND & 0.011 (0.001) & \textbf{0.000 (0.000)} & 0.058 (0.005) & 0.687 (0.003) \\
 & matNS-OR & 0.084 (0.004) & 0.005 (0.000) & 0.163 (0.007) & 0.690 (0.004) \\
 & GEMINI & 0.106 (0.008) & 0.012 (0.001) & 0.131 (0.007) & 0.700 (0.003) \\
 & G-FDR-0.10 & 0.000 (0.000) & \textbf{0.000 (0.000)} & 0.000 (0.000) & 0.506 (0.001) \\
\cmidrule(lr){2-6}
 & \multicolumn{5}{l}{\textit{$p=50$, $q=20$}} \\
 & matSPACE-sw & 0.533 (0.014) & 0.074 (0.004) & \textbf{0.426 (0.009)} & \textbf{0.839 (0.004)} \\
 & matSPACE-dew & 0.398 (0.012) & 0.042 (0.003) & 0.407 (0.010) & 0.808 (0.005) \\
 & matNS-AND & 0.029 (0.003) & \textbf{0.000 (0.000)} & 0.089 (0.008) & 0.807 (0.004) \\
 & matNS-OR & 0.265 (0.012) & 0.005 (0.001) & 0.392 (0.014) & 0.809 (0.005) \\
 & GEMINI & \textbf{0.558 (0.015)} & 0.102 (0.005) & 0.397 (0.009) & 0.822 (0.004) \\
 & G-FDR-0.10 & 0.006 (0.001) & \textbf{0.000 (0.000)} & 0.022 (0.004) & 0.580 (0.004) \\
\cmidrule(lr){2-6}
 & \multicolumn{5}{l}{\textit{$p=50$, $q=50$}} \\
 & matSPACE-sw & \textbf{0.453 (0.009)} & 0.034 (0.002) & \textbf{0.386 (0.005)} & \textbf{0.837 (0.002)} \\
 & matSPACE-dew & 0.275 (0.007) & 0.015 (0.001) & 0.334 (0.006) & 0.796 (0.003) \\
 & matNS-AND & 0.026 (0.002) & \textbf{0.000 (0.000)} & 0.115 (0.006) & 0.804 (0.003) \\
 & matNS-OR & 0.218 (0.007) & 0.004 (0.000) & 0.354 (0.008) & 0.805 (0.003) \\
 & GEMINI & 0.437 (0.011) & 0.040 (0.002) & 0.365 (0.006) & 0.822 (0.002) \\
 & G-FDR-0.10 & 0.000 (0.000) & \textbf{0.000 (0.000)} & 0.000 (0.000) & 0.551 (0.002) \\
\midrule
\midrule
\multirow{14}{*}{40} & \multicolumn{5}{l}{\textit{$p=20$, $q=20$}} \\
 & matSPACE-sw & 0.697 (0.011) & 0.086 (0.003) & \textbf{0.523 (0.007)} & \textbf{0.893 (0.003)} \\
 & matSPACE-dew & 0.509 (0.012) & 0.047 (0.003) & 0.488 (0.009) & 0.860 (0.004) \\
 & matNS-AND & 0.039 (0.003) & \textbf{0.000 (0.000)} & 0.119 (0.008) & 0.859 (0.004) \\
 & matNS-OR & 0.376 (0.013) & 0.005 (0.000) & \textbf{0.523 (0.012)} & 0.862 (0.004) \\
 & GEMINI & \textbf{0.762 (0.010)} & 0.150 (0.005) & 0.464 (0.006) & 0.878 (0.003) \\
 & G-FDR-0.10 & 0.031 (0.003) & \textbf{0.000 (0.000)} & 0.092 (0.008) & 0.676 (0.005) \\
\cmidrule(lr){2-6}
 & \multicolumn{5}{l}{\textit{$p=20$, $q=50$}} \\
 & matSPACE-sw & 0.626 (0.008) & 0.035 (0.002) & \textbf{0.505 (0.005)} & \textbf{0.899 (0.002)} \\
 & matSPACE-dew & 0.370 (0.008) & 0.016 (0.001) & 0.420 (0.007) & 0.855 (0.003) \\
 & matNS-AND & 0.044 (0.002) & \textbf{0.000 (0.000)} & 0.167 (0.006) & 0.865 (0.003) \\
 & matNS-OR & 0.338 (0.009) & 0.004 (0.000) & 0.483 (0.007) & 0.864 (0.003) \\
 & GEMINI & \textbf{0.666 (0.009)} & 0.062 (0.003) & 0.449 (0.006) & 0.886 (0.002) \\
 & G-FDR-0.10 & 0.006 (0.001) & \textbf{0.000 (0.000)} & 0.033 (0.004) & 0.641 (0.003) \\
\bottomrule
\end{tabular}
}
\end{minipage}%
\hfill
\begin{minipage}{0.49\textwidth}
\centering
\resizebox{\textwidth}{!}{%
\begin{tabular}{clcccc}
\toprule
$n$ & Method & TPR & FPR & MCC & AUC \\
\midrule
\multirow{14}{*}{40} & \multicolumn{5}{l}{\textit{$p=50$, $q=20$}} \\
 & matSPACE-sw & 0.922 (0.004) & 0.086 (0.003) & 0.670 (0.006) & \textbf{0.973 (0.001)} \\
 & matSPACE-dew & 0.763 (0.010) & 0.064 (0.004) & 0.643 (0.008) & 0.948 (0.002) \\
 & matNS-AND & 0.088 (0.005) & \textbf{0.000 (0.000)} & 0.220 (0.010) & 0.948 (0.002) \\
 & matNS-OR & 0.691 (0.010) & 0.005 (0.000) & \textbf{0.784 (0.007)} & 0.948 (0.002) \\
 & GEMINI & \textbf{0.953 (0.003)} & 0.192 (0.006) & 0.527 (0.007) & 0.964 (0.001) \\
 & G-FDR-0.10 & 0.282 (0.010) & \textbf{0.000 (0.000)} & 0.483 (0.010) & 0.877 (0.004) \\
\cmidrule(lr){2-6}
 & \multicolumn{5}{l}{\textit{$p=50$, $q=50$}} \\
 & matSPACE-sw & 0.882 (0.004) & 0.039 (0.001) & 0.636 (0.005) & \textbf{0.973 (0.001)} \\
 & matSPACE-dew & 0.575 (0.010) & 0.018 (0.001) & 0.572 (0.006) & 0.939 (0.002) \\
 & matNS-AND & 0.097 (0.004) & \textbf{0.000 (0.000)} & 0.281 (0.006) & 0.948 (0.001) \\
 & matNS-OR & 0.643 (0.008) & 0.005 (0.000) & \textbf{0.724 (0.005)} & 0.946 (0.001) \\
 & GEMINI & \textbf{0.917 (0.004)} & 0.092 (0.003) & 0.518 (0.008) & 0.964 (0.001) \\
 & G-FDR-0.10 & 0.112 (0.005) & \textbf{0.000 (0.000)} & 0.300 (0.008) & 0.841 (0.003) \\
\midrule
\midrule
\multirow{28}{*}{70} & \multicolumn{5}{l}{\textit{$p=20$, $q=20$}} \\
 & matSPACE-sw & 0.855 (0.007) & 0.087 (0.003) & 0.629 (0.006) & \textbf{0.949 (0.002)} \\
 & matSPACE-dew & 0.666 (0.011) & 0.059 (0.003) & 0.585 (0.008) & 0.919 (0.003) \\
 & matNS-AND & 0.061 (0.004) & \textbf{0.000 (0.000)} & 0.166 (0.010) & 0.919 (0.003) \\
 & matNS-OR & 0.573 (0.013) & 0.005 (0.000) & \textbf{0.694 (0.009)} & 0.920 (0.003) \\
 & GEMINI & \textbf{0.906 (0.005)} & 0.172 (0.005) & 0.519 (0.006) & 0.937 (0.002) \\
 & G-FDR-0.10 & 0.141 (0.007) & \textbf{0.000 (0.000)} & 0.302 (0.012) & 0.803 (0.005) \\
\cmidrule(lr){2-6}
 & \multicolumn{5}{l}{\textit{$p=20$, $q=50$}} \\
 & matSPACE-sw & 0.807 (0.005) & 0.039 (0.001) & 0.596 (0.005) & \textbf{0.953 (0.001)} \\
 & matSPACE-dew & 0.499 (0.010) & 0.018 (0.001) & 0.521 (0.006) & 0.912 (0.002) \\
 & matNS-AND & 0.071 (0.003) & \textbf{0.000 (0.000)} & 0.233 (0.006) & 0.922 (0.002) \\
 & matNS-OR & 0.529 (0.009) & 0.005 (0.000) & \textbf{0.638 (0.006)} & 0.921 (0.002) \\
 & GEMINI & \textbf{0.850 (0.005)} & 0.080 (0.004) & 0.514 (0.007) & 0.942 (0.001) \\
 & G-FDR-0.10 & 0.044 (0.003) & \textbf{0.000 (0.000)} & 0.165 (0.007) & 0.769 (0.003) \\
\cmidrule(lr){2-6}
 & \multicolumn{5}{l}{\textit{$p=50$, $q=20$}} \\
 & matSPACE-sw & 0.976 (0.002) & 0.084 (0.003) & 0.708 (0.006) & \textbf{0.992 (0.001)} \\
 & matSPACE-dew & 0.877 (0.009) & 0.114 (0.009) & 0.654 (0.010) & 0.977 (0.001) \\
 & matNS-AND & 0.147 (0.007) & \textbf{0.000 (0.000)} & 0.324 (0.010) & 0.977 (0.001) \\
 & matNS-OR & 0.840 (0.008) & 0.004 (0.000) & \textbf{0.886 (0.004)} & 0.976 (0.001) \\
 & GEMINI & \textbf{0.988 (0.002)} & 0.204 (0.005) & 0.527 (0.007) & 0.987 (0.001) \\
 & G-FDR-0.10 & 0.592 (0.009) & \textbf{0.000 (0.000)} & 0.748 (0.006) & 0.953 (0.002) \\
\cmidrule(lr){2-6}
 & \multicolumn{5}{l}{\textit{$p=50$, $q=50$}} \\
 & matSPACE-sw & 0.962 (0.002) & 0.036 (0.001) & 0.698 (0.006) & \textbf{0.992 (0.000)} \\
 & matSPACE-dew & 0.724 (0.011) & 0.026 (0.002) & 0.639 (0.007) & 0.972 (0.001) \\
 & matNS-AND & 0.152 (0.005) & \textbf{0.000 (0.000)} & 0.366 (0.006) & 0.977 (0.001) \\
 & matNS-OR & 0.815 (0.005) & 0.004 (0.000) & \textbf{0.839 (0.003)} & 0.976 (0.001) \\
 & GEMINI & \textbf{0.979 (0.001)} & 0.101 (0.004) & 0.529 (0.008) & 0.987 (0.000) \\
 & G-FDR-0.10 & 0.371 (0.007) & \textbf{0.000 (0.000)} & 0.594 (0.006) & 0.937 (0.002) \\
\bottomrule
\end{tabular}
}
\end{minipage}
\end{sidewaystable}

\begin{sidewaystable}
\centering
\spacingset{1}
\caption{Simulation results based on 300 data replications, where $U^{-1}$ is fixed at the band structure and $V^{-1}$ is fixed at the random structure, evaluated on the estimated $\widehat{U^{-1}}$.
Each entry reports an average of each metric with the standard error in parentheses. Boldface indicates the best-performing method.
}
\label{tab:random_U}
\renewcommand{\arraystretch}{0.85}
\tiny
\setlength{\tabcolsep}{4pt}
\begin{minipage}{0.49\textwidth}
\centering
\resizebox{\textwidth}{!}{%
\begin{tabular}{clcccc}
\toprule
$n$ & Method & TPR & FPR & MCC & AUC \\
\midrule
\multirow{28}{*}{10} & \multicolumn{5}{l}{\textit{$p=20$, $q=20$}} \\
 & matSPACE-sw & \textbf{0.400 (0.009)} & 0.125 (0.005) & \textbf{0.295 (0.005)} & \textbf{0.715 (0.003)} \\
 & matSPACE-dew & 0.252 (0.007) & 0.062 (0.003) & 0.262 (0.005) & 0.690 (0.003) \\
 & matNS-AND & 0.025 (0.002) & 0.001 (0.000) & 0.093 (0.006) & 0.703 (0.003) \\
 & matNS-OR & 0.084 (0.004) & 0.012 (0.001) & 0.170 (0.006) & 0.693 (0.003) \\
 & GEMINI & 0.329 (0.010) & 0.098 (0.005) & 0.274 (0.006) & 0.698 (0.003) \\
 & G-FDR-0.10 & 0.002 (0.001) & \textbf{0.000 (0.000)} & 0.010 (0.002) & 0.542 (0.002) \\
\cmidrule(lr){2-6}
 & \multicolumn{5}{l}{\textit{$p=20$, $q=50$}} \\
 & matSPACE-sw & 0.621 (0.008) & 0.145 (0.005) & \textbf{0.455 (0.004)} & \textbf{0.818 (0.002)} \\
 & matSPACE-dew & 0.435 (0.010) & 0.088 (0.004) & 0.397 (0.004) & 0.784 (0.002) \\
 & matNS-AND & 0.065 (0.003) & 0.001 (0.000) & 0.192 (0.007) & 0.800 (0.002) \\
 & matNS-OR & 0.178 (0.005) & 0.013 (0.001) & 0.312 (0.006) & 0.788 (0.002) \\
 & GEMINI & \textbf{0.665 (0.008)} & 0.214 (0.006) & 0.402 (0.004) & 0.796 (0.002) \\
 & G-FDR-0.10 & 0.035 (0.002) & \textbf{0.000 (0.000)} & 0.122 (0.007) & 0.661 (0.003) \\
\cmidrule(lr){2-6}
 & \multicolumn{5}{l}{\textit{$p=50$, $q=20$}} \\
 & matSPACE-sw & \textbf{0.282 (0.006)} & 0.058 (0.003) & \textbf{0.243 (0.003)} & \textbf{0.703 (0.002)} \\
 & matSPACE-dew & 0.158 (0.004) & 0.025 (0.001) & 0.203 (0.003) & 0.676 (0.002) \\
 & matNS-AND & 0.016 (0.001) & \textbf{0.000 (0.000)} & 0.096 (0.004) & 0.691 (0.002) \\
 & matNS-OR & 0.068 (0.002) & 0.009 (0.000) & 0.147 (0.003) & 0.679 (0.002) \\
 & GEMINI & 0.208 (0.007) & 0.035 (0.002) & 0.236 (0.003) & 0.689 (0.002) \\
 & G-FDR-0.10 & 0.000 (0.000) & \textbf{0.000 (0.000)} & 0.000 (0.000) & 0.531 (0.001) \\
\cmidrule(lr){2-6}
 & \multicolumn{5}{l}{\textit{$p=50$, $q=50$}} \\
 & matSPACE-sw & \textbf{0.495 (0.005)} & 0.062 (0.002) & \textbf{0.405 (0.003)} & \textbf{0.807 (0.002)} \\
 & matSPACE-dew & 0.277 (0.006) & 0.027 (0.001) & 0.329 (0.003) & 0.769 (0.001) \\
 & matNS-AND & 0.052 (0.002) & \textbf{0.000 (0.000)} & 0.200 (0.004) & 0.787 (0.001) \\
 & matNS-OR & 0.142 (0.003) & 0.007 (0.000) & 0.280 (0.003) & 0.775 (0.001) \\
 & GEMINI & 0.491 (0.007) & 0.081 (0.004) & 0.375 (0.004) & 0.786 (0.001) \\
 & G-FDR-0.10 & 0.007 (0.001) & \textbf{0.000 (0.000)} & 0.053 (0.004) & 0.621 (0.002) \\
\midrule
\midrule
\multirow{14}{*}{40} & \multicolumn{5}{l}{\textit{$p=20$, $q=20$}} \\
 & matSPACE-sw & 0.727 (0.006) & 0.147 (0.004) & \textbf{0.530 (0.004)} & \textbf{0.865 (0.002)} \\
 & matSPACE-dew & 0.543 (0.011) & 0.104 (0.005) & 0.470 (0.005) & 0.829 (0.002) \\
 & matNS-AND & 0.105 (0.004) & 0.001 (0.000) & 0.267 (0.007) & 0.845 (0.002) \\
 & matNS-OR & 0.248 (0.006) & 0.012 (0.001) & 0.399 (0.006) & 0.834 (0.002) \\
 & GEMINI & \textbf{0.787 (0.006)} & 0.265 (0.006) & 0.437 (0.005) & 0.843 (0.002) \\
 & G-FDR-0.10 & 0.121 (0.005) & \textbf{0.000 (0.000)} & 0.295 (0.007) & 0.749 (0.003) \\
\cmidrule(lr){2-6}
 & \multicolumn{5}{l}{\textit{$p=20$, $q=50$}} \\
 & matSPACE-sw & 0.881 (0.004) & 0.141 (0.004) & \textbf{0.652 (0.004)} & \textbf{0.944 (0.001)} \\
 & matSPACE-dew & 0.808 (0.009) & 0.199 (0.010) & 0.561 (0.007) & 0.913 (0.001) \\
 & matNS-AND & 0.250 (0.006) & 0.001 (0.000) & 0.448 (0.005) & 0.926 (0.001) \\
 & matNS-OR & 0.452 (0.007) & 0.013 (0.001) & 0.585 (0.005) & 0.917 (0.001) \\
 & GEMINI & \textbf{0.940 (0.003)} & 0.358 (0.006) & 0.462 (0.005) & 0.927 (0.001) \\
 & G-FDR-0.10 & 0.424 (0.006) & \textbf{0.000 (0.000)} & 0.609 (0.005) & 0.883 (0.002) \\
\bottomrule
\end{tabular}
}
\end{minipage}%
\hfill
\begin{minipage}{0.49\textwidth}
\centering
\resizebox{\textwidth}{!}{%
\begin{tabular}{clcccc}
\toprule
$n$ & Method & TPR & FPR & MCC & AUC \\
\midrule
\multirow{14}{*}{40} & \multicolumn{5}{l}{\textit{$p=50$, $q=20$}} \\
 & matSPACE-sw & 0.621 (0.005) & 0.067 (0.002) & \textbf{0.486 (0.003)} & \textbf{0.862 (0.001)} \\
 & matSPACE-dew & 0.374 (0.007) & 0.032 (0.002) & 0.405 (0.003) & 0.820 (0.001) \\
 & matNS-AND & 0.092 (0.002) & 0.001 (0.000) & 0.271 (0.004) & 0.840 (0.001) \\
 & matNS-OR & 0.218 (0.003) & 0.008 (0.000) & 0.366 (0.003) & 0.827 (0.001) \\
 & GEMINI & \textbf{0.656 (0.006)} & 0.123 (0.004) & 0.413 (0.004) & 0.839 (0.001) \\
 & G-FDR-0.10 & 0.038 (0.002) & \textbf{0.000 (0.000)} & 0.170 (0.004) & 0.710 (0.002) \\
\cmidrule(lr){2-6}
 & \multicolumn{5}{l}{\textit{$p=50$, $q=50$}} \\
 & matSPACE-sw & 0.825 (0.003) & 0.068 (0.002) & \textbf{0.622 (0.004)} & \textbf{0.942 (0.001)} \\
 & matSPACE-dew & 0.566 (0.008) & 0.037 (0.002) & 0.554 (0.003) & 0.903 (0.001) \\
 & matNS-AND & 0.234 (0.004) & 0.001 (0.000) & 0.455 (0.004) & 0.921 (0.001) \\
 & matNS-OR & 0.416 (0.004) & 0.009 (0.000) & 0.556 (0.003) & 0.912 (0.001) \\
 & GEMINI & \textbf{0.875 (0.003)} & 0.176 (0.005) & 0.467 (0.005) & 0.924 (0.001) \\
 & G-FDR-0.10 & 0.244 (0.003) & \textbf{0.000 (0.000)} & 0.476 (0.003) & 0.854 (0.001) \\
\midrule
\midrule
\multirow{28}{*}{70} & \multicolumn{5}{l}{\textit{$p=20$, $q=20$}} \\
 & matSPACE-sw & 0.832 (0.004) & 0.146 (0.004) & \textbf{0.611 (0.004)} & \textbf{0.918 (0.002)} \\
 & matSPACE-dew & 0.725 (0.010) & 0.174 (0.008) & 0.524 (0.006) & 0.883 (0.002) \\
 & matNS-AND & 0.189 (0.005) & 0.001 (0.000) & 0.381 (0.006) & 0.897 (0.002) \\
 & matNS-OR & 0.374 (0.006) & 0.013 (0.001) & 0.517 (0.005) & 0.888 (0.002) \\
 & GEMINI & \textbf{0.895 (0.004)} & 0.318 (0.006) & 0.465 (0.004) & 0.898 (0.002) \\
 & G-FDR-0.10 & 0.292 (0.006) & \textbf{0.000 (0.000)} & 0.494 (0.006) & 0.839 (0.002) \\
\cmidrule(lr){2-6}
 & \multicolumn{5}{l}{\textit{$p=20$, $q=50$}} \\
 & matSPACE-sw & 0.941 (0.003) & 0.122 (0.004) & 0.725 (0.005) & \textbf{0.974 (0.001)} \\
 & matSPACE-dew & 0.911 (0.006) & 0.305 (0.015) & 0.544 (0.011) & 0.949 (0.001) \\
 & matNS-AND & 0.377 (0.006) & 0.001 (0.000) & 0.565 (0.005) & 0.959 (0.001) \\
 & matNS-OR & 0.587 (0.006) & 0.010 (0.001) & 0.697 (0.004) & 0.953 (0.001) \\
 & GEMINI & \textbf{0.978 (0.001)} & 0.389 (0.006) & 0.465 (0.005) & 0.960 (0.001) \\
 & G-FDR-0.10 & 0.638 (0.005) & \textbf{0.000 (0.000)} & \textbf{0.766 (0.003)} & 0.936 (0.002) \\
\cmidrule(lr){2-6}
 & \multicolumn{5}{l}{\textit{$p=50$, $q=20$}} \\
 & matSPACE-sw & 0.753 (0.004) & 0.068 (0.002) & \textbf{0.574 (0.003)} & \textbf{0.915 (0.001)} \\
 & matSPACE-dew & 0.483 (0.008) & 0.033 (0.002) & 0.499 (0.003) & 0.873 (0.001) \\
 & matNS-AND & 0.172 (0.003) & 0.001 (0.000) & 0.386 (0.004) & 0.893 (0.001) \\
 & matNS-OR & 0.332 (0.004) & 0.008 (0.000) & 0.484 (0.003) & 0.881 (0.001) \\
 & GEMINI & \textbf{0.809 (0.005)} & 0.162 (0.005) & 0.450 (0.005) & 0.894 (0.001) \\
 & G-FDR-0.10 & 0.144 (0.003) & \textbf{0.000 (0.000)} & 0.361 (0.004) & 0.802 (0.002) \\
\cmidrule(lr){2-6}
 & \multicolumn{5}{l}{\textit{$p=50$, $q=50$}} \\
 & matSPACE-sw & 0.905 (0.002) & 0.059 (0.002) & \textbf{0.697 (0.004)} & \textbf{0.972 (0.001)} \\
 & matSPACE-dew & 0.659 (0.007) & 0.029 (0.002) & 0.651 (0.003) & 0.941 (0.001) \\
 & matNS-AND & 0.357 (0.004) & 0.001 (0.000) & 0.573 (0.003) & 0.955 (0.001) \\
 & matNS-OR & 0.559 (0.004) & 0.008 (0.000) & 0.676 (0.003) & 0.948 (0.001) \\
 & GEMINI & \textbf{0.948 (0.002)} & 0.211 (0.006) & 0.469 (0.006) & 0.958 (0.001) \\
 & G-FDR-0.10 & 0.450 (0.004) & \textbf{0.000 (0.000)} & 0.654 (0.003) & 0.921 (0.001) \\
\bottomrule
\end{tabular}
}
\end{minipage}
\end{sidewaystable}

\begin{sidewaystable}
\centering
\spacingset{1}
\caption{Simulation results based on 300 data replications, where $U^{-1}$ is fixed at the band structure and $V^{-1}$ is fixed at the random structure, evaluated on the estimated $\widehat{V^{-1}}$.
Each entry reports an average of each metric with the standard error in parentheses. Boldface indicates the best-performing method.
}
\label{tab:random_V}
\renewcommand{\arraystretch}{0.85}
\tiny
\setlength{\tabcolsep}{4pt}
\begin{minipage}{0.49\textwidth}
\centering
\resizebox{\textwidth}{!}{%
\begin{tabular}{clcccc}
\toprule
$n$ & Method & TPR & FPR & MCC & AUC \\
\midrule
\multirow{28}{*}{10} & \multicolumn{5}{l}{\textit{$p=20$, $q=20$}} \\
 & matSPACE-sw & \textbf{0.359 (0.013)} & 0.106 (0.006) & \textbf{0.263 (0.007)} & \textbf{0.727 (0.003)} \\
 & matSPACE-dew & 0.236 (0.010) & 0.060 (0.004) & 0.223 (0.007) & 0.682 (0.003) \\
 & matNS-AND & 0.014 (0.002) & \textbf{0.000 (0.000)} & 0.051 (0.005) & 0.693 (0.003) \\
 & matNS-OR & 0.077 (0.005) & 0.009 (0.001) & 0.139 (0.008) & 0.687 (0.003) \\
 & GEMINI & 0.261 (0.013) & 0.063 (0.004) & 0.230 (0.009) & 0.717 (0.003) \\
 & G-FDR-0.10 & 0.000 (0.000) & \textbf{0.000 (0.000)} & 0.001 (0.001) & 0.514 (0.001) \\
\cmidrule(lr){2-6}
 & \multicolumn{5}{l}{\textit{$p=20$, $q=50$}} \\
 & matSPACE-sw & \textbf{0.191 (0.008)} & 0.041 (0.002) & \textbf{0.166 (0.005)} & \textbf{0.685 (0.002)} \\
 & matSPACE-dew & 0.115 (0.005) & 0.022 (0.001) & 0.137 (0.004) & 0.638 (0.002) \\
 & matNS-AND & 0.005 (0.001) & \textbf{0.000 (0.000)} & 0.031 (0.003) & 0.651 (0.002) \\
 & matNS-OR & 0.039 (0.002) & 0.005 (0.000) & 0.091 (0.004) & 0.645 (0.002) \\
 & GEMINI & 0.081 (0.006) & 0.012 (0.001) & 0.113 (0.006) & 0.680 (0.002) \\
 & G-FDR-0.10 & 0.000 (0.000) & \textbf{0.000 (0.000)} & 0.000 (0.000) & 0.504 (0.000) \\
\cmidrule(lr){2-6}
 & \multicolumn{5}{l}{\textit{$p=50$, $q=20$}} \\
 & matSPACE-sw & 0.662 (0.010) & 0.129 (0.005) & \textbf{0.480 (0.006)} & \textbf{0.853 (0.003)} \\
 & matSPACE-dew & 0.434 (0.012) & 0.078 (0.004) & 0.383 (0.006) & 0.794 (0.003) \\
 & matNS-AND & 0.045 (0.003) & \textbf{0.000 (0.000)} & 0.138 (0.008) & 0.810 (0.003) \\
 & matNS-OR & 0.178 (0.008) & 0.011 (0.001) & 0.291 (0.010) & 0.802 (0.003) \\
 & GEMINI & \textbf{0.691 (0.011)} & 0.169 (0.006) & 0.449 (0.006) & 0.836 (0.003) \\
 & G-FDR-0.10 & 0.007 (0.001) & \textbf{0.000 (0.000)} & 0.030 (0.004) & 0.615 (0.004) \\
\cmidrule(lr){2-6}
 & \multicolumn{5}{l}{\textit{$p=50$, $q=50$}} \\
 & matSPACE-sw & 0.408 (0.009) & 0.042 (0.002) & \textbf{0.363 (0.005)} & \textbf{0.809 (0.002)} \\
 & matSPACE-dew & 0.227 (0.006) & 0.024 (0.001) & 0.263 (0.004) & 0.745 (0.003) \\
 & matNS-AND & 0.018 (0.001) & \textbf{0.000 (0.000)} & 0.089 (0.005) & 0.764 (0.002) \\
 & matNS-OR & 0.090 (0.004) & 0.004 (0.000) & 0.201 (0.006) & 0.756 (0.003) \\
 & GEMINI & \textbf{0.415 (0.010)} & 0.050 (0.002) & 0.354 (0.005) & 0.797 (0.002) \\
 & G-FDR-0.10 & 0.000 (0.000) & \textbf{0.000 (0.000)} & 0.000 (0.000) & 0.537 (0.002) \\
\midrule
\midrule
\multirow{14}{*}{40} & \multicolumn{5}{l}{\textit{$p=20$, $q=20$}} \\
 & matSPACE-sw & 0.792 (0.008) & 0.134 (0.004) & \textbf{0.561 (0.005)} & \textbf{0.909 (0.003)} \\
 & matSPACE-dew & 0.555 (0.012) & 0.088 (0.004) & 0.472 (0.006) & 0.851 (0.003) \\
 & matNS-AND & 0.079 (0.004) & 0.001 (0.000) & 0.211 (0.009) & 0.867 (0.003) \\
 & matNS-OR & 0.263 (0.009) & 0.010 (0.001) & 0.401 (0.009) & 0.860 (0.003) \\
 & GEMINI & \textbf{0.839 (0.008)} & 0.209 (0.006) & 0.502 (0.005) & 0.894 (0.002) \\
 & G-FDR-0.10 & 0.050 (0.004) & \textbf{0.000 (0.000)} & 0.145 (0.009) & 0.735 (0.005) \\
\cmidrule(lr){2-6}
 & \multicolumn{5}{l}{\textit{$p=20$, $q=50$}} \\
 & matSPACE-sw & 0.597 (0.007) & 0.050 (0.002) & \textbf{0.483 (0.004)} & \textbf{0.877 (0.002)} \\
 & matSPACE-dew & 0.315 (0.007) & 0.025 (0.001) & 0.347 (0.004) & 0.807 (0.002) \\
 & matNS-AND & 0.036 (0.002) & \textbf{0.000 (0.000)} & 0.154 (0.005) & 0.829 (0.002) \\
 & matNS-OR & 0.148 (0.004) & 0.004 (0.000) & 0.294 (0.006) & 0.821 (0.002) \\
 & GEMINI & \textbf{0.642 (0.009)} & 0.081 (0.003) & 0.447 (0.005) & 0.862 (0.002) \\
 & G-FDR-0.10 & 0.002 (0.000) & \textbf{0.000 (0.000)} & 0.011 (0.002) & 0.619 (0.003) \\
\bottomrule
\end{tabular}
}
\end{minipage}%
\hfill
\begin{minipage}{0.49\textwidth}
\centering
\resizebox{\textwidth}{!}{%
\begin{tabular}{clcccc}
\toprule
$n$ & Method & TPR & FPR & MCC & AUC \\
\midrule
\multirow{14}{*}{40} & \multicolumn{5}{l}{\textit{$p=50$, $q=20$}} \\
 & matSPACE-sw & 0.955 (0.003) & 0.129 (0.004) & \textbf{0.681 (0.005)} & \textbf{0.979 (0.001)} \\
 & matSPACE-dew & 0.802 (0.010) & 0.121 (0.006) & 0.613 (0.006) & 0.937 (0.002) \\
 & matNS-AND & 0.223 (0.008) & 0.001 (0.000) & 0.415 (0.009) & 0.950 (0.002) \\
 & matNS-OR & 0.503 (0.010) & 0.011 (0.001) & 0.624 (0.007) & 0.944 (0.002) \\
 & GEMINI & \textbf{0.979 (0.002)} & 0.284 (0.007) & 0.516 (0.006) & 0.969 (0.001) \\
 & G-FDR-0.10 & 0.401 (0.011) & \textbf{0.000 (0.000)} & 0.581 (0.010) & 0.920 (0.003) \\
\cmidrule(lr){2-6}
 & \multicolumn{5}{l}{\textit{$p=50$, $q=50$}} \\
 & matSPACE-sw & 0.875 (0.004) & 0.059 (0.002) & \textbf{0.634 (0.004)} & \textbf{0.965 (0.001)} \\
 & matSPACE-dew & 0.528 (0.008) & 0.030 (0.001) & 0.511 (0.004) & 0.905 (0.002) \\
 & matNS-AND & 0.132 (0.004) & \textbf{0.000 (0.000)} & 0.333 (0.006) & 0.926 (0.001) \\
 & matNS-OR & 0.348 (0.006) & 0.006 (0.000) & 0.505 (0.005) & 0.920 (0.002) \\
 & GEMINI & \textbf{0.919 (0.004)} & 0.129 (0.004) & 0.518 (0.006) & 0.954 (0.001) \\
 & G-FDR-0.10 & 0.077 (0.004) & \textbf{0.000 (0.000)} & 0.238 (0.008) & 0.833 (0.003) \\
\midrule
\midrule
\multirow{28}{*}{70} & \multicolumn{5}{l}{\textit{$p=20$, $q=20$}} \\
 & matSPACE-sw & 0.917 (0.005) & 0.138 (0.004) & \textbf{0.642 (0.005)} & \textbf{0.960 (0.002)} \\
 & matSPACE-dew & 0.728 (0.012) & 0.116 (0.005) & 0.558 (0.006) & 0.909 (0.002) \\
 & matNS-AND & 0.160 (0.007) & 0.001 (0.000) & 0.333 (0.010) & 0.924 (0.002) \\
 & matNS-OR & 0.414 (0.011) & 0.012 (0.001) & 0.539 (0.009) & 0.917 (0.002) \\
 & GEMINI & \textbf{0.948 (0.004)} & 0.272 (0.007) & 0.508 (0.006) & 0.947 (0.002) \\
 & G-FDR-0.10 & 0.220 (0.009) & \textbf{0.000 (0.000)} & 0.397 (0.012) & 0.863 (0.004) \\
\cmidrule(lr){2-6}
 & \multicolumn{5}{l}{\textit{$p=20$, $q=50$}} \\
 & matSPACE-sw & 0.788 (0.005) & 0.056 (0.002) & \textbf{0.591 (0.004)} & \textbf{0.939 (0.001)} \\
 & matSPACE-dew & 0.453 (0.008) & 0.029 (0.001) & 0.450 (0.004) & 0.872 (0.002) \\
 & matNS-AND & 0.085 (0.003) & \textbf{0.000 (0.000)} & 0.261 (0.006) & 0.894 (0.002) \\
 & matNS-OR & 0.265 (0.006) & 0.005 (0.000) & 0.427 (0.006) & 0.886 (0.002) \\
 & GEMINI & \textbf{0.835 (0.006)} & 0.110 (0.004) & 0.511 (0.006) & 0.926 (0.001) \\
 & G-FDR-0.10 & 0.020 (0.002) & \textbf{0.000 (0.000)} & 0.095 (0.006) & 0.755 (0.003) \\
\cmidrule(lr){2-6}
 & \multicolumn{5}{l}{\textit{$p=50$, $q=20$}} \\
 & matSPACE-sw & 0.990 (0.001) & 0.114 (0.003) & 0.729 (0.005) & \textbf{0.995 (0.000)} \\
 & matSPACE-dew & 0.899 (0.007) & 0.139 (0.007) & 0.662 (0.008) & 0.970 (0.001) \\
 & matNS-AND & 0.354 (0.009) & 0.001 (0.000) & 0.546 (0.008) & 0.978 (0.001) \\
 & matNS-OR & 0.664 (0.009) & 0.010 (0.001) & 0.750 (0.006) & 0.974 (0.001) \\
 & GEMINI & \textbf{0.997 (0.001)} & 0.295 (0.007) & 0.517 (0.006) & 0.990 (0.001) \\
 & G-FDR-0.10 & 0.721 (0.009) & \textbf{0.000 (0.000)} & \textbf{0.826 (0.006)} & 0.973 (0.001) \\
\cmidrule(lr){2-6}
 & \multicolumn{5}{l}{\textit{$p=50$, $q=50$}} \\
 & matSPACE-sw & 0.961 (0.002) & 0.055 (0.001) & \textbf{0.698 (0.005)} & \textbf{0.990 (0.000)} \\
 & matSPACE-dew & 0.672 (0.008) & 0.032 (0.001) & 0.608 (0.004) & 0.948 (0.001) \\
 & matNS-AND & 0.234 (0.005) & \textbf{0.000 (0.000)} & 0.459 (0.006) & 0.964 (0.001) \\
 & matNS-OR & 0.513 (0.007) & 0.006 (0.000) & 0.639 (0.005) & 0.960 (0.001) \\
 & GEMINI & \textbf{0.980 (0.001)} & 0.138 (0.004) & 0.539 (0.007) & 0.983 (0.001) \\
 & G-FDR-0.10 & 0.319 (0.008) & \textbf{0.000 (0.000)} & 0.540 (0.007) & 0.938 (0.002) \\
\bottomrule
\end{tabular}
}
\end{minipage}
\end{sidewaystable}

\end{document}